\documentclass[11pt,twoside]{amsart}

\usepackage{mathrsfs}
\usepackage{amssymb}
\usepackage{amsmath}
\usepackage{amsthm}
\usepackage{tikz}
\usepackage{amsfonts}
\usepackage{amsmath,amssymb,mathtools,booktabs,array,longtable,graphicx,pdflscape}
\usepackage[T1]{fontenc}
\usepackage{lmodern}
\usepackage[colorlinks=true,linkcolor=blue,citecolor=blue,urlcolor=blue]{hyperref}
\usepackage{color}

\usepackage[all]{xy}

\allowdisplaybreaks

\def\bint{{\ifinner\rlap{\bf\kern.30em--}
\int\else\rlap{\bf\kern.35em--}\int\fi}\ignorespaces}

\def\sbint{{\ifinner\rlap{\bf\kern.32em--}
\hspace{0.078cm}\int\else\rlap{\bf\kern.45em--}\int\fi}\ignorespaces}

\def\C{\mathbb C}
\def\F{\mathbb F}

\def\Cl{\operatorname{Cliff}}
\def\Tr{\operatorname{Tr}}
\def\End{\operatorname{End}}
\def\Sym{\operatorname{Sym}}
\def\sgn{\operatorname{sgn}}
\def\HH{\mathcal H}
\def\KK{\mathcal K}

\def\Dtr{D_{\mathrm{tr}}}

\newtheorem{theorem}{Theorem}[section]
\newtheorem{lemma}[theorem]{Lemma}
\newtheorem{corollary}[theorem]{Corollary}
\newtheorem{proposition}[theorem]{Proposition}

\newtheorem{example}[theorem]{Example}
\theoremstyle{definition}
\newtheorem{remark}[theorem]{Remark}
\newtheorem{definition}[theorem]{Definition}
\numberwithin{equation}{section}

\numberwithin{equation}{section}

\numberwithin{equation}{section}

\makeatletter
\newcommand{\charactertableheading}{%
  \@startsection{subsection}{2}{\z@}%
    {.5\linespacing\@plus.7\linespacing}{.5\linespacing}%
    {\normalfont\bfseries\centering}*}
\makeatother

\begin{document}

\title[Clifford Invariants and Three-Copy Quantum Tests]{Clifford Invariants and Three-Copy Quantum Tests}
\author[X. Fan]{Xingya Fan}
\address{College of Mathematics and System Sciences, Xinjiang University, Urumqi 830017, China}
\email{fanxingya@xju.edu.cn}
\author[Y. Luan]{Yongzhi Luan}
\address{School of Mathematics, Shandong University, Jinan 250100, China}
\email{luanyongzhi@email.sdu.edu.cn}



\date{\today}

\begin{abstract}
We give an explicit treatment of the third tensor power of the Clifford
representation on $n$ qutrits, for $n\geq2$. This representation has five
pairwise inequivalent irreducible constituents, with multiplicities
$1,1,1,1,2$. The decomposition reduces the Clifford twirl of a tensor cube to
four real parameters.  Two of them are spectral moments of the input.  The
other two are cubic invariants expressed as sums of cubes of the Weyl
characteristic function and the discrete Wigner function.  On trace-one
Hermitian matrices, the invariant polynomials of degree at most three
are spanned by the constant function and four algebraically independent
functions.  These functions do not separate Clifford orbits.  For pure
inputs a single parameter suffices. We recover its range and the
characterization of its maximum by stabilizer states. We bound the
stabilizer fidelity, determine full separability across the three copies
in the symmetric family, and recover the criterion for a complex
projective $3$-design.
We classify the jointly invariant quantum Boolean functions.  For
stabilizer testing and for testing unknown observables we give explicit
bounds on copies and queries. We also give character formulas for the
finite projective Clifford groups and record the complete character
table for two qutrits.
\end{abstract}
\maketitle

\tableofcontents
\section{Introduction}
\label{sec:intro}

For $n$ qudits of prime local dimension $d$, the Clifford group
$\Cl(n,d)$ is the unitary normalizer of the Heisenberg--Weyl group.
Its tensor powers describe the statistics accessible to measurements
that are invariant under a common Clifford transformation of the input
copies. In this paper, we study the third tensor power of the Clifford
representation on $n$ qutrits, with $d=3$ and $n\geq2$, and its
applications to quantum information processing.

The background for this study comes from three related directions.

(i) The Clifford group plays an important role in quantum error
correction \cite{gottesman97,cs96,s96b,s96a}. Clifford circuits with
stabilizer inputs and Pauli measurements admit efficient classical
simulation by the Gottesman--Knill theorem \cite{gottesman98}.
Aaronson and Gottesman developed efficient tableau algorithms for
qubit stabilizer circuits \cite{aaronson2004improved}. Universal quantum
computation requires resources outside this class, as illustrated by
magic-state distillation \cite{bravyi05}. In odd prime local dimension,
discrete Wigner functions provide a phase-space description of the
stabilizer formalism \cite{wootters87,gross06}. The finite-dimensional
Hudson theorem characterizes pure stabilizer states as precisely those
pure states with nonnegative discrete Wigner functions \cite{gross06}.
Wigner negativity is necessary for a computational speed-up in the
stabilizer computation model considered in \cite{veitch12}.

(ii) Random unitary methods underlie fidelity estimation and randomized
benchmarking \cite{dce09,emerson05,magesan11}. For prime $d$, the
Clifford group $\Cl(n,d)$ is a unitary $2$-design, so its second-moment
twirl agrees with the Haar twirl \cite{dce09,zhu17}. For qubits, the
Clifford group is also a unitary $3$-design \cite{webb16,zhu17}.
Related results on low tensor powers appear in \cite{gt05}, while
\cite{zkgg16} studies the failure of the qubit Clifford group to be a
unitary $4$-design. The $n$-qubit stabilizer states form a single
Clifford orbit and a complex projective $3$-design \cite{kueng15}.
For qutrits, the Clifford group is a unitary $2$-design but not a
unitary $3$-design \cite{zhu17}. Nevertheless, the third moment of any
pure-state Clifford orbit is determined by one real parameter
$q\in[0,3^{-n}]$, defined in \eqref{eq:pure}; see \cite{zmy24}.
This gives a starting point for studying third moments of qutrit
Clifford orbits.

(iii) Gross, Nezami and Walter \cite{gnw21} established a Schur--Weyl
duality for tensor powers of the Clifford group using finite geometry.
Its applications include stabilizer testing, de Finetti theorems for
stabilizer states, a robust finite-dimensional Hudson theorem, and
constructions of complex projective designs of arbitrary order.
Their Example~4.26 gives the five irreducible constituents and their
multiplicities in the three-copy qutrit case.

We give a self-contained derivation of this decomposition, with explicit
projectors and an intertwiner between the two equivalent copies; see
Theorem~\ref{rep:five-irreducibles}. The commutant calculation specializes
the symplectic-orbit method of \cite[Lemmas~1--2 and Appendix~A]{zhu17}
and \cite[Lemma~4.8]{gnw21} to ordered pairs of Weyl labels.
We use the resulting projectors to study mixed states. On the affine
space of trace-one Hermitian matrices, the Clifford invariant
polynomials of degree at most three have basis $1,t_2,t_3,u,v$, with
the four nonconstant functions defined in \eqref{eq:four}.
These four functions are algebraically independent, as shown by an
explicit Jacobian calculation in Appendix~\ref{app:independence}, but
they do not separate Clifford orbits. In the pure-state case, the range
of $q$, the family attaining that range and the criterion for a complex
projective $3$-design recover results of Zhu, Mao and Yi
\cite{zmy24}; precise comparisons are given in Section~\ref{sec:pure}.

We next investigate the information accessible through three-copy
measurements. Using the definition of quantum Boolean functions in
\cite{mo10}, we classify those that are invariant under both Clifford
transformations and permutations of the three copies. All of them can
be measured using a five-outcome sector measurement. Under the access
assumptions specified with each result, we derive explicit copy and
query complexity bounds. We also prove a lower bound for repeated
three-copy batches when every retained outcome effect is Clifford
invariant. The impossibility of a nontrivial perfectly complete
three-copy stabilizer test is a consequence of
\cite[Theorem~5.6 and Example~4.26]{gnw21}; we give a direct proof in
Proposition~\ref{prop:no-perfect}. The four-copy stabilizer test of
\cite[Theorem~3.11]{gnw21} uses a different measurement.

Our analysis requires care with scalar phases. The Clifford group
$\Cl(n,3)$ contains all scalar unitaries and is therefore infinite.
Its finite projective quotient is
$G_n=V_n\rtimes\operatorname{Sp}(2n,3)$ in \eqref{EG1}. We use the
finite lift $\widetilde G_n$ in \eqref{EG2} to define a tensor-cube
representation of $G_n$. The finite Weil representation provides
compatible symplectic implementers \cite[Theorem~1.2]{gh09}.
Since the scalar kernel of this lift has exponent three, its tensor
cube descends to $G_n$; see Propositions~\ref{rep:lift-properties} and
\ref{rep:cube-descent}. The structure of Clifford groups in odd
characteristic goes back to Bolt, Room and Wall \cite{bolt61}; the
extended Clifford group is treated in \cite{appleby05}.
Character tables for the projective Clifford groups on one, two and
three qubits were computed in \cite{peng25}. For two qutrits, the
ordinary character table is already available in the GAP Character
Table Library \cite{ctbllib25}. The five tensor-cube constituents
account for only five irreducible characters; for example, $G_2$ has
$58$ irreducible characters. Section~\ref{sec:finite-tables} gives the
little-group parametrization and character formulas for general $n$,
and Appendix~\ref{app:character-table-n2} records the complete
two-qutrit character table.

The paper is organized as follows.
Section~\ref{sec:foundations} proves the decomposition and the commutant
computation, including the explicit projectors, the comparison with the
ordinary three-copy Schur--Weyl duality, and the character formulas.
Section~\ref{sec:finite-tables} specifies the finite projective groups,
their orders, the parametrization by little groups and recursive
formulas for irreducible character degrees. The complete character table of $G_2$ is
printed in Appendix~\ref{app:character-table-n2}.
Sections~\ref{sec:arbitrary} to \ref{sec:pure} derive the consequences
for density matrices and for pure states, namely the twirl and the
four-parameter description of third moments, the linear and polynomial invariants together
with their algebraic independence, the range of $q$, the robust fidelity
bound, separability and the criterion for a complex projective
$3$-design. Section~\ref{sec:measurements} gives the physical
measurements and the statistically valid bounds, and
Section~\ref{sec:operators} treats unknown quantum Boolean functions.
The two appendices contain the algebraic independence calculation and
the character table of $G_2$.
\section{Representation-theoretic foundations}
\label{sec:foundations}
Unless stated otherwise, $n\geq2$,
$\HH=(\C^3)^{\otimes n}$ and $\KK=\HH^{\otimes3}$. Let $\End(\KK)$ be
the algebra of linear operators on $\KK$, and $\Tr$ the (unnormalized)
trace on $\End(\KK)$, so that $\Tr(I_{\KK})=\dim\KK$.
\subsection{Weyl operators, symplectic covariance, and scalar phases}
\label{rep:weyl-subsection}

Let
\(
\omega=e^{2\pi i/3}\) and
\(V_n=\mathbb F_3^n\oplus\mathbb F_3^n,\)
where $\mathbb F_3=\{0,1,2\}$ and $\mathbb F_3^n$ is the
$n$-dimensional vector space over $\mathbb F_3$. The computational basis
of the $n$-qutrit Hilbert space $\mathcal H$ is
\(
\{|x\rangle:x\in\mathbb F_3^n\}.
\)
For $a,b,x\in\mathbb F_3^n$, define
\[
X^b|x\rangle=|x+b\rangle,
\quad
Z^a|x\rangle=\omega^{a\cdot x}|x\rangle,
\]
where $a\cdot x$ is the standard bilinear form on $\mathbb F_3^n$.
For
\(u=(a,b)\) and \(v=(c,d)\) in $V_n$,
define the standard symplectic form by
\([u,v]=a\cdot d-b\cdot c.\)
We define the Weyl operator by
\(W_{(a,b)}=\omega^{-a\cdot b/2}Z^aX^b,\)
where $1/2$ is the inverse of $2$ in $\mathbb F_3$.
The relation
\begin{align}\label{E2.1}
X^bZ^c=\omega^{-b\cdot c}Z^cX^b
\end{align}
implies
\begin{align}\label{rep:weyl-identities}
W_uW_v=\omega^{[u,v]/2}W_{u+v},\quad
W_u^\dagger=W_{-u},
\quad
W_vW_uW_v^\dagger=\omega^{[v,u]}W_u,
\end{align}
where $\dagger$ denotes the adjoint. The Weyl operators are orthogonal with respect to the Hilbert--Schmidt
inner product:
\begin{align*}
\operatorname{Tr}(W_u)=3^n\delta_{u,0},
\quad
\operatorname{Tr}(W_u^\dagger W_v)=3^n\delta_{u,v},
\end{align*}
where $\delta_{u,v}$ denotes the Kronecker delta for $u,v\in V_n$. There are exactly $3^{2n}$ Weyl operators, and
\(\dim\operatorname{End}(\mathcal H)=3^{2n}.\)
Hence they form an orthogonal basis of
$\operatorname{End}(\mathcal H)$. Thus every operator $\mathcal{A}$ has the expansion
\[\mathcal{A}=3^{-n}\sum_{u\in V_n}\operatorname{Tr}(W_u^\dagger \mathcal{A})W_u.\]

Let
\(
\mathcal P_n=\{\omega^kW_u:k\in\mathbb F_3,\ u\in V_n\}
\)
be the Pauli group. Let $U(3^n)$ be the group of all unitary operators acting on $\mathcal{H}$. The Clifford group is the normalizer of
$\mathcal P_n$ defined by
\[\Cl(n,3)=\{\mathcal{C}\in U(3^n):\mathcal{C}\mathcal P_n\mathcal{C}^\dagger=\mathcal P_n\}.\]
Every scalar unitary operator $zI$, with $z\in U(1)$, belongs to
$\Cl(n,3)$.

Let $\operatorname{Sp}(2n,3)$ denote the group of linear transformations
preserving the symplectic form:
\[
[Mu,Mv]=[u,v]
\quad
\text{for all }u,v\in V_n.
\]
The following covariance statement is \cite[Lemma~2.1]{gnw21}, specialized to qutrits.
\begin{lemma}
\label{rep:covariance}
For every $M\in\operatorname{Sp}(2n,3)$ there exists a unitary operator
$U_M$ such that
\(U_MW_uU_M^\dagger=W_{Mu}\) for every $u\in V_n$.
Conversely, every Clifford unitary operator can be written as
\(\mathcal{C}=zW_vU_M,\)
where $z\in U(1)$, $v\in V_n$ and $M\in\operatorname{Sp}(2n,3)$.
\end{lemma}

The implementer $U_M$ is not uniquely determined. If $U_M$ implements
$M$, then so does $zU_M$ for every $z\in U(1)$. Therefore, an arbitrary
choice of implementers gives only a projective representation:
\[
U_MU_N=c(M,N)U_{MN},
\quad
|c(M,N)|=1.
\]
The function $c$ is a $2$-cocycle. Lemma \ref{rep:covariance} alone does not
imply that the cohomology class of this cocycle is trivial.
Nevertheless, the conjugation action on Weyl operators is unambiguous. If
\(\mathcal{C}=zW_vU_M,\)
then
\[
 \mathcal{C}W_u\mathcal{C}^\dagger=\omega^{[v,Mu]}W_{Mu}.
\]
The Weyl label transforms linearly by $M$; the translation parameter
$v$ appears in the phase. The label permutation determines $M$, and
nondegeneracy of the symplectic form then determines $v$ from the phases,
whereas the scalar $z$ is irrelevant.
It follows that there is a well-defined surjective homomorphism
\begin{align}\label{EG1}
\pi:\Cl(n,3)\to
G_n:=V_n\rtimes\operatorname{Sp}(2n,3),
\end{align}
given by
\(\pi(zW_vU_M)=(v,M).\) Here $G_n$ is the
projective qutrit Clifford group.
The multiplication rule in $G_n$ is
\((v,M)(w,N)=(v+Mw,MN).\)
Indeed, for arbitrary implementers,
\[
(W_vU_M)(W_wU_N)
=c(M,N)\omega^{[v,Mw]/2}W_{v+Mw}U_{MN},
\]
and the scalar phase disappears after applying $\pi$.
The kernel of $\pi$ consists precisely of scalar unitaries. Hence
\(
\Cl(n,3)/U(1)
\cong
V_n\rtimes\operatorname{Sp}(2n,3).
\)
The quotient group $G_n$ is finite, whereas $\Cl(n,3)$
contains the whole circle of scalar unitaries and is therefore infinite.
At this point, one must distinguish carefully between the quotient action
and a linear representation. The conjugation action on the Weyl basis
depends only on the image in $G_n$:
\[
(zW_vU_M)W_u(zW_vU_M)^\dagger=\omega^{[v,Mu]}W_{Mu}.
\]
Thus this conjugation action is a genuine action of the finite group $G_n$.
By contrast, the operators $U_M$ need not form a genuine representation
for an arbitrary choice of implementers. Over $\mathbb{F}_3$, the
finite Weil representation provides a choice for which the cocycle is
trivial; see \cite[Theorem~1.2]{gh09}.
We therefore choose the implementers so that
\(U_MU_N=U_{MN}\) and $U_MW_uU_M^{-1}=W_{Mu}$.
This choice uses the linearization theorem for the Weil representation;
it does not follow from Lemma~\ref{rep:covariance} alone.
These operators are unitary. Indeed, the covariance relation implies
that $U_M^\dagger U_M$ commutes with every Weyl operator.
Since the Weyl representation is irreducible, Schur's lemma gives
\(
U_M^\dagger U_M=\lambda_MI
\)
for some $\lambda_M>0$. Multiplicativity implies
\(
\lambda_{MN}=\lambda_M\lambda_N.
\)
Because $\operatorname{Sp}(2n,3)$ is finite, the only homomorphism from
$\operatorname{Sp}(2n,3)$ to $\mathbb R_{>0}$ is the trivial one. Hence
$\lambda_M=1$ for every $M$.

For this fixed choice the multiplication formula becomes
\begin{align}\label{E2.2J}
 (W_vU_M)(W_wU_N)=\omega^{[v,Mw]/2}W_{v+Mw}U_{MN}.
\end{align}
We now define the finite lift
\begin{align}\label{EG2}
\widetilde G_n
=
\{\omega^kW_vU_M:
k\in\mathbb F_3,\ v\in V_n,\ M\in\operatorname{Sp}(2n,3)\}.
\end{align}

\begin{proposition}
\label{rep:lift-properties}
Let $G_n$ and $\widetilde G_n$ be as in \eqref{EG1} and \eqref{EG2}, respectively. Then the map
\[\widetilde G_n\to G_n,\quad
\omega^kW_vU_M\mapsto(v,M)\]
is a surjective homomorphism with kernel
\(\{\omega^kI:k\in\mathbb F_3\}.\)
\end{proposition}

\begin{proof}
Since the phase is a power of $\omega$, the product in \eqref{E2.2J} belongs to
$\widetilde G_n$. The inverse also belongs to $\widetilde G_n$, because
\(U_M^{-1}=U_{M^{-1}}.\)
Thus $\widetilde G_n$ is a subgroup. It is finite because it is parametrized
by the finite sets $\mathbb F_3$, $V_n$ and $\operatorname{Sp}(2n,3)$.
The displayed map is a surjective homomorphism. If
$\omega^kW_vU_M$ is scalar, then it fixes every Weyl operator under
conjugation. Comparing labels gives $Mu=u$ for every $u$, hence $M=I$.
The remaining phases satisfy $\omega^{[v,u]}=1$ for every $u$, so
nondegeneracy gives $v=0$. Since $U_I=I$, the scalar element is
$\omega^kI$. Conversely, every $\omega^kI$ belongs to the kernel.
\end{proof}

The finite lift is useful because its scalar kernel has exponent three:
\((\omega^kI)^{\otimes3}=I.\)
This is exactly the reason that the tensor cube descends from
$\widetilde G_n$ to $G_n$.
\begin{proposition}\label{rep:cube-descent}
Define
\(\rho_3:G_n\to U(\mathcal K)\)
by \(\rho_3(v,M)=(W_vU_M)^{\otimes3}.\) The map $\rho_3$ is a well-defined genuine representation of $G_n$.
\end{proposition}

\begin{proof}
Suppose that $(v,M)$ is represented in $\widetilde G_n$ by $W_vU_M$.
If another representative in this finite lift is chosen, it differs by a scalar $\omega^k$, namely,
\(\omega^kW_vU_M.\)
Its tensor cube is unchanged because
\[(\omega^kW_vU_M)^{\otimes3}=(\omega^k)^3(W_vU_M)^{\otimes3}
=(W_vU_M)^{\otimes3}.\]
Thus $\rho_3$ is well-defined.
Together with \eqref{E2.2J}, the tensor cube eliminates the phase as follows:
\[
\begin{aligned}
(W_vU_M)^{\otimes3}(W_wU_N)^{\otimes3}
=
\bigl((W_vU_M)(W_wU_N)\bigr)^{\otimes3}=
\bigl(W_{v+Mw}U_{MN}\bigr)^{\otimes3}.
\end{aligned}
\]
Therefore,
\[
\rho_3(v,M)\rho_3(w,N)
=
\rho_3(v+Mw,MN).
\]
Hence $\rho_3$ is a genuine representation of $G_n$.
\end{proof}

Finally, every Clifford unitary has the form
\(\mathcal{C}=zW_vU_M.\)
Therefore,
\[
\mathcal{C}^{\otimes3}
=
z^3(W_vU_M)^{\otimes3}
=
z^3\rho_3(v,M).
\]
Let
\begin{align}\label{EG3}
\End_{\Cl(n,3)}(\KK):=
\left\{\mathcal{A}\in \End(\KK):\,\mathcal{A}\mathcal{C}^{\otimes 3}=\mathcal{C}^{\otimes 3}\mathcal{A},\ \forall \mathcal{C}\in \Cl(n,3)\right\}.
\end{align}
Since multiplication by a nonzero scalar does not change commutants,
invariant subspaces, or intertwiner spaces, we obtain
\[
\End_{\Cl(n,3)}(\mathcal K)
=
\End_{\rho_3(G_n)}(\mathcal K).
\]
Thus, the decomposition and irreducibility statements proved for
the finite group representation $\rho_3$ apply equally to the tensor-cube
action of the Clifford group.

\subsection{An eight-dimensional Clifford commutant from Weyl pairs}
\label{rep:pair-subsection}
The following calculation is the $d=t=3$ specialization of the
symplectic-orbit method in \cite[Lemmas~1--2 and Appendix~A]{zhu17}
and \cite[Lemma~4.8]{gnw21}. We include the argument to make the
decomposition self-contained.
The projective line $\mathbb P^1(\mathbb F_3)$ is the set of
one-dimensional linear subspaces of $\mathbb F_3^2$, i.e.\ the quotient
\[
  \mathbb P^1(\mathbb F_3)=\bigl(\mathbb F_3^2\setminus\{0\}\bigr)\big/\mathbb F_3^\times ,
  \quad \mathbb F_3^\times=\{1,2\}=\{\pm1\},
\]
of the eight nonzero vectors by scalar multiplication.
\begin{lemma}\label{rep:pair-orbits}
The diagonal action of $\operatorname{Sp}(2n,3)$ on ordered pairs
$(u,v)\in V_n^2$ has exactly eight orbits, consisting of the zero
pair, four nonzero rank-one types indexed by $\mathbb P^1(\mathbb F_3)$,
and three rank-two types indexed by the value $[u,v]\in\mathbb F_3$.
\end{lemma}
\begin{proof}
The dimension of $\operatorname{span}\{u,v\}$, and the stated data
within each dimension, are invariant. Any nonzero vector extends to a
symplectic basis: choose a partner with pairing one, split off their
nondegenerate plane, and continue in its symplectic orthogonal
complement. Thus the symplectic group is transitive on nonzero vectors.
A rank-one pair is $(ru,su)$ with $u\ne0$, and its invariant type is
$[r:s]\in\mathbb P^1(\F_3)$. There are four such types, represented
by $(u,0),(0,u),(u,u),(u,-u)$.
For an independent pair with $k=[u,v]\ne0$, the pair $(u,v/k)$ is
symplectic. Extending it to a symplectic basis proves transitivity for
each fixed $k$. For an independent pair with $[u,v]=0$, the functionals
$[u,\cdot]$ and $[v,\cdot]$ are independent. Choose $f_1,f_2$ satisfying
\[
 [u,f_1]=1,\quad[v,f_1]=0,\quad
 [u,f_2]=0,\quad[v,f_2]=1.
\]
If $c=[f_1,f_2]$, replacing $f_2$ by $f_2+cu$ makes its pairing with
$f_1$ zero without changing the other displayed pairings. These vectors
form two symplectic pairs and extend to a symplectic basis. Thus the
independent isotropic pairs also constitute one orbit; they exist
because $n\geq2$. The count is $1+4+3=8$.
\end{proof}

\begin{proposition}\label{rep:commutant-dimension}
Let $\operatorname{End}_{\Cl(n,3)}(\mathcal K)$ be as in \eqref{EG3}.
Then
\(
 \dim\End_{\Cl(n,3)}(\KK)=8.
\)
\end{proposition}
\begin{proof}
Every $\mathcal{A}\in\End(\KK)$ has an expansion
\[
\mathcal{A}=
\sum_{p,q,r\in V_n}
a_{p,q,r}\,
W_p\otimes W_q\otimes W_r,
\]
where the \(a_{p,q,r}\) are its coefficients in the Weyl tensor basis
\(
\{W_p\otimes W_q\otimes W_r:\, p,q,r\in V_n\}.
\)
We require that \(\mathcal{A}\) commute with every \(W_z^{\otimes 3}\), i.e.
\([W_z^{\otimes 3},\mathcal{A}]=0\), for any $z\in V_n$.
By \eqref{rep:weyl-identities}, we have
\[
W_z^{\otimes 3}
\left(
W_p\otimes W_q\otimes W_r
\right)
(W_z^\dagger)^{\otimes3}
=
\omega^{[z,p+q+r]}
\,
W_p\otimes W_q\otimes W_r.
\]
Whenever $a_{p,q,r}\ne0$, commutation requires
\(
\omega^{[z,p+q+r]}=1\)
for any $z\in V_n$.
By nondegeneracy of the symplectic form \([\cdot,\cdot]\), this is equivalent to
 \(r=-p-q.\)
Thus every commuting operator has the form
\begin{equation}\label{eq:pair-expansion}
\mathcal{A}=
\sum_{u,v\in V_n}
a_{u,v}\,
W_u\otimes W_v\otimes W_{-u-v}.
\end{equation}
So the three indices \(p,q,r\) are not independent; only two indices \(u,v\)
are free.
The commutant of the Pauli tensor-cube action is therefore parametrized by
coefficients \(a_{u,v}\), with \((u,v)\in V_n\times V_n\).
Clifford covariance imposes further constraints.
For every $M\in\operatorname{Sp}(2n,3)$, exact covariance sends the
basis operator indexed by $(u,v)$ to the one indexed by $(Mu,Mv)$,
with no additional phase (see Lemma~\ref{rep:covariance}). Hence an
operator $\mathcal{A}$ of the form \eqref{eq:pair-expansion} commutes with the
full Clifford action if and only if
\[
 a_{u,v}=a_{Mu,Mv} \quad\text{for every }M\in\operatorname{Sp}(2n,3).
\]
Thus its coefficients are constant on the
$\operatorname{Sp}(2n,3)$-orbits on $V_n\times V_n$. The orbit sums
have disjoint supports in the Weyl basis and are linearly independent.
Therefore
\[
 \dim\End_{\Cl(n,3)}(\KK)
 =\#\{\operatorname{Sp}(2n,3)\text{-orbits on }V_n\times V_n\}=8
\]
by Lemma~\ref{rep:pair-orbits}. This finishes the proof of
Proposition~\ref{rep:commutant-dimension}.
\end{proof}

\subsection{Ordinary three-copy Schur--Weyl duality}
\label{rep:classical-subsection}

We first recall the ordinary Schur--Weyl decomposition for three copies.
This provides a comparison with the decomposition of $\mathcal K$ under
the Clifford action considered above.
The symmetric group $S_3$ acts on $\mathcal K$ by permuting tensor
factors:
\(V_\pi(v_1\otimes v_2\otimes v_3)=v_{\pi^{-1}(1)}\otimes
v_{\pi^{-1}(2)}\otimes v_{\pi^{-1}(3)}.\)
The group $U(3^n)$ acts diagonally by
\(U\cdot(v_1\otimes v_2\otimes v_3)=Uv_1\otimes Uv_2\otimes Uv_3.\)
These two actions commute with each other, that is, for every
$\pi\in S_3$ and every $U\in U(3^n)$,
\(V_\pi U^{\otimes 3}=U^{\otimes 3}V_\pi.\)
Let
$$
\operatorname{End}_{U(3^n)}(\mathcal K)
=
\left\{
\mathcal{A}\in\operatorname{End}(\mathcal K):
\mathcal{A} U^{\otimes3}=U^{\otimes3}\mathcal{A},\ \forall U\in U(3^n)
\right\}.
$$
Then
\(\operatorname{End}_{U(3^n)}(\KK)=\operatorname{span}_{\mathbb C}\{V_\pi:\pi\in S_3\}.\)
In particular,
\begin{align}\label{E6}
\dim \operatorname{End}_{U(3^n)}(\KK)=6.
\end{align}
Define the central projections of the permutation algebra on $\mathcal K$ by
\begin{align}\label{eq:mixed-central-projections}
\Pi_{\mathrm{s}}
=
\frac{1}{6}\sum_{\pi\in S_3}V_\pi,
\quad
\Pi_{\mathrm{a}}
=
\frac{1}{6}\sum_{\pi\in S_3}
\operatorname{sgn}(\pi)V_\pi,
\quad
\Pi_{\mathrm{m}}
=
I-\Pi_{\mathrm{s}}-\Pi_{\mathrm{a}}.
\end{align}
Their ranges are the trivial, sign, and standard $S_3$-isotypic components:
\(\operatorname{Sym}^3\mathcal H\),
\(\bigwedge^3\mathcal H\) and
\(\mathbb S_{(2,1)}\mathcal H\otimes\mathrm{Std}\), respectively.

The irreducible representations of \(S_3\) consist of the trivial representation \(\mathbf{1}\), the sign representation \(\mathrm{sgn}\), and the 2-dimensional standard representation \(\mathrm{Std}\)  (see \cite{s04,v10} for background).
The double commutant theorem gives the usual three-copy
Schur--Weyl decomposition as follows:
\begin{align}\label{EDS}
\mathcal K
\cong
\left(\operatorname{Sym}^3\mathcal H\otimes\mathbf 1\right)
\oplus
\left(\wedge^3\mathcal H\otimes\operatorname{sgn}\right)
\oplus
\left(\mathbb S_{(2,1)}\mathcal H\otimes\operatorname{Std}\right).
\end{align}
Here the three spaces
\(
\operatorname{Sym}^3\mathcal H\),
\(\wedge^3\mathcal H\) and
\(\mathbb S_{(2,1)}\mathcal H\)
are irreducible and pairwise inequivalent $U(3^n)$-modules; see, for example, \cite{gnw21}.
For completeness, the mixed multiplicity space can be realized using
the transposition of the last two tensor factors. Let
\(
\mathcal H\otimes\operatorname{Sym}^2\mathcal H
\)
be the subspace fixed by the transposition of the last two tensor
factors. That transposition has fixed-space dimensions $1$, $0$, and $1$
in $\mathbf1$, $\sgn$, and $\mathrm{Std}$, respectively.
The full symmetrization map
\[
\Pi_{\mathrm{s}}:
\mathcal H\otimes\operatorname{Sym}^2\mathcal H
\to
\operatorname{Sym}^3\mathcal H
\]
is the projection onto the fully symmetric summand. Hence
\[
\mathbb S_{(2,1)}\mathcal H
=
\ker\left(
\Pi_{\mathrm{s}}:
\mathcal H\otimes\operatorname{Sym}^2\mathcal H
\to
\operatorname{Sym}^3\mathcal H
\right).
\]
The dimensions are
\begin{equation}\label{rep:classical-dimensions}
\left\{
  \begin{aligned}
  & \dim\Sym^3\HH=\binom{3^n+2}{3},\\
 &\dim\wedge^3\HH=\binom {3^n}3,\\
 &\dim\mathbb S_{(2,1)}\HH
 =3^n\binom{3^n+1}{2}-\binom{3^n+2}{3}.
  \end{aligned}
\right.
\end{equation}
Since the standard representation has dimension $2$, the mixed isotypic
component has dimension
\(2\dim\mathbb S_{(2,1)}\mathcal H.\)

\subsection{The Calderbank--Shor--Steane code projector}
\label{rep:css-subsection}
The Clifford commutant has dimension eight, whereas the permutation
algebra has dimension six. Thus the three permutation symmetry sectors
in \eqref{EDS} do not determine the Clifford decomposition. We use the
Calderbank--Shor--Steane (CSS) code projector of \cite[Example~4.26]{gnw21} to obtain the further
splitting.

Let
\begin{equation}\label{rep:css-definition}
|{\phi_x}\rangle:=\frac1{\sqrt3}\sum_{y\in\F_3}
                 |{x+y,y-x,y}\rangle,\quad
 p:=\sum_{x\in\F_3}|\phi_x\rangle\langle\phi_x|.
\end{equation}
The pairs $(x,y)$ give distinct triples, so $\phi_x$ are
orthonormal. They are fixed by $X^{\otimes3}$ and $Z^{\otimes3}$.
These order-three operators commute, and the orthogonal projector onto
their joint $+1$ eigenspace is
\[
 \frac19\sum_{a,b\in\F_3}(Z^{\otimes3})^a(X^{\otimes3})^b.
\]
This is a  CSS  code projector (see \cite{cs96,s96a,s96b} for more details).
Its trace is $3$, since all nonidentity terms have trace zero;
hence it equals $p$. After regrouping tensor factors by physical site, let $P:=p^{\otimes n}$. By the tensor product relation
$(a\otimes b)(c\otimes d)=(ac)\otimes (bd)$, \eqref{E2.1} and \eqref{rep:weyl-identities}, we obtain
\begin{equation}\label{rep:css-weyl-average}
 P=3^{-2n}\sum_{u\in V_n}W_u^{\otimes3},\quad \Tr P=3^n.
\end{equation}
The Weyl scalar phases disappear in tensor cubes. Under Clifford
conjugation each $W_u$ acquires a cube root phase and its label is
permuted, so \eqref{rep:css-weyl-average} proves that $P$ commutes with
every $\mathcal{C}^{\otimes3}$. It also commutes with every $V_\pi$.

Put $S=V_{(12)}$ and $T=V_{(123)}$. Swapping the first two entries
sends $\phi_x$ to $\phi_{-x}$. The cycle sends a summand of
\eqref{rep:css-definition} to $|{y,x+y,y-x}\rangle$, which is the same
sum after replacing $y$ by $y-x$. Thus every even permutation fixes
the code basis, and every odd permutation negates its label. On the
basis $|{\phi_{\boldsymbol x}}\rangle=\bigotimes_i|{\phi_{x_i}}\rangle$ of
$P\KK$, negation fixes exactly the zero label. Therefore
\begin{align}\label{eq:mixed-P-splitting}
 P\Pi_m=0,\ P\Pi_s=\tfrac12P(I+S),\ P\Pi_a=\tfrac12P(I-S).
\end{align}
Moreover,
$$\Tr(PS)=1,\  \operatorname{rank}(P\Pi_s)=\frac{3^n+1}{2},\
\operatorname{rank}(P\Pi_a)=\frac{3^n-1}{2}.$$
In the notation of the five sectors, these identities give
\begin{equation}\label{rep:five-projectors}
 Q_A=P\Pi_s,\quad Q_B=\Pi_s-Q_A,\quad
 Q_C=P\Pi_a,\quad Q_D=\Pi_a-Q_C,\quad Q_E=\Pi_m.
\end{equation}
They are orthogonal projections summing to $I$. Subtracting the first
and third ranks from \eqref{rep:classical-dimensions} yields
\begin{equation}\label{rep:rank-calculation}
\left\{
  \begin{aligned}
   & \operatorname{rank}Q_B
 =\frac{(3^n-1)(3^n+1)(3^n+3)}6,\\
    &\operatorname{rank}Q_D
 =\frac{(3^n-1)(3^n+1)(3^n-3)}6,\\
   &\operatorname{rank}Q_E=\frac{2\cdot3^n(3^{2n}-1)}3.
  \end{aligned}
\right.
\end{equation}
The following decomposition is given in
\cite[Example~4.26]{gnw21}. We derive it here from the orbit count
and the explicit projectors above.
\begin{theorem}\label{rep:five-irreducibles}
There are five pairwise inequivalent irreducible $\Cl(n,3)$-modules
$A$, $B$, $C$, $D$, $E$ with
\begin{equation}\label{rep:joint-decomposition}
 \KK\cong
 (A\otimes\mathbf1)\oplus(B\otimes\mathbf1)
 \oplus(C\otimes\sgn)\oplus(D\otimes\sgn)
 \oplus(E\otimes\mathrm{Std}).
\end{equation}
For the first four modules the dimensions equal the corresponding
sector ranks, while $\dim E=3^n(3^{2n}-1)/3$. Their multiplicities
as $\Cl(n,3)$-modules are $1,1,1,1,2$. Moreover,
\begin{align}
 \End_{\Cl(n,3)}(\KK)
 &=\operatorname{alg}_{\C}(P,V_\pi:\pi\in S_3)
   \cong\C^4\oplus M_2(\C),\label{rep:clifford-commutant}\\
 \End_{\Cl(n,3)\times S_3}(\KK)
 &=\bigoplus_{j=A,B,C,D,E}\C Q_j.\label{rep:joint-commutant}
\end{align}
Each summand in \eqref{rep:joint-decomposition} is irreducible for $\Cl(n,3)\times S_3$, and these $\Cl(n,3)\times S_3$-modules are inequivalent.
\end{theorem}
\begin{proof}
Let $\mathcal B=\operatorname{alg}_{\C}(P,V_\pi:\pi\in S_3)$.
It is contained in the Clifford commutant. On the symmetric sector,
$P$ separates the nonzero ranges of $Q_A$ and $Q_B$, while all
permutations act as the identity; thus this part of $\mathcal B$ is
$\C\oplus\C$. The same argument on the alternating sector gives
another $\C\oplus\C$. On the mixed sector $P$ vanishes, and the
permutations generate the full $M_2(\C)$ on the standard factor.
The central projections $\Pi_s,\Pi_a,\Pi_m$ belong to $\mathcal B$;
within the first two sectors $P$ supplies both further projections.
These blocks are therefore independent, so
\(
 \mathcal B\cong\C^4\oplus M_2(\C)\) and \(\dim\mathcal B=8.\)
Proposition~\ref{rep:commutant-dimension} proves that the entire
Clifford commutant has this same dimension. Inclusion and equality of
dimensions yield \eqref{rep:clifford-commutant}.
We spell out why this proves irreducibility, rather than only
invariance of the displayed subspaces. A unitary representation is
completely reducible because orthogonal complements of invariant
subspaces are invariant. On each of $Q_A\KK,Q_B\KK,Q_C\KK,Q_D\KK$
the commutant consists of scalar operators. A proper invariant subspace would
supply a nonscalar orthogonal projection in that corner, which is
impossible. On $Q_E\KK$ the Clifford group acts on the multiplicity factor
of its standard $S_3$-type. If this factor were reducible, its proper
invariant projection tensored with $I_2$ would enlarge the displayed
commutant beyond $I_E\otimes M_2(\C)$. Thus this factor is also
irreducible. If two distinct listed factors were equivalent, an
intertwiner, extended by zero, would create a nonzero off-diagonal
block in the commutant, contradicting its computed form. This proves
all five irreducibility and inequivalence assertions for the Clifford group and
the stated multiplicities.
Finally, imposing commutation with $S_3$ on the algebra
$\C^4\oplus M_2(\C)$ makes its last block scalar and leaves the first
four unchanged. This proves \eqref{rep:joint-commutant}. Its scalar
corners and absence of off-diagonal intertwiners prove the corresponding
product-group assertions by the same argument.
\end{proof}

\begin{remark}
By \eqref{E6} and Proposition~\ref{rep:commutant-dimension}, the
commutant of the Clifford action strictly contains
\(\operatorname{span}_{\mathbb C}\{V_\pi:\pi\in S_3\}\) for every $n\geq2$.
\end{remark}
\subsection{Primitive projectors and an explicit equivalence}
\label{rep:primitive-subsection}

The copy permutations are
\[
 S=V_{(12)}:\ |z_1,z_2,z_3\rangle\mapsto|z_2,z_1,z_3\rangle,
 \quad
 T=V_{(123)}:\ |z_1,z_2,z_3\rangle\mapsto|z_3,z_1,z_2\rangle,
\]
so that $S_3=\langle S,T\rangle$ with $S^2=T^3=I$ and $STS=T^{-1}$.
The sum of the spectral projections of $T$ for the eigenvalues
$\omega$ and $\omega^2$ is
\begin{equation}\label{rep:Pi_m-from-T}
 \Pi_m
 =\frac{I+\omega^2T+\omega T^2}{3}+\frac{I+\omega T+\omega^2T^2}{3}
 =\frac{2I-T-T^2}{3}.
\end{equation}

\begin{lemma}\label{lem:primitive-projectors}
Define $R_+=\tfrac12\Pi_m(I+S)$ and $
 R_-=\tfrac12\Pi_m(I-S).$
Then $R_+$ and $R_-$ are orthogonal projections with
\begin{equation}\label{eq:rank-Rpm}
 \operatorname{rank}R_+=\operatorname{rank}R_-
 =\frac{3^n(3^{2n}-1)}{3},
 \quad
 R_++R_-=Q_E .
\end{equation}
Each is primitive in the Clifford commutant, and $Q_A,Q_B,Q_C,Q_D,
R_+$ and $R_-$ are pairwise orthogonal projections summing to $I$.
\end{lemma}
\begin{proof}
By \eqref{rep:Pi_m-from-T}, $\Pi_m$ is a projection commuting with
$S$; hence so are $R_\pm$, and
$\frac{I+S}{2}+\frac{I-S}{2}=I$ gives $R_++R_-=\Pi_m=Q_E$.
For the ranks, $\operatorname{Tr}\Pi_m=\frac13(2\cdot3^{3n}-\operatorname{Tr}T
-\operatorname{Tr}T^2)$; and $T,T^2$ are $3$-cycles, whose fixed computational basis vectors
are exactly the $3^n$ computational basis vectors with $z_1=z_2=z_3$, so
$\operatorname{Tr}T=\operatorname{Tr}T^2=3^n$. Therefore
$\operatorname{Tr}\Pi_m=\frac{2\cdot3^n(3^{2n}-1)}{3}$ and
$S\Pi_m=\Pi_mS$ gives
\begin{equation*}
 \operatorname{Tr}(\Pi_mS)=\tfrac13\bigl(2\operatorname{Tr}S
 -\operatorname{Tr}(TS)-\operatorname{Tr}(T^2S)\bigr)=0,
\end{equation*}
because $S$, $TS$, and $T^2S$ are transpositions, each fixing
$3^{2n}$ computational basis vectors. Thus
$\operatorname{Tr}S=\operatorname{Tr}(TS)=\operatorname{Tr}(T^2S)=3^{2n}$.
Hence
\[
 \operatorname{Tr}R_\pm
 =\tfrac12\bigl(\operatorname{Tr}\Pi_m\pm\operatorname{Tr}(\Pi_mS)\bigr)
 =\frac{3^n(3^{2n}-1)}{3}.
\]
On the $E$-sector the Clifford commutant acts on the
multiplicity factor as $M_2(\mathbb C)$, and each of $R_+$ and $R_-$
acts there as a rank-one projection. The last assertion follows
from $Q_A+Q_B+Q_C+Q_D+Q_E=I$ and \eqref{eq:rank-Rpm}.
\end{proof}
By Lemma~\ref{lem:primitive-projectors}, the ranges of
$Q_A,Q_B,Q_C,Q_D,R_+,R_-$ give an orthogonal decomposition into minimal
Clifford-invariant subspaces.

\begin{lemma}\label{lem:K-intertwiner}
Set $K=\frac{T-T^2}{\sqrt3}.$
Then
\begin{equation}\label{rep:intertwiner}
 K^\dagger=-K,\quad K^2=-\Pi_m,\quad K^\dagger K=\Pi_m,
 \quad SKS=-K,\quad KR_+=R_-K .
\end{equation}
Moreover, $K$ commutes with every $\mathcal{C}^{\otimes3}$, and its
restriction $R_+\KK\to R_-\KK$ is a unitary Clifford intertwiner with
inverse $-K$.
\end{lemma}

\begin{proof}
On the $1$-eigenspace of $T$, $T=T^2=I$ and $K=0=\Pi_m$; on the
standard factor, diagonalize $T=\mathrm{diag}(\omega,\omega^2)$ and
compute $K=\frac{1}{\sqrt3}\,\mathrm{diag}(\omega-\omega^2,
\omega^2-\omega)$, so $K^\dagger=-K$ and
$K^2=\frac13(T-T^2)^2=-\Pi_m$. Hence $K^\dagger K=\Pi_m$, and $K$ is
unitary from $R_+\KK$ onto its image. The relation $SKS=-K$ is exactly
$STS^{-1}=T^{-1}$. Finally
\[
 KR_+=\tfrac12K\Pi_m(I+S)=\tfrac12K(I+S)
 =\tfrac12(I-S)K=R_-K,
\]
using $KS=-SK$ and $K\Pi_m=K$. Since both $T$ and $S$ commute with all
$\mathcal{C}^{\otimes3}$, so does $K$, and $\|Kv\|^2=\langle v,K^\dagger Kv
\rangle=\langle v,\Pi_mv\rangle=\|v\|^2$ on $\operatorname{ran}\Pi_m$.
The inverse of $K:\operatorname{ran}R_+\to\operatorname{ran}R_-$ is
$-K$, because $(-K)K=-K^2=\Pi_m$.
\end{proof}
\begin{corollary}\label{cor:standard-factor}
Put $E=R_+\KK$. Then
\begin{equation}\label{eq:En-tensor-C2}
 E\otimes\mathbb C^2\to Q_E\KK,
 \quad
 v\otimes e_+\mapsto v,\quad
 v\otimes e_-\mapsto Kv
\end{equation}
is a Clifford-module isomorphism, and in the basis $\{e_+,e_-\}$ the
copy permutations act only on the second factor, as
\begin{equation*}
 S=\begin{pmatrix}1&0\\0&-1\end{pmatrix},
 \quad
 T=\begin{pmatrix}-\tfrac12&-\tfrac{\sqrt3}{2}\\[2pt]
                  \tfrac{\sqrt3}{2}&-\tfrac12\end{pmatrix}.
\end{equation*}
\end{corollary}

\begin{proof}
The map \eqref{eq:En-tensor-C2} is bijective and $\mathbb C$-linear by
Lemma~\ref{lem:K-intertwiner}. By \eqref{rep:intertwiner},
$S(v)=v$ and $S(Kv)=-Kv$. On $\operatorname{ran}\Pi_m$, the identity
$T+T^2=-I$ and the definition of $K$ give
$T(v)=-\frac12v+\frac{\sqrt3}{2}Kv$,
$T(Kv)=-\frac{\sqrt3}{2}v-\frac12Kv$; equivalently, on
$\operatorname{ran}\Pi_m$,
\[
 (T+\tfrac12 I)^2=-\tfrac34 I, \quad
 (T+\tfrac12 I)K=-\tfrac{\sqrt3}{2}I.
\]
These identities give the displayed matrices. The Clifford group acts on
the first factor and $S_3$ on the second, so the two actions commute.
\end{proof}

\begin{remark}\label{rem:compare-CSS}
The projectors $R_\pm$ are distinct from the symmetric and alternating
subprojections $P\Pi_s$ and $P\Pi_a$ of the CSS code projector $P$: by
$\operatorname{Tr}(P\Pi_s)=\frac{\operatorname{Tr}P+\operatorname{Tr}(PS)}{2}
=\frac{3^n+1}{2}$ and $\operatorname{Tr}(P\Pi_a)=\frac{3^n-1}{2}$, those
ranges have dimensions $\frac{3^n+1}{2}$ and $\frac{3^n-1}{2}$, whereas
each of $R_+\KK$ and $R_-\KK$ has dimension $\frac{3^n(3^{2n}-1)}{3}$; the two families of
Clifford modules are inequivalent, and their ranges are orthogonal.
\end{remark}
\begin{remark}\label{rem:small-modules}
Choose one representative from each pair
$\{\boldsymbol x,-\boldsymbol x\}$ in
$\mathbb F_3^n\setminus\{0\}$; there are $\frac{3^n-1}{2}$ such pairs.
Since
$S|\phi_{\boldsymbol x}\rangle=|\phi_{-\boldsymbol x}\rangle$ and
$T|\phi_{\boldsymbol x}\rangle=|\phi_{\boldsymbol x}\rangle$, the code
space $P\KK$ splits orthogonally as
\begin{equation*}
 Q_A\KK
 =\operatorname{span}\Bigl\{|\phi_{\boldsymbol 0}\rangle,\
   \frac{|\phi_{\boldsymbol x}\rangle+|\phi_{-\boldsymbol x}\rangle}{\sqrt2}\Bigr\},
 \quad
 Q_C\KK
 =\operatorname{span}\Bigl\{\frac{|\phi_{\boldsymbol x}\rangle
   -|\phi_{-\boldsymbol x}\rangle}{\sqrt2}\Bigr\}.
\end{equation*}
Thus the two code sectors have dimensions
\begin{equation*}
 \dim(Q_A\KK)=\frac{3^n+1}{2},
 \quad
 \dim(Q_C\KK)=\frac{3^n-1}{2}.
\end{equation*}
The modules $B$ and $D$ are the orthogonal complements of $A$
and $C$ in $\Sym^3\HH$ and $\bigwedge^3\HH$, respectively, with
dimensions given in \eqref{rep:rank-calculation}. In particular,
$\dim D=0$ occurs exactly when $n=1$; for $n\geq1$ one has the
consistency identity
\begin{equation*}
 \dim A+\dim B+\dim C+\dim D+2\dim E
 =3^{3n},
 \quad
 \dim E=\frac{3^n(3^{2n}-1)}{3}.
\end{equation*}
\end{remark}
The identities $\operatorname{Tr}(P\Pi_s)=\frac{3^n+1}{2}$ and
$\operatorname{Tr}(P\Pi_a)=\frac{3^n-1}{2}$ are forced by
$\operatorname{Tr}P=3^n$ and $\operatorname{Tr}(PS)=1$. They give the dimensions
$\frac{3^n\pm1}{2}$ of the intersections of the CSS code with the symmetric
and alternating tensor powers. The ambient tensor powers have the larger
dimensions $\binom{3^n+2}{3}$ and $\binom{3^n}{3}$, respectively.
\subsection{Characters and numerical instances}
\label{rep:characters-subsection}
The projectors give character formulas for unitary Clifford matrices;
module dimensions alone do not specify rows of an entire character
table. For $\mathcal{C}\in\Cl(n,3)$ set
\[
 t_k(\mathcal{C})=\Tr(\mathcal{C}^k),\quad
 \alpha(\mathcal{C})=\Tr(P\mathcal{C}^{\otimes3}),\quad
 \beta(\mathcal{C})=\Tr(PS\mathcal{C}^{\otimes3}).
\]
For every operator $\mathcal{B}$ on $\HH$ and every $\pi\in S_3$,
\begin{equation}\label{eq:mixed-cycle-trace}
 \Tr\bigl(\mathcal{B}^{\otimes3}V_\pi\bigr)
 =\prod_{d=1}^{3}\bigl(\Tr \mathcal{B}^{d}\bigr)^{m_d(\pi)},
\end{equation}
where $m_d(\pi)$ is the number of cycles of $\pi$ of length $d$.
Indeed, choose an orthonormal basis $\{|e_a\rangle\}$ of $\HH$.  For
basis vectors, the definition of $V_\pi$ gives
\[
 \langle e_{a_1}e_{a_2}e_{a_3}|\,\mathcal{B}^{\otimes3}V_\pi\,|e_{a_1}e_{a_2}e_{a_3}\rangle
 =\prod_{i=1}^{3}\langle e_{a_i}|\mathcal{B}|e_{a_{\pi^{-1}(i)}}\rangle .
\]
Summation over $a_1,a_2,a_3$ factorizes over the cycles of $\pi$: a cycle
$(i_1i_2\cdots i_d)$ of $\pi^{-1}$ contributes
\[
 \sum_{a_{i_1},\dots,a_{i_d}}
 \langle e_{a_{i_1}}|\mathcal{B}|e_{a_{i_2}}\rangle
 \langle e_{a_{i_2}}|\mathcal{B}|e_{a_{i_3}}\rangle\cdots
 \langle e_{a_{i_d}}|\mathcal{B}|e_{a_{i_1}}\rangle=\Tr(\mathcal{B}^{d}),
\]
and distinct cycles involve disjoint sets of summation indices.  This
proves \eqref{eq:mixed-cycle-trace}.
Applying this identity to the sector projections yields the irreducible
characters
\begin{equation*}
\left\{\begin{aligned}
    \chi_A(\mathcal{C})&=\tfrac12\bigl(\alpha(\mathcal{C})+\beta(\mathcal{C})\bigr),\\
     \chi_B(\mathcal{C})&=\tfrac16(t_1^3+3t_1t_2+2t_3)-\chi_A(\mathcal{C}),\\
    \chi_C(\mathcal{C})&=\tfrac12\bigl(\alpha(\mathcal{C})-\beta(\mathcal{C})\bigr),\\
    \chi_D(\mathcal{C})&=\tfrac16(t_1^3-3t_1t_2+2t_3)-\chi_C(\mathcal{C}),\\
     \chi_E(\mathcal{C})&=\tfrac13(t_1^3-t_3).
  \end{aligned}
\right.
\end{equation*}
Here $t_k=t_k(\mathcal{C})$. For the last formula, the trace on $Q_E\KK$ is
$2(t_1^3-t_3)/3$, and each of its two equivalent constituents contributes
half. More precisely, $\Tr(\mathcal{C}^{\otimes3}\Pi_mS)=0$ shows directly that
$R_+$ and $R_-$ have these identical characters.
Equation~\eqref{rep:css-weyl-average} further gives formulas involving
only $3^n\times 3^n$ matrices:
\begin{equation*}
 \begin{aligned}
 \alpha(\mathcal{C})&=3^{-2n}\sum_{u\in V_n}\bigl(\Tr(W_u\mathcal{C})\bigr)^3,\\
 \beta(\mathcal{C})&=3^{-2n}\sum_{u\in V_n}
           \Tr\bigl((W_u\mathcal{C})^2\bigr)\Tr(W_u\mathcal{C}).
 \end{aligned}
\end{equation*}
The square in the second formula is $W_u\mathcal{C}W_u\mathcal{C}$. In particular
\[
 (\Tr \mathcal{C})^3=\chi_A(\mathcal{C})+\chi_B(\mathcal{C})+
\chi_C(\mathcal{C})+\chi_D(\mathcal{C})+2\chi_E(\mathcal{C}).
\]
All five characters satisfy $\chi_j(z\mathcal{C})=z^3\chi_j(\mathcal{C})$. This records
the scalar action explicitly. For the specified Weil lift in
Proposition~\ref{rep:cube-descent}, substituting $\mathcal{C}=W_vU_M$ gives the genuine
irreducible characters of $G_n$. This substitution is independent of the
remaining cube-root scalar in $\widetilde G_n$; an arbitrary unitary
phase change multiplies each character by its cube.

For $n=2,3,4$ the irreducible degrees are
\begin{equation*}
\begin{array}{c|rrrrr|r}
 n&\dim A&\dim B&\dim C&\dim D&\dim E&\dim\KK\\\hline
 2&5&160&4&80&240&729\\
 3&14&3640&13&2912&6552&19683\\
 4&41&91840&40&85280&177120&531441\\\hline
 \text{multiplicity}&1&1&1&1&2&
\end{array}
\end{equation*}
In each row the first four entries plus twice the fifth equal $3^{3n}$.
The sector $Q_E\KK$ has twice the fifth displayed dimension. At $n=1$
the orbit of linearly independent isotropic pairs is absent and $Q_D=0$, which
explains why all five sectors are nonzero only when $n\geq2$.

\section{Finite projective Clifford groups and exact character data}
\label{sec:finite-tables}

Write
\(\mathsf S_n:=\operatorname{Sp}(2n,3)\) and \(G_n=V_n\rtimes\mathsf S_n\)
with $G_0=\mathsf S_0=1$. The group $G_n$ acts on phase space by
affine symplectic transformations. The tensor-cube representation
$\mathcal{C}\mapsto \mathcal{C}^{\otimes3}$ of the chosen finite unitary lift descends to
$G_n$ by Proposition~\ref{rep:cube-descent}. Its irreducible constituents
therefore correspond to ordinary irreducible characters of $G_n$.
A larger central unitary lift has a different character table and may
have a different order and number of irreducible characters.

As in the qubit character-table computations of
Lee, Yu, Peng, and Lai~\cite{peng25}, we quotient out scalar phases;
here the underlying field is $\mathbb F_3$. Their conjugation action
on a matrix space and the tensor-power action on a Hilbert space are
different representations. The row and column labels in
Appendix~\ref{app:character-table-n2} are fixed by the accompanying
GAP data.

\subsection{Orders and the little-group construction}

For $n\geq1$, counting symplectic bases gives
\[
 |\mathsf S_n|=(3^{2n}-1)\,3^{2n-1}\,|\mathsf S_{n-1}|,
 \quad
 |G_n|=3^{2n}|\mathsf S_n|.
\]
In particular,
$|\mathsf S_1|=24$.
Write $\langle u,v\rangle=[u,v]$ for the symplectic form and fix a nontrivial additive character $\psi$ of $\mathbb F_3$. The
symplectic form identifies $V_n$ with its character group
$\widehat{V_n}$ by $
 v\mapsto\lambda_v$, $ \lambda_v(x)=\psi\bigl(\langle v,x\rangle\bigr).$
For $n\geq1$, the group $\mathsf S_n$ has exactly two orbits on $\widehat{V_n}$,
namely $\{1\}$ and the set of nontrivial characters; equivalently, its
action on $V_n$ has the two orbits $\{0\}$ and
$V_n\setminus\{0\}$, since every nonzero vector extends to a symplectic
basis.

\begin{lemma}\label{lem:little-group}
Let $H_n$ be the stabilizer in $\mathsf S_n$ of a fixed nonzero vector
$v\in V_n$, so that $[\mathsf S_n:H_n]=3^{2n}-1$. The character
$\lambda_v$ extends to the linear character
$\widetilde\lambda_v(x,h)=\lambda_v(x)$
of $V_n\rtimes H_n$. Let $\operatorname{Inf}_{\mathsf S_n}^{G_n}$ be the inflation along
$G_n\twoheadrightarrow G_n/V_n\cong\mathsf S_n$, i.e.\
$\operatorname{Inf}\sigma(x,M)=\sigma(M)$. Then the irreducible characters are parametrized by the disjoint union
\begin{equation}\label{eq:little-group}
 \operatorname{Irr}(G_n)
 =\operatorname{Inf}_{\mathsf S_n}^{G_n}\operatorname{Irr}(\mathsf S_n)
 \;\sqcup\;
 \Bigl\{
 \operatorname{Ind}_{V_n\rtimes H_n}^{G_n}
 \bigl(\widetilde\lambda_v\otimes
 \operatorname{Inf}_{H_n}^{V_n\rtimes H_n}\tau\bigr):
 \tau\in\operatorname{Irr}(H_n)\Bigr\},
\end{equation}
where $\tau$ ranges over the irreducible characters of the stabilizer
$H_n$ and is inflated to $V_n\rtimes H_n$. The induced character corresponding to
$\tau$ has degree
\begin{equation}\label{eq:induced-degree}
 (3^{2n}-1)\,\tau(1).
\end{equation}
\end{lemma}

\begin{proof}
Since $H_n$ fixes $v$, the map
$\widetilde\lambda_v(x,h)=\lambda_v(x)$ is multiplicative:
\[
 \widetilde\lambda_v\bigl((x,h)(y,g)\bigr)
 =\lambda_v(x+hy)=\lambda_v(x)\lambda_v(hy)
 =\lambda_v(x)\lambda_v(y)
 =\widetilde\lambda_v(x,h)\widetilde\lambda_v(y,g),
\]
where the third equality holds because $h\in H_n$ fixes $\lambda_v$. The orbit of
$\lambda_v$ under $\mathsf S_n$ is $\Lambda=\widehat{V_n}\setminus\{1\}$,
of size $3^{2n}-1$. The general little-group theorem parametrizes
$\operatorname{Irr}(G_n)$ by the $\mathsf S_n$-orbits on
$\widehat{V_n}$ together with the irreducible characters of the
corresponding stabilizer subgroups: over the orbit $\{1\}$ the stabilizer
is all of $\mathsf S_n$ and the irreducibles are the inflations of
$\operatorname{Irr}(\mathsf S_n)$, while over $\Lambda$ they are the
inductions in \eqref{eq:little-group}. The degree in
\eqref{eq:induced-degree} is
$[G_n:V_n\rtimes H_n]\cdot\tau(1)=(3^{2n}-1)\tau(1)$. Distinct
$\tau\in\operatorname{Irr}(H_n)$ give distinct induced irreducibles
by Clifford correspondence: the $\lambda_v$-isotypic subspace
recovers the inducing representation. This is the standard little-group
construction; see also \cite[Lemma~2.2]{fg21}.
\end{proof}

For explicit character evaluation, let $\Lambda$ be the orbit of
$\lambda_v$ in $\widehat{V_n}$. For each $\lambda\in\Lambda$ choose
$t_\lambda\in\mathsf S_n$ with
\begin{equation*}
 t_\lambda\cdot\lambda_v=\lambda,
 \quad
 (M\cdot\lambda)(x)=\lambda(M^{-1}x).
\end{equation*}
Then the character induced from $\tau\in\operatorname{Irr}(H_n)$ is
\begin{equation*}
 \chi_\tau(x,M)
 =\sum_{\substack{\lambda\in\Lambda\\ M\cdot\lambda=\lambda}}
 \lambda(x)\,\tau\bigl(t_\lambda^{-1}Mt_\lambda\bigr),
\end{equation*}
where $t_\lambda^{-1}Mt_\lambda\in H_n$ whenever
$M\cdot\lambda=\lambda$. A different choice of $t_\lambda$ changes
this element by conjugation inside $H_n$, so the character value is independent of this choice. For $\theta\in\operatorname{Irr}(\mathsf S_n)$, its inflation satisfies
$\bigl(\operatorname{Inf}_{\mathsf S_n}^{G_n}\theta\bigr)(x,M)=\theta(M)$. Thus the little-group construction gives both
a parametrization of all irreducible characters of $G_n$ and an exact
formula for their values.

\subsection{The vector stabilizer and a recursive degree formula}

Choose $e$ to be the fixed nonzero vector defining $H_n$, choose $f$
with $\langle e,f\rangle=1$, and put
$W=(\mathbb F_3e\oplus\mathbb F_3f)^\perp$. Then
\begin{equation*}
 V_n=\mathbb F_3e\oplus W\oplus\mathbb F_3f,
 \quad \langle e,f\rangle=1,\quad \dim W=2n-2 .
\end{equation*}
For $v\in W$ and $t\in\mathbb F_3$, define
\begin{equation*}
 u(v,t)(a,w,b)
 =\bigl(a+\langle v,w\rangle+tb,\;w+bv,\;b\bigr).
\end{equation*}

\begin{lemma}\label{lem:Heisenberg}
Each $u(v,t)$ is an element of $\mathsf S_n$ fixing $e$, and
\begin{equation}\label{eq:Heisenberg-law}
 u(v,t)\,u(v',t')
 =u\bigl(v+v',\,t+t'+\langle v,v'\rangle\bigr).
\end{equation}
Thus the $u(v,t)$ form a Heisenberg group $N_n$ of order
$3^{2n-1}$ with centre $Z=\{u(0,t):t\in\mathbb F_3\}$ of order three.
\end{lemma}

\begin{proof}
For $x=(a,w,b)$ and $y=(a',w',b')$,
\begin{align*}
 \langle u(v,t)x,u(v,t)y\rangle
 &=\bigl(a+\langle v,w\rangle+tb\bigr)b'
 -\bigl(a'+\langle v,w'\rangle+tb'\bigr)b
 +\langle w+bv,\,w'+b'v\rangle\\
 &=ab'-a'b+\langle w,w'\rangle ,
\end{align*}
so $u(v,t)$ preserves the symplectic form and fixes $e$.
Composition gives
\begin{align*}
 u(v,t)u(v',t')(a,w,b)
 &=u(v,t)\bigl(a+\langle v',w\rangle+t'b,\;w+bv',\;b\bigr)\\
 &=u(v+v',t+t'+\langle v,v'\rangle)(a,w,b),
\end{align*}
which is \eqref{eq:Heisenberg-law}. Since $u(v,t)f=(t,v,1)$, distinct pairs $(v,t)$ give distinct
transformations. They therefore form a group of order
$|W|\cdot3=3^{2n-1}$. Its commutators are
\[
 [u(v,t),u(v',t')]=u(0,2\langle v,v'\rangle),
 \quad [g,h]=ghg^{-1}h^{-1}.
\]
Nondegeneracy of the symplectic form on $W$ shows that the centre is
$Z=\{u(0,t)\}$, of order three. For $n=1$, $W=0$ and $N_1=Z$ is abelian.
\end{proof}

The embedding
\begin{equation*}
 \mathsf S_{n-1}\hookrightarrow H_n,\quad
 M(a,w,b)=(a,Mw,b),
\end{equation*}
satisfies $Mu(v,t)M^{-1}=u(Mv,t)$ and gives
\begin{equation}\label{eq:H-structure}
 H_n=N_n\rtimes\mathsf S_{n-1},
 \quad
 H_n/Z\cong G_{n-1}.
\end{equation}
Indeed the displayed subgroup of $H_n$ has order
\[
 3^{2n-1}|\mathsf S_{n-1}|
 =\frac{|\mathsf S_n|}{3^{2n}-1}=|H_n|,
\]
and hence equals $H_n$; the quotient by $Z$ is
$W\rtimes\mathsf S_{n-1}\cong V_{n-1}\rtimes\mathsf S_{n-1}=G_{n-1}$.
For $n=1$, this notation means $N_1=Z\cong C_3$.

The following degree formula refines the class-count recursion obtained
from the same stabilizer construction in \cite[Lemma~5.4]{fg21}.

\begin{theorem}\label{thm:degree-recursion}
Let $\mathcal S_n$ denote the multiset of irreducible character degrees
of $\mathsf S_n$, and $\mathcal D_n$ that of $G_n$. Let $\mathcal H_n$
denote the multiset of irreducible character degrees of the stabilizer
$H_n$ in Lemma~\ref{lem:little-group}.
For $n\geq1$,
\begin{equation}\label{eq:H-recursion}
 \mathcal H_n
 =\mathcal D_{n-1}\ \sqcup\
  \bigsqcup_{\nu=1}^{2}\,3^{n-1}\mathcal S_{n-1},
\end{equation}
and
\begin{equation}\label{eq:D-recursion}
 \mathcal D_n
 =\mathcal S_n\ \sqcup\ (3^{2n}-1)\,\mathcal H_n .
\end{equation}
Here multiplication of a multiset by an integer scales each of its entries,
and $\sqcup$ retains multiplicities. In particular,
$\mathcal H_1=\mathcal D_0\sqcup\mathcal S_0\sqcup\mathcal S_0
=\{1,1,1\}$ and
$\mathcal S_0=\mathcal D_0=\{1\}$.
\end{theorem}

\begin{proof}

\textbf{Step 1: the recursion for $\mathcal H_n$.} Let
$\mathrm{Inf}:\operatorname{Irr}(H_n/Z)\to\operatorname{Irr}(H_n)$ be
inflation along $H_n\to H_n/Z$. Its image consists exactly of the
irreducibles of $H_n$ that are trivial on $Z$, and
$H_n/Z\cong G_{n-1}$ by \eqref{eq:H-structure}; hence this part
contributes the multiset $\mathcal D_{n-1}$.
For the remaining irreducibles, let $\zeta$ be one of the two
nontrivial characters of $Z$. By the finite Stone--von Neumann
theorem, $N_n$ has, up to equivalence, a unique irreducible representation $\eta_\zeta$
with central character $\zeta$, of degree
$[N_n:Z]^{1/2}=3^{n-1}$. Since the base field has odd characteristic,
the finite Weil representation gives an extension of $\eta_\zeta$
to a linear representation
$\widetilde\eta_\zeta$ of $N_n\rtimes\mathsf S_{n-1}$; see~\cite[Theorems~1.1 and~1.2]{gh09}. For $n=1$, this extension is
simply the character $\zeta$ of $N_1=H_1=Z$. Since $\mathsf S_{n-1}$
fixes $Z$ pointwise, $\eta_\zeta$ is invariant up to equivalence.
Clifford theory for the normal subgroup $N_n$ therefore gives
\begin{equation*}
 \operatorname{Irr}(H_n\mid\zeta)
 =\Bigl\{\widetilde\eta_\zeta\otimes
 \operatorname{Inf}_{\mathsf S_{n-1}}^{H_n}\theta:
 \theta\in\operatorname{Irr}(\mathsf S_{n-1})\Bigr\},
\end{equation*}
where $\operatorname{Irr}(H_n\mid\zeta)$ denotes the irreducible
characters whose restrictions to $Z$ contain $\zeta$. Their degrees are
$3^{n-1}\theta(1)$. Summing over the two choices of $\zeta$
gives the second term of \eqref{eq:H-recursion}.

\textbf{Step 2: the recursion for $\mathcal D_n$.} Insert
\eqref{eq:H-recursion} into Lemma~\ref{lem:little-group}: the
trivial orbit contributes $\mathcal S_n$, and the nontrivial orbit,
of size $3^{2n}-1$, contributes $(3^{2n}-1)\mathcal H_n$ by
\eqref{eq:induced-degree}. This is \eqref{eq:D-recursion}.

\textbf{Step 3: consistency checks.} The multiset
\eqref{eq:H-recursion} satisfies
\[
 \sum_{d\in\mathcal H_n}d^2
 =\sum_{d\in\mathcal D_{n-1}}d^2+2\cdot3^{2n-2}\sum_{d\in\mathcal S_{n-1}}d^2
 =|G_{n-1}|+2\cdot3^{2n-2}|\mathsf S_{n-1}|
 =3^{2n-1}|\mathsf S_{n-1}|=|H_n|,
\]
and \eqref{eq:D-recursion} satisfies
\[
 \sum_{d\in\mathcal D_n}d^2
 =|\mathsf S_n|+(3^{2n}-1)^2|H_n|
 =|\mathsf S_n|+(3^{2n}-1)|\mathsf S_n|
 =3^{2n}|\mathsf S_n|=|G_n|,
\]
using $[\mathsf S_n:H_n]=3^{2n}-1$ and
$|G_n|=3^{2n}|\mathsf S_n|$. Both identities hold for every $n\geq1$;
as a check, for $n=1$ one obtains
$\mathcal H_1=\{1,1,1\}$ and
\[
 \mathcal D_1=\{1,1,1,2,2,2,3\}\ \sqcup\ 8\{1,1,1\}
 =\{1,1,1,2,2,2,3,8,8,8\},
\]
whose squared degrees sum to $1+1+1+4+4+4+9+3\cdot64=216=|G_1|$.
\end{proof}

\begin{remark}\label{rem:recursion-warning}
Substituting \eqref{eq:H-recursion} into \eqref{eq:D-recursion} gives
\[
\begin{aligned}
 \mathcal D_n={}&\mathcal S_n\sqcup(3^{2n}-1)\mathcal D_{n-1}\sqcup\bigsqcup_{\nu=1}^2((3^{2n}-1)3^{n-1})\mathcal S_{n-1}.
\end{aligned}
\]
This equivalent form satisfies the same degree-square identity.
\end{remark}

\section{Density operators and invariants}\label{sec:arbitrary}
\subsection{Density operators}
A density operator on a finite-dimensional Hilbert space is a linear
operator $\sigma$ satisfying
\begin{align}\label{eq:r}
\sigma=\sigma^\dagger,\quad \sigma\succeq0,\quad \Tr(\sigma)=1.
\end{align}
Here $\sigma\succeq0$ means that
$\langle x,\sigma x\rangle\geq0$ for every vector $x$.
The set of density operators on $\mathcal{K}$ will be denoted by
\[
\mathcal D(\mathcal{K})
 =
\{\sigma\in\End(\mathcal{K}):\sigma=\sigma^\dagger,\ \sigma\succeq0,\
\Tr(\sigma)=1\}.
\]
Let $Q_A,Q_B,Q_C,Q_D$ and $Q_E$ be the five mutually orthogonal
projectors from \eqref{rep:five-projectors}. Then by Theorem \ref{rep:five-irreducibles}, we have the irreducible $\Cl(n,3)\times S_3$-module decomposition
\[
\mathcal{K}=\bigoplus_{j=A}^E Q_j\mathcal{K}.
\]
A density operator invariant under the joint action of $\Cl(n,3)$ and $S_3$ has the form
\begin{equation}\label{eq:invariant-states}
 \sigma_{\boldsymbol w}
 =
 \sum_{j=A}^E \frac{w_j}{r_j}Q_j,\quad \boldsymbol w=(w_A,w_B,w_C,w_D,w_E),
 \quad
 w_j\geq0,\quad
 \sum_{j=A}^Ew_j=1.
\end{equation}
Indeed, Schur's lemma implies that the restriction of $\sigma$ to
each irreducible joint component is a scalar multiple of the identity.
The trace condition gives the normalization factor $r_j^{-1}$.
Conversely, every probability vector $\boldsymbol w$
defines an invariant density operator through
\eqref{eq:invariant-states}.

For a general density operator, the only restrictions are as in \eqref{eq:r}.
All diagonal and off-diagonal blocks may occur, subject to these
conditions; in particular the invariant states \eqref{eq:invariant-states}
form a proper subset of $\mathcal D(\mathcal K)$.
\subsection{The joint twirl and linear invariants}
For $\sigma\in\mathcal D(\mathcal K)$, define $w_j(\sigma)=\Tr(Q_j\sigma)$ for $j\in\{A,B,C,D,E\}$. These weights satisfy
\[
w_j(\sigma)\geq0,\quad
\sum_{j=A}^Ew_j(\sigma)=1.
\]
The joint twirl is the trace-preserving conditional expectation
\begin{equation}\label{eq:twirl}
 \mathcal E(\sigma)
 =
 \mathbb E_{\mathcal{C},\pi}
 \left[
 \mathcal{C}^{\otimes3}V_\pi\sigma V_\pi^\dagger
 \mathcal{C}^{\dagger\otimes3}
 \right]
 =
 \sum_{j=A}^E
 \frac{w_j(\sigma)}{r_j}Q_j.
\end{equation}
Here $\pi$ is uniform on $S_3$, independently of a uniform element of
the finite projective Clifford group. The choice of unitary
representative $\mathcal{C}$ does not affect conjugation.
Thus the five weights determine the joint twirl completely.
The weights are linear polynomial invariants of the joint conjugation
action.
\begin{remark}\label{RV}
The five linear invariants $w_j(\sigma)$ for
$j\in\{A,B,C,D,E\}$ do not determine an arbitrary state
$\sigma\in\mathcal D(\KK)$ up to unitary conjugation, even though they do
determine its joint twirl.  The simplest obstruction occurs inside the
$A$ sector as in Theorem \ref{rep:five-irreducibles}.  Recall the one-site code vector of
\eqref{rep:css-definition},
\[
 |\phi_0\rangle=\frac{1}{\sqrt3}\bigl(|000\rangle+|111\rangle
 +|222\rangle\bigr),
\]
which is fixed by every permutation of the three copies, and set
\[
 |v\rangle
 =\bigotimes_{i=1}^{n}|\phi_0\rangle,
 \quad
 \sigma_1=|v\rangle\langle v|,
 \quad
 \sigma_2=\frac{Q_A}{r_A}.
\]
The vector $|v\rangle$ is a unit vector in
the range of $Q_A$. Indeed, $p|\phi_0\rangle=|\phi_0\rangle$ gives
$P|v\rangle=|v\rangle$, and $V_\pi|v\rangle=|v\rangle$ for every
$\pi\in S_3$ gives $\Pi_s|v\rangle=|v\rangle$, so that
$Q_A|v\rangle=P\Pi_s|v\rangle=|v\rangle$.  Thus $\sigma_1$ is a pure state
concentrated on the $A$ sector, while $\sigma_2$ is the maximally mixed
state of that sector, that is, the normalized projection $Q_A$. In
particular, $\sigma_2$ is invariant under the joint action.  The two states
have the same five weights,
\begin{equation}\label{eq:collision-weights}
 w_A(\sigma_1)=w_A(\sigma_2)=1,\quad
 w_B=w_C=w_D=w_E=0,
\end{equation}
because $Q_A|v\rangle=|v\rangle$, $Q_j\succeq0$ and $\sum_jQ_j=I$ force
$Q_j|v\rangle=0$ for $j\neq A$, while $\Tr(Q_jQ_A)=r_A\delta_{jA}$.
Thus, the two states have the same joint twirl,
\[
 \mathcal E(\sigma_1)=\mathcal E(\sigma_2)=\frac{Q_A}{r_A}=\sigma_2 ,
\]
as \eqref{eq:twirl} and \eqref{eq:collision-weights} show.  Since the
twirl is the Hilbert--Schmidt orthogonal projection onto the joint
commutant, every measurement whose effects commute with the joint action
therefore yields the same outcome distribution on $\sigma_1$ and on
$\sigma_2$.  Their spectra differ, as reflected in their purities:
\begin{equation}\label{eq:collision-spectra}
 \Tr(\sigma_1^2)=1,\quad
 \Tr(\sigma_2^2)=\frac1{r_A}=\frac{2}{3^n+1}.
\end{equation}
Indeed, $\sigma_1$ has rank one, whereas $\sigma_2$ is maximally mixed on a subspace of dimension $r_A=\tfrac12(3^n+1)$, with every nonzero eigenvalue equal to $1/r_A$.  Since unitary conjugation preserves the spectrum, there is
no unitary $U$ on $\KK$ with $U\sigma_1U^\dagger=\sigma_2$, and hence no Clifford
operator $\mathcal{C}$ on $\HH$ with $\mathcal{C}^{\otimes3}\sigma_1\mathcal{C}^{\dagger\otimes3}=\sigma_2$ either.  In particular, the
degree-two invariant polynomial $\Tr(\sigma^2)$ takes the two different
values \eqref{eq:collision-spectra} on $\sigma_1$ and $\sigma_2$. Hence the Clifford-invariant
polynomials on $\mathcal D(\KK)$ are not all polynomials in the five
linear weights.
\end{remark}
\subsection{Spectral quantities and distances for invariant
states}\label{sec:spectral}
For $\mathcal{A}\in\End(\KK)$, let $|\mathcal{A}|:=\sqrt{\mathcal{A}^\dagger\mathcal{A}}$ be its absolute value and
$\|\mathcal{A}\|_1:=\Tr|\mathcal{A}|$ its trace norm.
Define the trace distance on $\mathcal D(\KK)$ by
\begin{equation*}
 \Dtr(\sigma,\tau):=\tfrac12\|\sigma-\tau\|_1,
 \quad \sigma,\tau\in\mathcal D(\KK).
\end{equation*}

For an invariant state $\sigma_{\boldsymbol w}$ as in \eqref{eq:invariant-states}, the following formulas follow from the orthogonality relations $Q_jQ_\ell=\delta_{j\ell}Q_j$ and the rank identities $\Tr(Q_j)=r_j$. In particular,
$$
Q_j\sigma_{\boldsymbol w}Q_\ell
=\delta_{j\ell}\frac{w_j}{r_j}Q_j,
\quad
\sigma_{\boldsymbol w}^{\,k}
=\sum_{j=A}^E\left(\frac{w_j}{r_j}\right)^kQ_j,
\quad k\geq1.
$$
Taking traces term by term gives
\begin{align*}
 \Tr(\sigma_{\boldsymbol w}^k)
=\sum_{j=A}^E\frac{w_j^k}{r_j^{\,k-1}},
   \quad k\geq1
\end{align*}
and
\begin{align}\label{eq:tv}
\Dtr(\sigma_{\boldsymbol w},\sigma_{\boldsymbol v})
=\frac12\sum_{j=A}^E|w_j-v_j|.
\end{align}
Identity \eqref{eq:tv} follows from the additivity of the trace norm over orthogonal direct sums. Indeed, the operator
$$
\sigma_{\boldsymbol w}-\sigma_{\boldsymbol v}
=\sum_{j=A}^E\frac{w_j-v_j}{r_j}Q_j
$$
is block diagonal with respect to the orthogonal decomposition of $\KK$, so
$$
\|\sigma_{\boldsymbol w}-\sigma_{\boldsymbol v}\|_1
=\sum_{j=A}^E\frac{|w_j-v_j|}{r_j}\|Q_j\|_1
=\sum_{j=A}^E|w_j-v_j|.
$$
\begin{remark}\label{rem:tv-vs-weights}
Let
$$
\mathcal I
=
\left\{
\sigma_{\boldsymbol w}
=
\sum_{j=A}^{E}\frac{w_j}{r_j}Q_j
\;\middle|\;
w_j\geq0,\ \sum_{j=A}^{E}w_j=1
\right\}
$$
be the family of jointly invariant density operators defined in \eqref{eq:invariant-states}. This family is the simplex with vertices $Q_A/r_A,\dots,Q_E/r_E$. Identity \eqref{eq:tv} holds on $\mathcal I$ but does not extend to arbitrary states in $\mathcal D(\KK)$. Indeed, the states $\sigma_1=|v\rangle\langle v|$ and $\sigma_2=Q_A/r_A$ from Remark~\ref{RV} have the same weight vector $(1,0,0,0,0)$, yet
$\Dtr(\sigma_1,\sigma_2)=1-r_A^{-1}>0.$ Here $\sigma_1\notin\mathcal I$ and $\sigma_2\in\mathcal I$, while their joint twirls coincide
$
\mathcal E(\sigma_1)=\mathcal E(\sigma_2)=Q_A/r_A,
$
where $\mathcal E$ is defined in \eqref{eq:twirl}.
For every $\sigma\in\mathcal D(\KK)$, the twirl satisfies $\mathcal E(\sigma)\in\mathcal I$ and preserves the sector weights. Thus, for any $\sigma,\tau\in\mathcal D(\KK)$,
$$
\Dtr\bigl(\mathcal E(\sigma),\mathcal E(\tau)\bigr)
=
\frac12\sum_{j=A}^{E}|w_j(\sigma)-w_j(\tau)|
\leq
\Dtr(\sigma,\tau),
$$
where the inequality follows from contractivity of the trace distance. Thus the sector weights determine the trace distance between the twirled states and provide a lower bound for the trace distance between the original states, but do not determine the latter in general.
\end{remark}

\subsection{Tensor cubes of density operators}

Now let
\(\sigma=\rho^{\otimes3}\), \(\rho\succeq0\) and \(\Tr(\rho)=1.\)
Since $\rho^{\otimes3}$ commutes with every copy permutation, its
Clifford twirl equals the joint twirl $\mathcal E(\rho^{\otimes3})$ as in
\eqref{eq:twirl}.  Let
\begin{equation}\label{eq:four}
 t_2=\Tr(\rho^2),\quad
 t_3=\Tr(\rho^3),\quad
 u=\Tr(P\rho^{\otimes3}),\quad
 v=\Tr(PS\rho^{\otimes3}).
\end{equation}
The operator $PS$ is Hermitian, so $v$ is real, although it need not be
nonnegative.

\begin{theorem}\label{thm:mixed}
For $j\in\{A,B,C,D,E\}$, let $w_j=\Tr(Q_j\rho^{\otimes3})$.
Then
\begin{equation}\label{eq:weights}
\left\{
\begin{aligned}
 w_A&=\frac{u+v}{2},\\
 w_B&=\frac{1+3t_2+2t_3}{6}-\frac{u+v}{2},\\
 w_C&=\frac{u-v}{2},\\
 w_D&=\frac{1-3t_2+2t_3}{6}-\frac{u-v}{2},\\
 w_E&=\frac{2(1-t_3)}{3}.
\end{aligned}
\right.
\end{equation}
Moreover, the third moment of the Clifford orbit of $\rho$ is determined by
the four scalar quantities $(t_2,t_3,u,v)$.
\end{theorem}

\begin{proof}
Let $\Pi_s$, $\Pi_a$ and $\Pi_m$ be as in \eqref{eq:mixed-central-projections}.
They are pairwise orthogonal
projections with $\Pi_s+\Pi_a+\Pi_m=I$, corresponding to the trivial,
sign, and two-dimensional irreducible representations of $S_3$.  Hence
$\{\Pi_s,\Pi_a,\Pi_m\}$ is a projective measurement on
$\KK=\HH^{\otimes3}$, and we write
$\Pr(\mathrm{sym})$, $\Pr(\mathrm{alt})$, $\Pr(\mathrm{mix})$ for the
probabilities of its three outcomes on the state $\rho^{\otimes3}$, i.e.
\begin{equation*}
 \Pr(\mathrm{sym}):=\Tr\bigl(\Pi_s\rho^{\otimes3}\bigr), \quad
 \Pr(\mathrm{alt}):=\Tr\bigl(\Pi_a\rho^{\otimes3}\bigr), \quad
 \Pr(\mathrm{mix}):=\Tr\bigl(\Pi_m\rho^{\otimes3}\bigr),
\end{equation*}
so that the three numbers are nonnegative and sum to one.  Then
\begin{equation}\label{eq:mixed-blocksums}
 w_A+w_B=\Pr(\mathrm{sym}), \quad
 w_C+w_D=\Pr(\mathrm{alt}), \quad
 w_E=\Pr(\mathrm{mix}).
\end{equation}

The proof has five steps.

{\bf{Step 1.}}
By \eqref{eq:mixed-cycle-trace} and $\Tr\rho=1$, the three cycle types give
\begin{equation*}
\begin{array}{c|ccc}
 \pi & e & \text{trans\-position} & \text{$3$-cycle}\\ \hline
 m_d(\pi) & m_1=3 & m_1=m_2=1 & m_3=1\\[2pt]
 \Tr\bigl(\rho^{\otimes3}V_\pi\bigr) & 1 & t_2 & t_3
\end{array}
\end{equation*}
where $t_2=\Tr(\rho^2)$ and $t_3=\Tr(\rho^3)$ as in \eqref{eq:four}.
Taking the two averages that define $\Pi_s$ and $\Pi_a$ in
\eqref{eq:mixed-central-projections}, and using that $\sgn$ equals $+1$
on $e$ and on $3$-cycles and $-1$ on transpositions, we obtain
\begin{equation}\label{eq:mixed-sym-alt}
 \Pr(\mathrm{sym})=\frac{1+3t_2+2t_3}{6}, \quad
 \Pr(\mathrm{alt})=\frac{1-3t_2+2t_3}{6},
\end{equation}
and then, from
$\Pr(\mathrm{mix})=1-\Pr(\mathrm{sym})-\Pr(\mathrm{alt})$, it follows that
\begin{equation}\label{eq:mixed-mix}
 \Pr(\mathrm{mix})=\frac{2(1-t_3)}{3}.
\end{equation}

{\bf{Step 2.}}
Recall from Subsection~\ref{rep:css-subsection} that the
CSS projector $P$ of \eqref{rep:css-weyl-average} is
the orthogonal projection onto the span of the orthonormal vectors
\[
 |\phi_{\boldsymbol x}\rangle
 =\bigotimes_{i=1}^{n}|\phi_{x_i}\rangle,
 \quad
 |\phi_x\rangle
 =3^{-1/2}\sum_{y\in\F_3}|x+y,\,y-x,\,y\rangle,
 \quad \boldsymbol x\in\F_3^n ,
\]
and that $P$ commutes with every $V_\pi$.  On this code space the
permutations act as
\begin{equation}\label{eq:mixed-code-action}
 V_\pi\bigl|\phi_{\boldsymbol x}\bigr\rangle
 =\left\{
   \begin{aligned}
    &\bigl|\phi_{\boldsymbol x}\bigr\rangle,\ \  \pi\ \text{even},\\
    &\bigl|\phi_{-\boldsymbol x}\bigr\rangle, \ \ \pi\ \text{odd}.
   \end{aligned}
 \right.
\end{equation}
Let us verify
the two cases on a single site.  Let $S=V_{(12)}$ and $T=V_{(123)}$.  Interchanging the first two entries of
the summands of $|\phi_x\rangle$ and then using the definition of
$|\phi_{-x}\rangle$ gives
\[
 S|\phi_x\rangle
 =3^{-1/2}\sum_{y\in\F_3}|y-x,\,y+x,\,y\rangle
 =|\phi_{-x}\rangle ,
\]
and for the cycle $T$ the same computation gives
\[
 T|\phi_x\rangle
 =3^{-1/2}\sum_{y\in\F_3}|y,\,x+y,\,y-x\rangle
 =|\phi_x\rangle ,
\]
where the last equality follows by setting $y'=y-x$, since
$-2x\equiv x$ in $\F_3$.  Tensoring over the $n$ sites yields
\eqref{eq:mixed-code-action}.
Thus, for every $\boldsymbol x$,
\[
 \Pi_s\bigl|\phi_{\boldsymbol x}\bigr\rangle
 =\frac16\Bigl(3\bigl|\phi_{\boldsymbol x}\bigr\rangle
 +3S\bigl|\phi_{\boldsymbol x}\bigr\rangle\Bigr)
 =\tfrac12(I+S)\bigl|\phi_{\boldsymbol x}\bigr\rangle,
 \quad
 \Pi_a\bigl|\phi_{\boldsymbol x}\bigr\rangle
 =\tfrac12(I-S)\bigl|\phi_{\boldsymbol x}\bigr\rangle,
 \quad
 \Pi_m\bigl|\phi_{\boldsymbol x}\bigr\rangle=0 ,
\]
because three of the six permutations are even and three are odd, while
$\Pi_m=I-\Pi_s-\Pi_a$.  The vectors $|\phi_{\boldsymbol x}\rangle$ span
the range of $P$, so multiplying by the projection $P$ turns these
identities into the operator identities \eqref{eq:mixed-P-splitting}.
With $Q_A=P\Pi_s$ and $Q_C=P\Pi_a$, and with the definitions of $u$ and
$v$ in \eqref{eq:four}, this gives
\[
 w_A=\Tr\bigl(Q_A\rho^{\otimes3}\bigr)
 =\frac12\Bigl(\Tr(P\rho^{\otimes3})+\Tr(PS\rho^{\otimes3})\Bigr)
 =\frac{u+v}{2},
 \quad
 w_C=\Tr\bigl(Q_C\rho^{\otimes3}\bigr)=\frac{u-v}{2}.
\]
As a check on the normalization, \eqref{eq:mixed-P-splitting} together
with $\Tr(PS)=1$ yields
$\operatorname{rank}Q_A=\Tr(P\Pi_s)=\tfrac12(3^n+1)$,
$\operatorname{rank}Q_C=\Tr(P\Pi_a)=\tfrac12(3^n-1)$ in agreement with \eqref{rep:rank-calculation}.

{\bf{Step 3.}}
Because $Q_B=\Pi_s-Q_A$ and $Q_D=\Pi_a-Q_C$, the block sums
\eqref{eq:mixed-blocksums} and \eqref{eq:mixed-sym-alt} give
\[
 w_B=\Pr(\mathrm{sym})-w_A
 =\frac{1+3t_2+2t_3}{6}-\frac{u+v}{2}, \quad
 w_D=\Pr(\mathrm{alt})-w_C
 =\frac{1-3t_2+2t_3}{6}-\frac{u-v}{2}.
\]

{\bf{Step 4.}}
By \eqref{eq:mixed-blocksums} and
\eqref{eq:mixed-mix}, we have
\[
 w_E=\Tr\bigl(\Pi_m\rho^{\otimes3}\bigr)=\Pr(\mathrm{mix})
 =\frac{2(1-t_3)}{3},
\]
which is also what the last line of \eqref{eq:weights} gives.
The five numbers $w_j=\Tr(Q_j\rho^{\otimes3})$ are nonnegative, because
$Q_j\succeq0$ and $\rho^{\otimes3}\succeq0$, and
\[
 \sum_{j=A}^{E}w_j
 =\Pr(\mathrm{sym})+\Pr(\mathrm{alt})+\Pr(\mathrm{mix})=1.
\]
This
identity is a consistency check on \textbf{Steps 1--4}.

{\bf{Step 5.}}
The third moment of the uniform Clifford orbit of $\rho$ is
\[
 \mathbb E_{\mathcal{C}}\bigl[(\mathcal{C}\rho \mathcal{C}^\dagger)^{\otimes3}\bigr]
 =\mathbb E_{\mathcal{C}}\bigl[\mathcal{C}^{\otimes3}\rho^{\otimes3}\mathcal{C}^{\dagger\otimes3}\bigr],
\]
and since $V_\pi$ commutes with $\mathcal{C}^{\otimes3}$ and
$V_\pi\rho^{\otimes3}V_\pi^\dagger=\rho^{\otimes3}$ for every $\pi$, this
equals the joint twirl
\[
 \mathbb E_{\mathcal{C},\pi}\Bigl[\mathcal{C}^{\otimes3}V_\pi\rho^{\otimes3}
 V_\pi^\dagger \mathcal{C}^{\dagger\otimes3}\Bigr]
 =\sum_{j=A}^{E}\frac{w_j}{r_j}Q_j ,
\]
the last equality being \eqref{eq:twirl} and
$w_j=\Tr(Q_j\rho^{\otimes3})$.  The projections $Q_A,\dots,Q_E$ are the
central projections of the joint commutant
$\End_{\Cl(n,3)\times S_3}(\KK)=\bigoplus_j\C Q_j$ of
\eqref{rep:joint-commutant}. They are linearly independent when $n\geq2$,
and for $n=1$ the projection $Q_D$ vanishes, so the zero sector is
omitted from the twirl sum and $w_D=0$. Hence in either case the moment operator is determined by
the weights, and \textbf{Steps 1--4} express those weights in terms of
$(t_2,t_3,u,v)$. This proves the second assertion.
\end{proof}

\subsection{Polynomial invariants of degree at most three}
Throughout this subsection, assume that $n\geq2$. 
Let $V$ be the real vector space of Hermitian
$3^n\times 3^n$ matrices, with Clifford conjugation action
$X\mapsto CXC^\dagger$.
\begin{theorem}\label{thm:polynomials}
The space of homogeneous cubic real polynomials on $V$ invariant
under Clifford conjugation is spanned by
\[
(\Tr X)^3,\quad
(\Tr X)\Tr(X^2),\quad
\Tr(X^3),\quad
\Tr(PX^{\otimes3}),\quad
\Tr(PSX^{\otimes3}).
\]
These five polynomials are linearly independent.
On the affine space of trace-one Hermitian matrices, the restrictions
of polynomials invariant under Clifford conjugation and of total degree at most
three have basis $1,t_2,t_3,u,v$, with $t_2,t_3,u,v$ as in
\eqref{eq:four}.
Moreover, the four functions
$t_2,t_3,u,v$ are algebraically independent for $n\geq2$.
\end{theorem}
\begin{proof}
Write $h(X)=\Tr X$, $p_2(X)=\Tr(X^2)$, $p_3(X)=\Tr(X^3)$,
$u(X)=\Tr(PX^{\otimes3})$, $v(X)=\Tr(PSX^{\otimes3})$, and
$F_j(X)=\Tr(Q_jX^{\otimes3})$ for $j\in\{A,B,C,D,E\}$.

The proof has six steps.

{\bf{Step 1.}} Let $p$ be a homogeneous cubic real polynomial on $V$.
Its polarization is the symmetric real trilinear form
$$
B(X,Y,Z)
=
\frac16
\left.
\partial_r\partial_s\partial_t
p(rX+sY+tZ)
\right|_{r=s=t=0},
$$
which satisfies $B(X,X,X)=p(X)$.
The trace pairing
$
\langle X,Y\rangle=\operatorname{Tr}(XY)
$
is a real inner product on $V$. Hence there is a unique tensor
$\mathcal{A}\in V^{\otimes_{\mathbb R}3}$ such that
$
B(X,Y,Z)
=
\operatorname{Tr}\bigl(\mathcal{A}(X\otimes Y\otimes Z)\bigr)
$
for all $X,Y,Z\in V$. Here we identify
$V^{\otimes_{\mathbb R}3}$ with the real vector space of
Hermitian operators on $\mathcal K$.
Indeed, tensor products of a real orthonormal basis of $V$
form a real orthonormal basis of that space. In particular,
$\mathcal{A}$ is already Hermitian.
Since $B$ is symmetric, uniqueness of its representing tensor
implies
$V_\pi \mathcal{A} V_\pi^\dagger=\mathcal{A}$ for every $\pi\in S_3.$
Evaluating on the diagonal gives
$
p_{\mathcal{A}}(X)=\operatorname{Tr}(\mathcal{A}X^{\otimes3}).
$
Conversely, every permutation-invariant Hermitian operator
$\mathcal{A}$ defines a homogeneous cubic real polynomial by this
formula. Its associated trilinear form
$
B_{\mathcal{A}}(X,Y,Z)
=
\operatorname{Tr}\bigl(\mathcal{A}(X\otimes Y\otimes Z)\bigr)
$
is symmetric, and is therefore the polarization of $p_{\mathcal{A}}$.
Thus, the assignment $\mathcal{A}\mapsto p_{\mathcal{A}}$ is a real linear
bijection between permutation-invariant Hermitian operators
on $\mathcal K$ and homogeneous cubic real polynomials on $V$.
The permutation invariance condition is essential for uniqueness.
For an arbitrary Hermitian operator $\mathcal{A}$, define
$\operatorname{Sym}(\mathcal{A})=\frac16\sum_{\pi\in S_3}V_\pi \mathcal{A} V_\pi^\dagger.$
Then
$
p_{\mathcal{A}}(X)
=
\operatorname{Tr}\bigl(\operatorname{Sym}(\mathcal{A})X^{\otimes3}\bigr),
$
so the diagonal polynomial determines only
$\operatorname{Sym}(\mathcal{A})$, rather than the original operator $\mathcal{A}$.

{\bf{Step 2.}}
Let $\mathcal{A}$ be Hermitian and invariant under the copy permutations, and
suppose that $\Tr(\mathcal{A}X^{\otimes3})=0$ for every $X\in V$.  Let
$H_1,H_2,H_3\in V$ and put $X=t_1H_1+t_2H_2+t_3H_3$.  Expanding
$X^{\otimes3}$, the coefficient of $t_1t_2t_3$ in
$\Tr(\mathcal{A}X^{\otimes3})$ is
\[
 \sum_{\pi\in S_3}\Tr\bigl(\mathcal{A}(H_{\pi(1)}\otimes H_{\pi(2)}\otimes
 H_{\pi(3)})\bigr)
 =6\,\Tr\bigl(\mathcal{A}(H_1\otimes H_2\otimes H_3)\bigr),
\]
by permutation invariance of $\mathcal{A}$.  Hence
$\Tr\bigl(\mathcal{A}(H_1\otimes H_2\otimes H_3)\bigr)=0$ for all Hermitian $H_i$.
These tensor products span the real space of Hermitian operators on
$\KK$, so $\mathcal{A}=0$.

{\bf{Step 3.}}
Let $p=p_{\mathcal A}$ and $B$ be its polarization as in {\bf{Step 1}}.

(i)  Since $B$ is symmetric,
$\Tr(\mathcal{A}(X_1\otimes X_2\otimes X_3))
=\Tr(\mathcal{A}(X_{\pi(1)}\otimes X_{\pi(2)}\otimes X_{\pi(3)}))$ for all
Hermitian $X_1,X_2,X_3$ and all $\pi\in S_3$.  Since
$V_\pi(Y_1\otimes Y_2\otimes Y_3)V_\pi^\dagger
=Y_{\pi^{-1}(1)}\otimes Y_{\pi^{-1}(2)}\otimes Y_{\pi^{-1}(3)}$, this
identity reads $\Tr(\mathcal{A}Y)=\Tr(V_\pi^\dagger \mathcal{A}V_\pi Y)$ for the monomials
$Y=X_1\otimes X_2\otimes X_3$, and these span $\End(\KK)$ over $\C$, and hence
$V_\pi \mathcal{A}V_\pi^\dagger=\mathcal{A}$.

(ii) Suppose that $p$ is invariant under Clifford conjugation. Then, for every Clifford operator $\mathcal{C}$ and every Hermitian $X$,
\[
 \Tr(\mathcal{A}X^{\otimes3})=\Tr\bigl(\mathcal{A}\,(\mathcal{C}X\mathcal{C}^\dagger)^{\otimes3}\bigr)
 =\Tr\bigl(\mathcal{C}^{\dagger\otimes3}\mathcal{A}\mathcal{C}^{\otimes3}X^{\otimes3}\bigr),
\]
so $\Tr\bigl((\mathcal{C}^{\dagger\otimes3}\mathcal{A}\mathcal{C}^{\otimes3}
-\mathcal{A})X^{\otimes3}\bigr)=0$
for all $X\in V$.  The operator
$\mathcal{C}^{\dagger\otimes3}\mathcal{A}\mathcal{C}^{\otimes3}-\mathcal{A}$ is Hermitian, and it commutes with
every $V_\pi$ because $\mathcal{A}$ does and because $\mathcal{C}^{\otimes3}$ commutes with
$V_\pi$. {\bf{Step 2}} therefore yields
$\mathcal{C}^{\otimes3}\mathcal{A}\mathcal{C}^{\dagger\otimes3}=\mathcal{A}$.

Conversely, if $\mathcal A$ lies in the joint commutant, cyclicity of the
trace shows that $p_{\mathcal A}$ is invariant under Clifford conjugation.
Hence invariance of $p_{\mathcal{A}}$ is equivalent to
$\mathcal{A}\in\operatorname{End}_{\Cl(n,3)\times S_3}(\KK)$, i.e.\ to $\mathcal{A}$ lying
in the Hermitian part of the joint commutant.

{\bf{Step 4.}}
From \eqref{rep:joint-commutant}, it follows that
$\operatorname{End}_{\Cl(n,3)\times S_3}(\KK)=\bigoplus_{j}\C Q_j$, and
each $Q_j$ is an orthogonal projection.  By the proof of \textbf{Steps 1--3}, we know that the
Hermitian part of the joint commutant is the real span of
$Q_A,\dots,Q_E$.  All five projections are nonzero for $n\geq2$ by
\eqref{rep:rank-calculation}.  They are pairwise orthogonal and
independent, so this real span has dimension five.  By {\bf{Steps 1--3}} the
map $\mathcal A\mapsto p_{\mathcal A}$ is a linear bijection from that space onto the space
of invariant cubics, which therefore has dimension five as well.
Under this bijection, the projections $Q_j$ correspond to the
polynomials $F_j$ for $j\in\{A,B,C,D,E\}$, which therefore form a basis of the invariant cubics.
We now express this basis in terms of the five displayed polynomials. Recall that for every operator $\mathcal{B}$ on $\HH$ and every
$\pi\in S_3$ one has \eqref{eq:mixed-cycle-trace}. Applied to $\mathcal{B}=X\in V$ and to
the three cycle types of $S_3$, this gives
$\Tr(V_\pi X^{\otimes3})=h^3$ for $\pi=e$, $hp_2$ for a
transposition, and $p_3$ for a $3$-cycle.
Taking the averages that define $\Pi_s$ and $\Pi_a$, and using
$\Pi_m=(2I-T-T^2)/3$ from \eqref{rep:Pi_m-from-T}, we get
\begin{equation}\label{eq:F-summary}
\left\{
  \begin{aligned}
 &F_A+F_B=\Tr(\Pi_sX^{\otimes3})
   =\frac{h^3+3hp_2+2p_3}{6},\\
 &F_C+F_D=\Tr(\Pi_aX^{\otimes3})
   =\frac{h^3-3hp_2+2p_3}{6},\\
 &F_E=\Tr(\Pi_mX^{\otimes3})=\frac{2(h^3-p_3)}{3}.
\end{aligned}
\right.
\end{equation}
By the operator identities \eqref{eq:mixed-P-splitting}, we have
\begin{equation}\label{eq:F-AC}
 F_A=\frac{\Tr(PX^{\otimes3})+\Tr(PSX^{\otimes3})}{2}=\frac{u+v}{2},
 \quad
 F_C=\frac{u-v}{2},
\end{equation}
and subtracting \eqref{eq:F-AC} from the first two lines of
\eqref{eq:F-summary},
\begin{equation}\label{eq:F-BD}
 F_B=\frac{h^3+3hp_2+2p_3}{6}-\frac{u+v}{2},\quad
 F_D=\frac{h^3-3hp_2+2p_3}{6}-\frac{u-v}{2}.
\end{equation}
Equations \eqref{eq:F-summary}--\eqref{eq:F-BD} express the $F_j$
linearly in terms of $h^3,hp_2,p_3,u,v$.  Conversely, $\sum_jQ_j=I$ gives $h^3=\sum_jF_j$. Subtracting
the second line of \eqref{eq:F-summary} from the first gives
$hp_2=F_A+F_B-F_C-F_D$. The third line of \eqref{eq:F-summary} gives
$p_3=F_A+F_B+F_C+F_D-\tfrac12F_E$, and
\eqref{eq:F-AC} gives $u=F_A+F_C$ and $v=F_A-F_C$.  The two systems are
therefore mutually inverse linear maps, so the five displayed
polynomials $h^3,hp_2,p_3,u,v$ span the invariant cubics.  Since that
space has dimension five and there are five of them, they are linearly
independent, and they form a basis.  This proves the first two
assertions of Theorem \ref{thm:polynomials}.

{\bf{Step 5.}}
Let $f$ be an invariant polynomial of total degree at most three and let
$f=\sum_{k=0}^{3}f_k$ be its decomposition into homogeneous components.
Each $f_k$ is invariant, because the action is linear and preserves
degree. Hence
\[
 g(X)=f_0(X)(\Tr X)^3+f_1(X)(\Tr X)^2+f_2(X)\Tr X+f_3(X)
\]
is a homogeneous cubic invariant polynomial, and $g$ agrees with $f$ on
the affine hyperplane $\{\Tr X=1\}$. Thus every restriction of an invariant
polynomial of total degree at most three is the restriction of an invariant
cubic. The restriction map on invariant cubics is also injective: if a homogeneous cubic invariant $g$
vanishes on all positive-definite trace-one matrices, then for
$X\succ0$ with $\lambda=\Tr X>0$ the matrix $\lambda^{-1}X$ is again
positive definite with trace one, so
$g(X)=\lambda^3g(\lambda^{-1}X)=0$. Thus $g$ vanishes on the open cone
of positive definite matrices, and a polynomial vanishing on a nonempty
open set vanishes identically.  Thus the restriction map is an
isomorphism onto its image, and the images of the basis
$h^3,hp_2,p_3,u,v$ are the five functions
\[
 1,\quad t_2=\Tr(X^2),\quad t_3=\Tr(X^3),\quad
 u=\Tr(PX^{\otimes3}),\quad v=\Tr(PSX^{\otimes3}),
\]
which therefore form a basis of the space of
restrictions of invariant polynomials of total degree at most three.
In particular no additional linear relation appears after restriction.

{\bf{Step 6.}}
This is the content of Appendix~\ref{app:independence}, where an
explicit four-parameter family of two-qutrit states is exhibited whose
Jacobian determinant for the map
$\rho\mapsto(t_2,t_3,u,v)$ equals
$71/774840978\neq0$. By the inverse function theorem, we know that the image contains
a nonempty open subset of $\mathbb{R}^4$, so no nonzero polynomial relation
among $t_2,t_3,u,v$ can hold.  This proves the last assertion for $n=2$,
and the appendix reduces $n>2$ to $n=2$ by tensoring with $I_3/3$.
\end{proof}

These four invariants do not separate Clifford orbits. Indeed,
Appendix~\ref{app:independence} gives a positive-definite state at which
the map $\rho\mapsto(t_2,t_3,u,v)$ has rank four. The submersion theorem
then gives a local level set of dimension $3^{2n}-5>0$ within the
trace-one positive-definite matrices. Each projective Clifford orbit is
finite, so this level set contains infinitely many distinct orbits.

\subsection{Weyl and Wigner representations of the nonspectral invariants}

Let $\rho$ be a density operator on $\HH$. Recall that the CSS code projector $P$ is as in \eqref{rep:css-weyl-average}. The Weyl characteristic function of $\rho$ is
\(\chi_\rho(x)=\Tr(W_x\rho)\) for all \(x\in V_n\).
Let $\{|z\rangle:z\in\mathbb F_3^n\}$ be the computational basis and
let the parity operator $J$  be defined by $J|z\rangle=|-z\rangle$,
negation being taken in $\mathbb F_3^n$.

\begin{lemma}\label{lem:J-traces}
The parity operator $J$ is a Hermitian involution with $\Tr(J)=1$, and
\( \Tr(W_xJ)=1\) for \(x\in V_n\).
\end{lemma}

\begin{proof}
The operator $J$ is the permutation matrix of $z\mapsto-z$ on $\mathbb F_3^n$, so it is unitary.
Since negation is an involution, $J^2=I$. Therefore $J^\dagger=J^{-1}=J$.
Its diagonal entries are $\delta_{z,-z}$, so
\[
 \Tr J=\#\{z\in\F_3^n:2z=0\}=1 ,
\]
since $2$ is invertible in $\F_3$ and therefore $2z=0$ forces $z=0$. For $x=(a,b)\in V_n$ the Weyl convention gives
\[
 W_{a,b}J|z\rangle=\omega^{-a\cdot z+a\cdot b/2}|b-z\rangle.
\]
A diagonal matrix element requires $z=b/2$, and its phase is then
one. There is exactly one such computational label, so $\Tr(W_xJ)=1$.
\end{proof}
For each phase-space point
$x\in V_n$, the phase-point operator and the Wigner function are defined, respectively, by (see \cite{gnw21})
\[
  A_x=W_xJW_x^\dagger,\quad
  \mathcal W_\rho(x)=3^{-n}\Tr(A_x\rho).
\]
By Lemma~\ref{lem:J-traces}, $J$ is Hermitian with trace one, and hence so is $A_x$.
Expanding an arbitrary operator $\mathcal{B}$ in the Weyl basis and applying the character orthogonality relation
$$
\sum_{x\in V_n}\omega^{[x,y]}=3^{2n}\delta_{y,0}
$$
yields the depolarizing identity
$$
3^{-2n}\sum_{x\in V_n}W_x\mathcal{B}W_x^\dagger
=3^{-n}\Tr(\mathcal{B})I.
$$
Taking $\mathcal{B}=J$ yields
\[
  \sum_{x\in V_n}A_x=3^n I.
\]
Thus $\mathcal W_\rho$ is real and
$\sum_{x\in V_n}\mathcal W_\rho(x)=1$. Its values need not be
nonnegative, so it is a quasiprobability distribution; see, for example, \cite{gross06,gnw21}.
\begin{proposition}\label{prop:weyl-wigner-moments}
Let $u$ and $v$ be as in \eqref{eq:four}. For every density operator $\rho$ and every $n\geq1$,
\begin{equation}\label{eq:weyl-moments}
  u=3^{-2n}\sum_{x\in V_n}\chi_\rho(x)^3,\ v=3^{-2n}\sum_{x\in V_n}
     \Tr(W_x\rho W_x\rho)\,\chi_\rho(x).
\end{equation}
Moreover, the following operator identity and moment formula
\begin{equation}\label{eq:wigner}
  PS=3^{-2n}\sum_{x\in V_n}A_x^{\otimes3},
  \quad
  v=3^n\sum_{x\in V_n}\mathcal W_\rho(x)^3
\end{equation}
hold, where $S=V_{(12)}$ exchanges the first two copies of $\HH$.
\end{proposition}
\begin{proof}
Expanding the $x$-sum in
\eqref{rep:css-weyl-average} and evaluating traces in
$\rho^{\otimes3}$ gives
\[
  u=3^{-2n}\sum_{x\in V_n}\Tr(W_x^{\otimes3}\rho^{\otimes3})
   =3^{-2n}\sum_{x\in V_n}\chi_\rho(x)^3 .
\]
 Since $S$ commutes with
$W_x^{\otimes3}$ and with $\rho^{\otimes3}$, it follows  that
\[
  v=3^{-2n}\sum_{x\in V_n}\Tr\bigl(S(W_x\rho)^{\otimes3}\bigr)
   =3^{-2n}\sum_{x\in V_n}\Tr(W_x\rho W_x\rho)\,\chi_\rho(x).
\]
Here, for the second identity, use the swap identity
\begin{equation*}
  \Tr\bigl(S(B_1\otimes B_2\otimes B_3)\bigr)=\Tr(B_1B_2)\Tr(B_3),
\end{equation*}
valid for all operators $B_1,B_2,B_3$ on $\HH$.

We first take $n=1$, so $P=p$ and $a,b,z\in\mathbb F_3$.
For $x=(a,b)$ the Weyl convention and
$J|z\rangle=|-z\rangle$ give
\begin{equation*}
  A_{(a,b)}|z\rangle=\omega^{2a(b-z)}|2b-z\rangle .
\end{equation*}
Write $\boldsymbol z=(z_1,z_2,z_3)$ and
$\boldsymbol1=(1,1,1)$. Averaging
$A_{(a,b)}^{\otimes3}$ with respect to $a$ and using
$$\sum_{a\in\mathbb F_3}\omega^{-2a(z_1+z_2+z_3)}
=3\delta_{z_1+z_2+z_3,0}$$
annihilates every
basis vector with $z_1+z_2+z_3\neq0$, and on the hyperplane
$z_1+z_2+z_3=0$ one obtains
\begin{equation}\label{eq:A-average}
  \frac19\sum_{a,b\in\mathbb F_3}A_{(a,b)}^{\otimes3}|\boldsymbol z\rangle
  =
  \begin{cases}
    \displaystyle\frac13\sum_{b\in\mathbb F_3}
      |-\boldsymbol z+2b\boldsymbol1\rangle,&z_1+z_2+z_3=0,\\[8pt]
    0,&z_1+z_2+z_3\neq0.
  \end{cases}
\end{equation}
Likewise, expanding $p$ for a single qutrit and using the same character
orthogonality yields
\begin{equation}\label{eq:P-action}
  P|\boldsymbol z\rangle
  =
  \begin{cases}
    \displaystyle\frac13\sum_{c\in\mathbb F_3}
      |\boldsymbol z+c\boldsymbol1\rangle,&z_1+z_2+z_3=0,\\[8pt]
    0,&z_1+z_2+z_3\neq0.
  \end{cases}
\end{equation}
On the hyperplane, we have
\begin{equation*}
  S\boldsymbol z=(z_2,z_1,z_3)=-\boldsymbol z-z_3\boldsymbol1 ,
  \quad z_1+z_2+z_3=0 .
\end{equation*}
From this, it follows that
\[
  PS|\boldsymbol z\rangle
  =\frac13\sum_{c\in\mathbb F_3}
     |-\boldsymbol z+(c-z_3)\boldsymbol1\rangle
  =\frac13\sum_{b\in\mathbb F_3}
     |-\boldsymbol z+2b\boldsymbol1\rangle,
\]
the last step using that both $c-z_3$ and $2b$ run through
$\mathbb F_3$. Outside the hyperplane both
\eqref{eq:A-average} and \eqref{eq:P-action} vanish, and $S$ preserves
the coordinate sum, so the two sides of the identity agree on every
basis vector. This proves $PS=3^{-2n}\sum_xA_x^{\otimes3}$ for $n=1$.
After regrouping the tensor factors by site, the Weyl operators,
phase point operators, and copy swap factorize as
$W_x=\bigotimes_jW_{x_j}$, $A_x=\bigotimes_jA_{x_j}$ and
$S=\bigotimes_jS_j$, where $S_j$ is the swap of the first two copies
at site $j$. Tensoring the one-site identity over the $n$ sites
therefore gives the operator identity in \eqref{eq:wigner}. It is an
operator identity and requires no product state assumption on $\rho$.
Taking expectations in $\rho^{\otimes3}$ and using
$\Tr(A_x^{\otimes3}\rho^{\otimes3})=(\Tr(A_x\rho))^3$
together with $\Tr(A_x\rho)=3^n\mathcal W_\rho(x)$ gives
\[v=3^{-2n}\sum_{x\in V_n}\bigl(\Tr(A_x\rho)\bigr)^3
   =3^n\sum_{x\in V_n}\mathcal W_\rho(x)^3.\]
This completes the proof of Proposition \ref{prop:weyl-wigner-moments}.
\end{proof}

\begin{remark}\label{rem:weyl-cubes}
\begin{itemize}
  \item[\rm(i)] The cubes in the Weyl expression for $u$ are ordinary
complex cubes, not absolute cubes, and individual summands need not be
real. Using $\rho^\dagger=\rho$ and $W_x^\dagger=W_{-x}$ one
first gets $\chi_\rho(-x)=\overline{\chi_\rho(x)}$, hence
$\chi_\rho(-x)^3=\overline{\chi_\rho(x)^3}$, and the terms at $x$ and
$-x$ combine into $2\operatorname{Re}\chi_\rho(x)^3$. The sum is
nonnegative because $u=\Tr(P\rho^{\otimes3})$ is the trace of the
product of the positive operators $P$ and $\rho^{\otimes3}$.
  \item[\rm(ii)] Note that $u$ and $v$ are invariant under Clifford conjugation. In fact, if $\mathcal{C}\in\Cl(n,3)$ acts by
$\mathcal{C}W_x\mathcal{C}^\dagger=\omega^{k(x)}W_{Mx}$ with $M\in\operatorname{Sp}(2n,3)$,
then, writing $\rho'=\mathcal{C}\rho \mathcal{C}^\dagger$,
\[
  \chi_{\rho'}(Mx)=\omega^{-k(x)}\chi_\rho(x),
\]
the cube absorbs the phase and conjugation by $\mathcal{C}$ merely permutes the
labels in the sum defining $u$. For $v$, the quadratic trace transforms as
\[
  \Tr(W_{Mx}\rho'W_{Mx}\rho')
  =\omega^{-2k(x)}\Tr(W_x\rho W_x\rho).
\]
Multiplication by $\chi_{\rho'}(Mx)=\omega^{-k(x)}\chi_\rho(x)$
cancels the phase, since $\omega^{-3k(x)}=1$.
Hence $v$ is invariant as well.
\end{itemize}
\end{remark}
The following inequalities give necessary conditions on the
quadruple $(t_2,t_3,u,v)$ associated with a density operator.
\begin{corollary}\label{cor:moment-constraints}
Let $t_2$, $t_3$, $u$ and $v$ be as in \eqref{eq:four}. For every density operator $\rho$ and every $n\geq1$, we have
\begin{equation*}
  3^{-n}\leq t_2\leq1,\quad
  t_2^2\leq t_3\leq t_2^{\frac32},\quad
  |v|\leq u\leq\frac{t_2}{3^n}.
\end{equation*}
If in addition $n\geq2$, then all five sector weights in
\eqref{eq:weights} are nonnegative, and $|v|\leq u$ is equivalent to
the pair of inequalities $w_A\geq0$, $w_C\geq0$.
\end{corollary}

\begin{proof}
Let $\lambda_i$ be the eigenvalues of $\rho$. Since $\lambda_i\geq0$
and $\sum_i\lambda_i=1$, Cauchy--Schwarz gives
$1\leq 3^n\sum_i\lambda_i^2=3^nt_2$, while $t_2\leq1$ is immediate. Next,
\[
  t_2^2=\Bigl(\sum_i\lambda_i^{\frac12}\lambda_i^{\frac32}\Bigr)^2
  \leq\Bigl(\sum_i\lambda_i\Bigr)\Bigl(\sum_i\lambda_i^3\Bigr)=t_3 ,
\]
and monotonicity of the $\ell^p$ norms of the eigenvalue vector gives
$t_3\leq t_2^{3/2}$.
Since $P$ is an orthogonal projection and $S$ is a Hermitian
involution commuting with $P$, both $I\pm S$ have spectrum in
$\{0,2\}$, thus
\(
  P\pm PS=P(I\pm S)\succeq0 .
\)
Taking expectations in $\rho^{\otimes3}$ yields $u\pm v\geq0$, that
is, $|v|\leq u$.
Finally, Weyl orthogonality and unitarity give
\[
  \sum_{x\in V_n}|\chi_\rho(x)|^2=3^nt_2,
  \quad |\chi_\rho(x)|\leq1 .
\]
By Remark~\ref{rem:weyl-cubes}, $u$ is real and
$u=3^{-2n}\sum_x\chi_\rho(x)^3$, whence
\[
  0\leq u
  =\Bigl|3^{-2n}\sum_{x\in V_n}\chi_\rho(x)^3\Bigr|
  \leq 3^{-2n}\sum_{x\in V_n}|\chi_\rho(x)|^3
  \leq 3^{-2n}\sum_{x\in V_n}|\chi_\rho(x)|^2
  =\frac{t_2}{3^n}.
\]
For $n\geq2$ the sector inequalities are
$w_j=\Tr(Q_j\rho^{\otimes3})\geq0$, and by
\eqref{eq:weights},
\[
  w_A=\frac{u+v}{2},\quad w_C=\frac{u-v}{2},
\]
so $|v|\leq u$ is precisely the pair $w_A\geq0$, $w_C\geq0$.
\end{proof}

\begin{remark}\label{rem:n1-exception}
The Weyl moment formulas \eqref{eq:weyl-moments}, the identities \eqref{eq:wigner},
and the polynomial sector weight formulas \eqref{eq:weights} remain
valid for $n=1$. In that case $Q_D=0$ and the other sector ranks are
$r_A=2$, $r_B=8$, $r_C=1$, and $r_E=16$.
The Clifford commutant has dimension seven because the orbit of linearly independent
isotropic pairs is absent. The algebra generated by $P$ and the
permutations is $\C^3\oplus M_2(\C)$, also of dimension seven, so it
is the full Clifford commutant. The joint commutant is therefore
$\C^4$, and the twirl formula holds after omitting the zero sector.
For $\rho=I/3$ the weights are
$(2/27,8/27,1/27,0,16/27)$, in agreement with \eqref{eq:weights}.
The decomposition into five nonzero sectors and formulas that divide by
$r_D$ require $n\geq2$; the polynomial weight identities do not.
\end{remark}

The Wigner moment formula gives an interpretation of the sign of
$v$. If $\mathcal W_\rho(x)\geq0$ for all $x\in V_n$, then
\[
  v=3^n\sum_{x\in V_n}\mathcal W_\rho(x)^3\geq0 .
\]
Therefore,
\[
  v<0
  \quad\Longrightarrow\quad
  \mathcal W_\rho(x)<0\ \text{ for some }x\in V_n .
\]
Thus a negative value of $v$ witnesses Wigner negativity. The converse
does not follow: positive contributions to the cubic moment may
outweigh the negative ones. The quantities $u$ and $v$ therefore
connect the sector description of the third Clifford moment with the
Weyl characteristic function and the Wigner function.
\subsection{Explicit density matrices}\label{rep:explicit-checks}
We first consider the maximally mixed state $\rho_{\mathrm{mix}}=I/3^n$. Then
\[
 \Tr(\rho_{\mathrm{mix}}^2)=3^{-n},\quad
 \Tr(\rho_{\mathrm{mix}}^3)=3^{-2n},\quad
 \Tr(P)=3^n,\quad \Tr(PS)=1,
\]
where the last two identities follow from
\eqref{rep:css-weyl-average} and the swap trace identity. Hence
\begin{equation*}
 (t_2,t_3,u,v)=\bigl(3^{-n},\,3^{-2n},\,3^{-2n},\,3^{-3n}\bigr).
\end{equation*}
Since $\rho_{\mathrm{mix}}^{\otimes3}=3^{-3n}I$ is proportional to the
identity on the tensor cube and $Q_j$ is an orthogonal projection, it follows that
\begin{equation*}
 w_j=\Tr\bigl(Q_j\rho_{\mathrm{mix}}^{\otimes3}\bigr)
 =\frac{r_j}{3^{3n}},
\end{equation*}
thus the sector weights are exactly the normalized sector dimensions.
\begin{proposition}\label{prop:rho-minus}
Let $\rho_-=\frac{I-J}{3^n-1}$,
where $J$ is as in Lemma \ref{lem:J-traces}. Then $\rho_-$ is a density operator of rank $(3^n-1)/2$ whose nonzero
eigenvalues all equal $2/(3^n-1)$. Its Weyl characteristic function and
Wigner function are
\begin{equation}\label{eq:rho-minus-chi}
 \chi_{\rho_-}(0)=1,\quad
 \chi_{\rho_-}(x)=-\frac1{3^n-1}\quad (x\neq0),
\end{equation}
\begin{equation}\label{eq:rho-minus-W}
 \mathcal W_{\rho_-}(0)=-\frac1{3^n},\quad
 \mathcal W_{\rho_-}(x)=\frac1{3^n(3^n-1)}\quad (x\neq0).
\end{equation}
Then
\begin{equation}\label{eq:negative-parity}
 t_2=\frac2{3^n-1},\quad
 t_3=\frac4{(3^n-1)^2},\quad
 u=-v=\frac{3^n-3}{3^n(3^n-1)^2}.
\end{equation}
\end{proposition}

\begin{proof}
By Lemma~\ref{lem:J-traces}, $I-J$ is twice the projection onto the
negative parity subspace, so it has rank $(3^n-1)/2$. Normalizing by
$3^n-1$ gives a density operator whose nonzero eigenvalues are
$2/(3^n-1)$. Then
\[
 t_2=\frac{3^n-1}{2}\Bigl(\frac{2}{3^n-1}\Bigr)^2=\frac2{3^n-1},
 \quad
 t_3=\frac{3^n-1}{2}\Bigl(\frac{2}{3^n-1}\Bigr)^3=\frac4{(3^n-1)^2}.
\]
Next, $\chi_{\rho_-}(x)=\bigl(\Tr W_x-\Tr(W_xJ)\bigr)/(3^n-1)$. Since
 $\Tr W_x=3^n\delta_{x,0}$ and $\Tr(W_xJ)=1$, we obtain
\[
 \chi_{\rho_-}(0)=1,\quad
 \chi_{\rho_-}(x)=-\frac1{3^n-1}\quad(x\neq0),
\]
which is \eqref{eq:rho-minus-chi}. Substituting into the first Weyl
moment formula, we obtain
\[
 u=3^{-2n}\sum_{x\in V_n}\chi_{\rho_-}(x)^3
 =\frac1{3^{2n}}\Bigl(1-\frac{3^{2n}-1}{(3^n-1)^3}\Bigr)
 =\frac{3^n-3}{3^n(3^n-1)^2}.
\]
For $v$, expanding $\rho_-$ and using $JW_xJ=W_{-x}$ gives
\[
\Tr(W_x\rho_-W_x\rho_-)
=\frac{3^n\delta_{x,0}+3^n-2}{(3^n-1)^2}.
\]
Hence
\[
 v=3^{-2n}\sum_{x\in V_n}\Tr(W_x\rho_-W_x\rho_-)\chi_{\rho_-}(x)
 =-\frac{3^n-3}{3^n(3^n-1)^2}.
\]
For the Wigner function, $$\mathcal W_{\rho_-}(x)
=\bigl(\Tr A_x-\Tr(A_xJ)\bigr)/\bigl(3^n(3^n-1)\bigr).$$
Here $\Tr A_x=\Tr J=1$ and $A_xJ=W_{2x}$, so
$\Tr(A_xJ)=3^n\delta_{x,0}$. These identities give
\eqref{eq:rho-minus-W}.
Finally,
\[
3^n\sum_{x\in V_n}\mathcal W_{\rho_-}(x)^3
 =-\frac1{3^{2n}}+\frac{3^n+1}{3^{2n}(3^n-1)^2}
 =-\frac{3^n-3}{3^n(3^n-1)^2},
\]
in agreement with the value of $v$ above.
\end{proof}

\begin{example}\label{EX}
For $n=2$, \eqref{eq:negative-parity} gives
\(t_2 =\tfrac14\), \(t_3=\tfrac1{16}\),
\(u=-v=\tfrac1{96},\)
and substituting into \eqref{eq:weights} yields the sector weights
\begin{equation}\label{eq:D9-weights}
 (w_A,w_B,w_C,w_D,w_E)
 =\Bigl(0,\ \tfrac5{16},\ \tfrac1{96},\ \tfrac5{96},\ \tfrac58\Bigr),
 \quad \sum_jw_j=1 .
\end{equation}
For $n\geq2$ one has $v<0$, so $\rho_-$ is an explicit mixed state with a
negative cubic Wigner moment. By \eqref{eq:wigner}, this witnesses
Wigner negativity of $\mathcal W_{\rho_-}$.
\end{example}

\begin{proposition}\label{prop:depolarized}
For a unit vector $|\psi\rangle\in\HH$, let
\begin{equation*}
 \rho_\lambda
 =\lambda|\psi\rangle\langle\psi|+(1-\lambda)\frac{I}{3^n},
 \quad 0\leq\lambda\leq1,
\end{equation*}
and write $\tau=|\psi\rangle\langle\psi|$ and
$q=\Tr\bigl(P\tau^{\otimes3}\bigr)$.
Then
\begin{align}
 t_2&=\lambda^2+\frac{1-\lambda^2}{3^n},\nonumber\\
 t_3&=\lambda^3+\frac{3\lambda^2(1-\lambda)}{3^n}
     +\frac{(1-\lambda)^2(1+2\lambda)}{3^{2n}},\nonumber\\
 u&=\lambda^3q+\frac{1-\lambda^3}{3^{2n}},\label{eq:dep-u}\\
 v&=\lambda^3q+\frac{3\lambda^2(1-\lambda)}{3^{2n}}
     +\frac{(1-\lambda)^2(1+2\lambda)}{3^{3n}}.\label{eq:dep-v}
\end{align}
\end{proposition}

\begin{proof}
Write $\rho_\lambda=\lambda\tau+(1-\lambda)3^{-n}I$ and expand the
powers, using $\tau^k=\tau$ and $\operatorname{Tr}\tau=1$.
For the two moment formulas we need the partial traces of the
one-site operators
\begin{equation*}
 \Tr_3(p)=\frac{I\otimes I}{3},\quad
 \Tr_{23}(p)=I,\quad
 \Tr_3(ps)=\frac{s_{12}}{3},\quad
 \Tr_{23}(ps)=\frac{I}{3},
\end{equation*}
where $s_{12}$ swaps two qutrits and $s=s_{12}\otimes I$ acts on
three qutrits; $\Tr_3$ traces out the third copy and $\Tr_{23}$ traces out
the last two. Tensoring over the sites gives the corresponding
identities for $P$ and $PS$.
For $u$ we expand $\rho_\lambda^{\otimes3}$ into eight terms. A term
with $j$ factors of $\tau$ contributes
\[
 \lambda^j(1-\lambda)^{3-j}
 \Tr\bigl(P\,\tau^{\otimes j}\otimes(3^{-n}I)^{\otimes(3-j)}\bigr).
\]
For $j=3$ this is $\lambda^3q$. For $j=2$, tracing out the maximally mixed factor
gives a factor $3^{-2n}$ together with the sum over the three positions
of the maximally mixed factor, contributing $3\lambda^2(1-\lambda)3^{-2n}$. For $j=1$ all
three positions contribute the same factor $3^{-2n}$, giving
$3\lambda(1-\lambda)^23^{-2n}$. Finally $j=0$ contributes
$(1-\lambda)^33^{-2n}$. Thus the terms with $j<3$ sum to
\[
 \frac{3\lambda^2(1-\lambda)+3\lambda(1-\lambda)^2+(1-\lambda)^3}
 {3^{2n}}
 =\frac{1-\lambda^3}{3^{2n}}.
\]
Adding $\lambda^3q$ proves \eqref{eq:dep-u}.
For $v$, the same expansion is applied to $PS$ rather than to $P$. The
cases $j=3$ and $j=0$ contribute $\lambda^3q$ and
$(1-\lambda)^33^{-3n}$; the case $j=2$ contributes
$3\lambda^2(1-\lambda)3^{-2n}$; and the case $j=1$ contributes
$3\lambda(1-\lambda)^23^{-3n}$. Combining the $j=1$ and $j=0$
contributions gives \eqref{eq:dep-v}.
\end{proof}

\begin{remark}
\label{rem:multiplicativity}
If $\rho=\rho_1\otimes\rho_2$ is a tensor product across disjoint sets
of qutrits, then
\[
 P=P_1\otimes P_2,
 \quad
 PS=(P_1S_1)\otimes(P_2S_2),
\]
and therefore
\begin{equation*}
 u(\rho)=u(\rho_1)\,u(\rho_2),
 \quad
 v(\rho)=v(\rho_1)\,v(\rho_2).
\end{equation*}
This follows from the tensor product structure of the Weyl operators,
the parity operator, and the copy permutation.
\end{remark}

\begin{remark}\label{rem:n1-checks}
The Weyl moment formulas \eqref{eq:weyl-moments}, the moment formulas
\eqref{eq:negative-parity}, Proposition~\ref{prop:depolarized}, and
the polynomial identities \eqref{eq:weights} hold for every $n\geq1$.
The numerical table \eqref{eq:D9-weights} concerns $n=2$.
For $n=1$, $\rho_-$ has rank one and $u=v=0$, so its cubic
Wigner moment does not witness its Wigner negativity.
\end{remark}
\section{Pure tensor cubes}\label{sec:pure}
Throughout this section, $n\geq1$. Unless stated otherwise,
$\rho=|\psi\rangle\langle\psi|$ is a pure state on $\HH$.
We use the notation of Section~\ref{sec:arbitrary}. In particular, $t_2,t_3,u,v$ are as in
\eqref{eq:four}. Numbered references to \cite{zmy24} in this section
use the numbering of arXiv:2410.13575v1.
\begin{proposition}\label{prop:pure-reduction}
For a pure state $\rho=|\psi\rangle\langle\psi|$,
\begin{equation}\label{eq:pure}
 t_2=t_3=1,\quad u=v=q:=\Tr\bigl(P\rho^{\otimes3}\bigr)
\end{equation}
and
\begin{equation}\label{eq:pure-weights}
 (w_A,w_B,w_C,w_D,w_E)=(q,\,1-q,\,0,\,0,\,0),\quad 0\leq q\leq1 .
\end{equation}
For a general density operator $\rho$ and the input $\rho^{\otimes3}$,
$w_E=0$ holds if and only if $\rho$ is pure.
\end{proposition}
\begin{proof}
Let $\rho=|\psi\rangle\langle\psi|$ be pure, write
$|\Psi\rangle=|\psi\rangle^{\otimes3}$ and
$\rho^{\otimes3}=|\Psi\rangle\langle\Psi|$. Since $\rho^k=\rho$ for every
$k\geq1$, it follows that
\[
 t_2=\Tr(\rho^2)=1,\quad t_3=\Tr(\rho^3)=1,\quad
 u=\Tr(P\rho^{\otimes3})=q
\]
by the definition of $q$ in \eqref{eq:pure}. The transposition
$S=V_{(12)}$ permutes the three tensor factors, and $|\Psi\rangle$ is
invariant under every such permutation. Thus $S|\Psi\rangle=|\Psi\rangle$, which
implies that
\[v=\Tr(PS\rho^{\otimes3})=\Tr(P\rho^{\otimes3})=u=q .\]
Substituting $t_2=t_3=1$ and $u=v=q$ into \eqref{eq:weights} gives
 \eqref{eq:pure-weights}.
The bound $0\leq q\leq1$ holds because
$P$ is a projection and $\rho^{\otimes3}$ is a state.

It remains to prove the equivalence with purity. If $\rho$ is pure, then
$t_3=1$ and the last line of \eqref{eq:weights} gives $w_E=0$. Conversely,
\eqref{eq:weights} is stated for every density operator, so
$w_E=2(1-t_3)/3$ in general, and $w_E=0$ is equivalent to $t_3=1$. Let
$\lambda_1,\dots,\lambda_{3^n}$ be the eigenvalues of $\rho$. Then $\lambda_i\geq0$ and $\sum_i\lambda_i=\Tr(\rho)=1$,
while $\lambda_i\leq1$ because $\rho\preceq I$. Thus
$\lambda_i^3\leq\lambda_i^2\leq\lambda_i$ for every $i$  and
\[
 t_3=\sum_{i=1}^{3^n}\lambda_i^3
 \ \leq\ t_2=\sum_{i=1}^{3^n}\lambda_i^2
 \ \leq\ \sum_{i=1}^{3^n}\lambda_i=1 .
\]
If $w_E=0$, then $t_3=1$ and both inequalities are equalities. The
equality $\sum_i\lambda_i^2=\sum_i\lambda_i$ forces
$\lambda_i^2=\lambda_i$ for every $i$. Together with $\sum_i\lambda_i=1$ this shows
that exactly one eigenvalue equals one and all the others vanish, so
$\rho$ is a rank-one projection, that is, a pure state.
\end{proof}

The range below recovers
\cite[Proposition~4 and Lemma~20]{zmy24}; see also
Remark~\ref{rem:overlap-zmy} for the normalization and the family
attaining all values.
\begin{theorem}\label{thm:range}
For every pure $n$-qutrit state $\psi$, we have
$ 0\leq q\leq\frac1{3^n} .$
Every value in $[0,1/3^n]$ occurs, and equality $q=1/3^n$ holds if and
only if $\psi$ is a stabilizer state.
\end{theorem}

\begin{proof}
Write $\rho=|\psi\rangle\langle\psi|$ and
\begin{align}\label{Echi}
\chi_x=\langle\psi|W_x|\psi\rangle
\end{align}
for $x\in V_n$. Two ingredients are
used repeatedly. The Weyl expansion of the code projector $P$ in
\eqref{rep:css-weyl-average} gives
\begin{equation}\label{eq:range-q-Weyl}
 q=\Tr\bigl(P\rho^{\otimes3}\bigr)
 =3^{-2n}\sum_{x\in V_n}\chi_x^3 ,
\end{equation}
the sum being real because it equals $3^{2n}q$. The Weyl
operators satisfy the dual basis identity
$\sum_{x\in V_n}W_x\otimes W_x^\dagger=3^n\mathrm{SWAP}$, and
$\Tr\bigl(\mathrm{SWAP}(\rho\otimes\rho)\bigr)=\Tr(\rho^2)$, where
\[
\mathrm{SWAP}(|a\rangle\otimes|b\rangle):=|b\rangle\otimes|a\rangle.
\]
Therefore
\begin{equation}\label{eq:purity-parseval}
 \sum_{x\in V_n}|\chi_x|^2
 =\sum_{x\in V_n}\Tr(W_x\rho)\Tr(W_x^\dagger\rho)
 =3^n\Tr(\rho^2)=3^n ,
\end{equation}
where the last equality uses purity. Moreover, $W_x$ is unitary, so
$|\chi_x|\leq1$ for every $x$.
Since $P\succeq0$ and
$\rho^{\otimes3}\succeq0$, we have $q\geq0$. Because the sum in
\eqref{eq:range-q-Weyl} is real, $\mathrm{Re}(z^3)\leq|z|^3$ and
$|\chi_x|^3\leq|\chi_x|^2$ for $|\chi_x|\leq1$,
\begin{equation}\label{eq:q-bound}
 3^{2n}q=\sum_{x\in V_n}\operatorname{Re}\bigl(\chi_x^3\bigr)
 \leq\sum_{x\in V_n}|\chi_x|^3
 \leq\sum_{x\in V_n}|\chi_x|^2=3^n ,
\end{equation}
which is $q\leq1/3^n$.

 Suppose
$q=1/3^n$. Then both inequalities of \eqref{eq:q-bound} are equalities
term by term. Equality in the first one gives
$\operatorname{Re}(\chi_x^3)=|\chi_x|^3$ for every $x$, and equality in
the second gives $|\chi_x|^3=|\chi_x|^2$ for every $x$. Writing
$\chi_x=r_xe^{i\varphi_x}$ with $r_x=|\chi_x|$, these conditions give
$r_x\in\{0,1\}$ and, when $r_x=1$, $\cos(3\varphi_x)=1$.
Thus
$\chi_x\in\{0,1,\omega,\omega^2\}$ for every $x\in V_n$.
By \eqref{eq:purity-parseval}, the set
$M=\{x\in V_n:\,|\chi_x|=1\}$
has exactly $3^n$ elements. For $x\in M$, equality in the
Cauchy--Schwarz inequality gives
$W_x|\psi\rangle=\chi_x|\psi\rangle$.
If $x,y\in M$, the scalar commutator relation
$W_xW_y=\omega^{[x,y]}W_yW_x$, applied to $|\psi\rangle$, forces
$[x,y]=0$. The Weyl multiplication rule then gives
\[
 W_{x+y}|\psi\rangle=W_xW_y|\psi\rangle
 =\chi_x\chi_y|\psi\rangle.
\]
Thus $x+y\in M$. Since $0\in M$ and the field is $\F_3$, the set
$M$ is an isotropic linear subspace. Its order is $3^n$, so it is
maximal isotropic. The joint eigenspaces of the Weyl operators labelled
by $M$ are one dimensional and consist of stabilizer states; hence
$\psi$ is a stabilizer state.
Conversely, let $\psi$ be a
stabilizer state and $M\subset V_n$ its maximal isotropic label
set. Then $W_x|\psi\rangle=c_x|\psi\rangle$ with $c_x\in\{1,\omega,
\omega^2\}$ for $x\in M$ and $|M|=3^n$. For $x\in M$ we have
$\chi_x=c_x$, so $|\chi_x|=1$ and $\chi_x^3=1$, hence
$\operatorname{Re}(\chi_x^3)=1=|\chi_x|^3$ and $|\chi_x|^3=|\chi_x|^2$.
For $x\notin M$ we claim $\chi_x=0$. Since $M$ is maximal isotropic, its
symplectic complement is $M$ itself, so there is $y\in M$ with
$[x,y]\neq0$. Then
\[
 W_y\bigl(W_x|\psi\rangle\bigr)=\omega^{[y,x]}W_xW_y|\psi\rangle
 =\omega^{[y,x]}c_y\bigl(W_x|\psi\rangle\bigr),
\]
while $W_y|\psi\rangle=c_y|\psi\rangle$. The two eigenvalues
$\omega^{[y,x]}c_y$ and $c_y$ differ because $\omega^{[y,x]}\neq1$, so
the eigenvectors $W_x|\psi\rangle$ and $|\psi\rangle$ of the unitary
$W_y$ are orthogonal and $\chi_x=0$. Thus every term of
\eqref{eq:q-bound} is an equality in this case, and
\[
 3^{2n}q=\sum_{x\in V_n}|\chi_x|^3=\sum_{x\in M}1=3^n ,
\]
that is, $q=1/3^n$.

 For $0\leq\theta\leq\pi/4$, we let
\begin{equation}\label{eq:family-theta}
 |\psi_\theta\rangle
 =\bigl(\cos\theta\,|0\rangle-\sin\theta\,|1\rangle\bigr)
  \otimes|0^{\otimes(n-1)}\rangle ,
 \quad
 |a_\theta\rangle=\cos\theta\,|0\rangle-\sin\theta\,|1\rangle .
\end{equation}
The following construction is given in
\cite[Section~X.C]{zmy24}.
The code states factor over the $n$ sites:
\[
\bigl|\phi_{\boldsymbol x}\bigr\rangle
=\bigotimes_{i=1}^{n}|\phi_{x_i}\rangle,
\quad P=\bigotimes_{i=1}^{n}P_i,
\]
where $P_i=\sum_{x\in\F_3}|\phi_x\rangle\langle\phi_x|$ acts on the
three copies at site $i$. For a product state this gives the factorization
\begin{equation}\label{eq:range-factorization}
 q(\psi_\theta)=p(a_\theta)\,p(|0\rangle)^{n-1},\quad
 p(a)=\langle a^{\otimes3}|P_i|a^{\otimes3}\rangle
 =\sum_{x\in\F_3}\bigl|\langle\phi_x|a^{\otimes3}\rangle\bigr|^2 .
\end{equation}
On one site,
\[
\langle\phi_x|a^{\otimes3}\rangle
=3^{-1/2}\sum_{y\in\F_3}a_{x+y}a_{y-x}a_y,
\]
with $a_0=\cos\theta$,
$a_1=-\sin\theta$ and $a_2=0$. Then
\[
\left\{
  \begin{aligned}
  & \sum_y a_y^3=\cos^3\theta-\sin^3\theta,\  \text{if}\  x=0, \\
   & a_1a_2a_0+a_2a_0a_1+a_0a_1a_2=0,\  \text{if}\  x=1,\\
   & a_2a_1a_0+a_0a_2a_1+a_1a_0a_2=0,\ \text{if}\  x=2.
  \end{aligned}
\right.
 \]
Hence $p(a_\theta)=(\cos^3\theta-\sin^3\theta)^2/3$. For the
state $|0\rangle$, the only nonvanishing overlap is
$\langle\phi_0|000\rangle=3^{-1/2}$. Then $p(|0\rangle)=1/3$. Substituting
into \eqref{eq:range-factorization} gives
\begin{equation}\label{eq:family}
 q(\psi_\theta)=\frac{(\cos^3\theta-\sin^3\theta)^2}{3^n} ,
\end{equation}
which equals $1/3^n$ at $\theta=0$ and $0$ at $\theta=\pi/4$. The right
side of \eqref{eq:family} is continuous on $[0,\pi/4]$, and on the open
interval its derivative is
\[
 -\frac6{3^n}\bigl(\cos^3\theta-\sin^3\theta\bigr)\sin\theta\cos\theta
 \bigl(\cos\theta+\sin\theta\bigr)<0 ,
\]
so it is strictly decreasing there. Hence
$\theta\mapsto q(\psi_\theta)$ maps $[0,\pi/4]$ bijectively onto
$[0,1/3^n]$, and every value in that interval occurs.
\end{proof}
\begin{remark}\label{rem:overlap-zmy}
The bounds of $q$ agree with
\cite[Proposition~4 and Lemma~20]{zmy24}. In the qutrit case, the
expectations of the two noninvertible commutant operators used there
equal $3^nq$. The family \eqref{eq:family-theta} and the resulting
construction of $3$-designs also appear in
\cite[Section~X.C]{zmy24}. These results are included for completeness.
\end{remark}
\subsection{A robust stabilizer-fidelity bound}

Let $F_{\mathrm{stab}}(\psi)=\max_T|\langle T|\psi\rangle|^2$, where
$T$ ranges over pure stabilizer states of $n$ qutrits, and put
$K=2+\sqrt2$. The proof below follows the use of commuting Weyl operators
in the stabilizer testing argument of \cite[Section~3.2]{gnw21}.

\begin{proposition}\label{prop:robust}
For every pure $n$-qutrit state $\psi$, we have
\begin{equation}\label{eq:robust}
 F_{\mathrm{stab}}(\psi)\geq1-K\,(1-3^nq).
\end{equation}
In particular, if $F_{\mathrm{stab}}(\psi)\leq1-\zeta$ then
$q\leq(1-\zeta/K)/3^n$.
\end{proposition}

\begin{proof}
We apply the strategy described above to the deficit $1-3^nq$ by
identifying a commuting set of Weyl operators with large expectation
values and bounding the probability mass outside this set. Let $\chi_x$ be as in \eqref{Echi} and
$p_x=|\chi_x|^2/3^n$. Then, by \eqref{eq:purity-parseval}, the numbers $p_x$
form a probability distribution on $V_n$.

The proof has three steps.

\textbf{Step 1.} If $W_x,W_y$ do not commute,
a Clifford unitary implementing a suitable symplectic transformation
sends their labels to those of $Z$ and $X$ or $X^{-1}$ acting on the
first qutrit. Since $X$ and $X^{-1}$ have conjugate expectation values,
it suffices to consider the
first qutrit reduced state $\tau$. The Weyl Parseval identity
gives
\[
1+2|\Tr(Z\tau)|^2+2|\Tr(X\tau)|^2
\leq\sum_{a,b\in\mathbb F_3}|\Tr(W_{(a,b)}\tau)|^2
=3\Tr(\tau^2)\leq3.
\]
The factors on the remaining sites are identities, and Clifford
phases do not change the absolute values. Then
$|\chi_x|^2+|\chi_y|^2\leq1$, where $
|\chi_x|=|\operatorname{Tr}(Z\tau)|$ and
$|\chi_y|=|\operatorname{Tr}(X\tau)|$. Set
\begin{equation*}
 H=\{x\in V_n:|\chi_x|^2>1/2\}.
\end{equation*}
Then the Weyl operators labelled by elements of \(H\) commute pairwise.
 Its linear span is isotropic
and hence is contained in a maximal isotropic subspace $M$ of $V_n$.

\textbf{Step 2.} Since
$3^{2n}q=\sum_x\operatorname{Re}\chi_x^3$ and
$\sum_x|\chi_x|^2=3^n$, it follows that
\begin{equation*}
1-3^nq
 =\frac1{3^n}\sum_{x\in V_n}\Bigl(|\chi_x|^2-\operatorname{Re}\chi_x^3\Bigr)
 =\sum_{\substack{x\in V_n\\\chi_x\ne0}}
 p_x\left(1-\frac{\operatorname{Re}\chi_x^3}{|\chi_x|^2}\right).
\end{equation*}
The omitted terms vanish because $p_x=0$ when $\chi_x=0$. For
$|\chi_x|\leq1$ one has $\operatorname{Re}\chi_x^3\leq|\chi_x|^3$, so
each summand is at least $p_x(1-|\chi_x|)$. Moreover, for
$x\notin H$ one has $|\chi_x|\leq1/\sqrt2$, and for every complex
number $w$ with $|w|\leq1/\sqrt2$,
\begin{equation*}
 |w|^2-\operatorname{Re}(w^3)
 \geq\bigl(1-\tfrac1{\sqrt2}\bigr)|w|^2 .
\end{equation*}
In fact, writing $r=|w|\leq1/\sqrt2$, one has
$\operatorname{Re}(w^3)\leq r^3$, and hence
$r^2-\operatorname{Re}(w^3)\geq r^2(1-r)
\geq(1-1/\sqrt2)r^2$. It follows that
\begin{equation}\label{eq:Delta-bound}
1-3^nq\geq\bigl(1-\tfrac1{\sqrt2}\bigr)\sum_{x\notin H}p_x .
\end{equation}

\textbf{Step 3.} Let
$\{|T_j\rangle\}_{j=1}^{3^n}$ be the common eigenbasis of the commuting
Weyl operators labelled by $M$, and put
$r_j=|\langle T_j|\psi\rangle|^2$. Each $T_j$ is a stabilizer state.
Finite Fourier orthogonality on $M$ gives
\begin{equation*}
 \sum_{x\in M}|\chi_x|^2=3^n\sum_{j=1}^{3^n}r_j^2.
\end{equation*}
By this, we obtain
\[
 \sum_{j=1}^{3^n}r_j^2=\frac1{3^n}\sum_{x\in M}|\chi_x|^2
 =\sum_{x\in M}p_x\geq\sum_{x\in H}p_x .
\]
Since a probability distribution satisfies
$\max_jr_j\geq\sum_jr_j^2$, it follows that
\begin{equation*}
 F_{\mathrm{stab}}(\psi)\geq\sum_{j=1}^{3^n}r_j^2\geq\sum_{x\in H}p_x
 =1-\sum_{x\notin H}p_x ,
\end{equation*}
which, together with \eqref{eq:Delta-bound}, implies that
\[
 1-F_{\mathrm{stab}}(\psi)
 \leq\sum_{x\notin H}p_x
 \leq\frac{1-3^nq}{1-1/\sqrt2}=K(1-3^nq).
\]
This is
\eqref{eq:robust}. If $F_{\mathrm{stab}}\leq1-\zeta$, combining this
with $1-K(1-3^nq)\leq F_{\mathrm{stab}}$ gives
$1-3^nq\geq\zeta/K$, proving the final assertion.
\end{proof}
\begin{remark}
Proposition~\ref{prop:robust} shows that a pure state whose value of $3^nq$ is close to one has high fidelity with a stabilizer state.
\end{remark}
\subsection{Full separability of symmetric invariant states}

For $0\leq q\leq1$ define the density operator
\begin{equation*}
 \Omega(q)=\frac{q}{R_A}Q_A+\frac{1-q}{R_B}Q_B ,
\end{equation*}
where $R_A=\operatorname{rank}Q_A$ and $R_B=\operatorname{rank}Q_B$.

\begin{theorem}\label{thm:separable}
The state $\Omega(q)$ is fully separable across the three copies
$\HH:\HH:\HH$ if and only if $q\leq1/3^n$. More generally, every fully
separable three-copy state $\sigma$ satisfies
$\Tr(P\sigma)\leq1/3^n$.
\end{theorem}

\begin{proof}
Let $\rho_1\otimes\rho_2\otimes\rho_3$ be a product state. The Weyl expansion \eqref{rep:css-weyl-average} gives
\[
 \Tr\bigl(P(\rho_1\otimes\rho_2\otimes\rho_3)\bigr)
 =3^{-2n}\sum_{x\in V_n}\prod_{i=1}^3\Tr(W_x\rho_i).
\]
H\"older's inequality with exponents $(3,3,3)$ yields
\[
 \Bigl|\sum_{x\in V_n}\prod_{i=1}^3\Tr(W_x\rho_i)\Bigr|
 \leq\prod_{i=1}^3\Bigl(\sum_{x\in V_n}|\Tr(W_x\rho_i)|^3\Bigr)^{1/3}.
\]
For each $i$, $|\Tr(W_x\rho_i)|\leq1$, and Weyl orthogonality gives
\[
 \sum_x|\Tr(W_x\rho_i)|^3
 \leq\sum_x|\Tr(W_x\rho_i)|^2
 =3^n\Tr(\rho_i^2)\leq3^n.
\]
Thus,
\[
 \Tr\bigl(P(\rho_1\otimes\rho_2\otimes\rho_3)\bigr)
 \leq3^{-2n}\cdot3^n=1/3^n.
\]
Convexity extends the bound to all fully separable states, and
$\Tr(P\Omega(q))=q$ gives the necessity.

Conversely, suppose $0\leq q\leq1/3^n$. By Theorem~\ref{thm:range}, there
is a pure state $\psi$ such that
\[\Tr(P(|\psi\rangle\langle\psi|)^{\otimes3})=q.\]
The weights in \eqref{eq:pure-weights} show that the third Clifford
moment of this pure state is $\Omega(q)$. Hence
\[
 \Omega(q)=\mathbb E_{\mathcal C}
 \bigl[(\mathcal C|\psi\rangle\langle\psi|\mathcal C^\dagger)^{\otimes3}\bigr],
\]
where the average is uniform over the projective Clifford group. Each
summand is a product of the same state on the three copies and is
therefore fully separable, and the set of fully separable states is
convex.
\end{proof}

\begin{remark}\label{rem:separability-copies}
\begin{itemize}
  \item[\rm(i)] By Theorem~\ref{thm:range}, the states obtained by
Clifford twirling pure tensor cubes are exactly $\Omega(q)$ with
$0\leq q\leq3^{-n}$.
For the witness $I/3^n-P$, pass to the more general product
$\sigma_1\otimes\sigma_2\otimes\sigma_3$ of states of the three copies.
As shown in the proof of Theorem~\ref{thm:separable}, its expectation of
$P$ is at most $3^{-n}$, and every fully separable state, being a convex
combination of such products, shares that bound. Hence the operator
\(I/3^n-P\) detects failure of full separability for arbitrary input $\sigma$ whenever its expectation is negative.
A negative value of $\Tr((I/3^n-P)\sigma)$ therefore certifies that
$\sigma$ is not fully separable. For the family the value is
explicit,
\(\Tr ((I/3^n-P)\Omega(q))=3^{-n}-q ,\)
so every member with $q>3^{-n}$ is detected.
  \item[\rm(ii)] Here separability refers to the
division into the three copies, each regarded as a single system, and not
to the division of the physical sites inside a copy. Full separability
means that
\[
 \Omega(q)=\sum_kp_k\,
 \sigma_k^{(1)}\otimes\sigma_k^{(2)}\otimes\sigma_k^{(3)},
 \quad p_k\geq0,\quad \sum_kp_k=1,
\]
where each density operator $\sigma_k^{(i)}$ acts on the whole $i$th
copy, so that entanglement among the physical sites within a copy is
allowed. For $n=2$ the two sites of a copy are denoted by $A$ and $B$,
and the division is $(A_1B_1):(A_2B_2):(A_3B_3)$. The bipartite state
$|\psi\rangle_{AB}=\frac{1}{\sqrt3}(|00\rangle+|11\rangle+|22\rangle)$ is
entangled across $A:B$ inside one copy, while
$|\psi\rangle_{A_1B_1}\otimes|\psi\rangle_{A_2B_2}
\otimes|\psi\rangle_{A_3B_3}$ is a product state across the three copies.
The example only illustrates the division. It is not needed for the
bounds above.
  \item[\rm(iii)] The operator acting on the three copies is the
tensor cube $\mathcal{C}^{\otimes3}=\mathcal{C}\otimes \mathcal{C}\otimes \mathcal{C}$, so every term
$(\mathcal{C}\rho \mathcal{C}^\dagger)^{\otimes3}$ of the average is a product state across
the three copies and the convex average
\(\mathbb E_{\mathcal C}[(\mathcal{C}\rho \mathcal{C}^\dagger)^{\otimes3}]\)
is fully separable across the three copies. Every tensor cube is
invariant under conjugation by copy permutations,
$V_\pi\rho^{\otimes3}V_\pi^\dagger=\rho^{\otimes3}$, so its Clifford twirl
agrees with the joint twirl computed in Theorem~\ref{thm:mixed}. For a
pure tensor cube that twirl belongs to the family $\Omega(q)$ and
realizes exactly the range $0\leq q\leq3^{-n}$. For a mixed
$\rho$ it does not, because
\(
 w_E=\tfrac23\bigl(1-\Tr(\rho^3)\bigr)>0 ,
\)
whereas every member of the family has $w_E=0$. The twirl of a mixed
tensor cube is therefore fully separable across the three copies and yet
lies outside the family.
\end{itemize}
\end{remark}

\subsection{Complex projective 3-designs}

Let $s=\dim\Sym^3\HH$ be as in \eqref{rep:classical-dimensions}, so that
$s=R_A+R_B$ because $\Pi_s=Q_A+Q_B$. The third moment of Haar-random pure states is
the normalized projector $\Pi_s/s$ onto the symmetric subspace.

Proposition~\ref{prop:design} records \cite[Proposition~6]{zmy24} in the
notation of this paper.

\begin{proposition}\label{prop:design}
The $A$-weight of the Haar third moment is
\begin{equation}\label{eq:design}
 q_*=\frac{R_A}{s}=\frac3{3^n(3^n+2)} .
\end{equation}
For a pure state $\psi$, the uniform ensemble on its projective Clifford
orbit is a complex projective $3$-design if and only if $q=q_*$. Moreover,
\begin{align}
 &\Dtr\!\Bigl(\Omega(q),\frac{\Pi_s}{s}\Bigr)
=|q-q_*|,\label{eq:design-distance}\\
 &\Tr\bigl(\Omega(q)^2\bigr)
=\frac{q^2}{R_A}+\frac{(1-q)^2}{R_B}
 =\frac1s+\frac{s}{R_AR_B}(q-q_*)^2 .\label{eq:frame}
\end{align}
\end{proposition}
\begin{proof}
By Theorem~\ref{thm:mixed}, Remark~\ref{rem:n1-exception}, and the
pure-state identities~\eqref{eq:pure}, this twirl equals $\Omega(q)$ with
$q=\Tr(P\rho^{\otimes3})$.
A complex projective $3$-design is, by
definition, an ensemble whose third moment equals the Haar third moment.
Hence the orbit is a $3$-design if and only if $\Omega(q)=\frac{\Pi_s}{s}$, that is, if and only if \({q}/{R_A}=1/s\) and
\( (1-q)/{R_B}=1/s\). The criterion is therefore the single
condition $q=q_*$, where
\[ q_*=\frac{R_A}{s} =\frac3{3^n(3^n+2)} ,\]
which is \eqref{eq:design}. The $A$-weight of $\Pi_s/s$ is
$\Tr(Q_A\Pi_s)/s=\Tr(Q_A)/s=R_A/s$, because
$Q_A\Pi_s=Q_A$.  Hence
\[
 \bigl\|\Omega(q)-\Pi_s/s\bigr\|_1
 =R_A\Bigl|\frac q{R_A}-\frac1s\Bigr|
 +R_B\Bigl|\frac{1-q}{R_B}-\frac1s\Bigr|
 =|q-q_*|+|q_*-q|=2|q-q_*| .
\]
Taking half gives \eqref{eq:design-distance}.
By the orthogonality of $Q_A$ and
$Q_B$, $R_A+R_B=s$ and
$q_*=R_A/s$, we obtain \eqref{eq:frame}.
\end{proof}
\begin{remark}
\begin{itemize}
  \item[\rm(i)] For the uniform ensemble on the projective Clifford orbit of
$|\psi\rangle$, the quantity $\Tr(\Omega(q)^2)$ is its third frame
potential. Indeed, if the distinct orbit states are
$\{|\psi_1\rangle,\ldots,|\psi_N\rangle\}$, then
\[
 \Omega(q)
 =
 \frac1N\sum_{i=1}^N
 (|\psi_i\rangle\langle\psi_i|)^{\otimes3},
\]
and hence
\[
 \Tr\bigl(\Omega(q)^2\bigr)
 =
 \frac1{N^2}\sum_{i,j=1}^N
 |\langle\psi_i|\psi_j\rangle|^6.
\]
Thus \eqref{eq:frame} shows that the third frame potential is minimized
precisely when $q=q_*$. Its minimum value is
\[
 \frac1s
 =
 \Tr\left(\left(\frac{\Pi_s}{s}\right)^2\right),
\] which is the Haar value.
  \item[\rm(ii)] For every $n\geq1$ there is a pure state $\psi$ with $q(\psi)=q_*$.
Concretely, in the family \eqref{eq:family-theta} choose $\theta$ so that
\( (\cos^3\theta-\sin^3\theta)^2=\frac{3}{3^n+2},\)
which is possible because the left side ranges over $[0,1]$ while
$3/(3^n+2)\in(0,1]$ for $n\geq1$. Moreover the third moments of the
Clifford orbits of any two pure states are at trace distance
$\Dtr(\Omega(q),\Omega(q'))=|q-q'|\leq1/3^n$. Equality of $q$ determines
equality of these third moments, but does not determine the Clifford
orbit in general.
\end{itemize}
\end{remark}

\section{Measurements and statistically valid tests}
\label{sec:measurements}

Throughout this section we assume $n\geq2$. Under this
assumption all five sector ranks are nonzero (see \eqref{rep:rank-calculation} above).

\subsection{Jointly invariant quantum Boolean functions}

Following~\cite{mo10}, we call a Hermitian unitary $F$ a
quantum Boolean function. We use this definition on $\mathcal K$
and view $F$ as an observable with outcomes in $\{\pm1\}$.
Call $F$ jointly invariant if it commutes with both the
tensor-cube Clifford action and the copy permutations.
By the joint commutant computation, we have
\begin{equation}\label{eq:boolean}
 F_{\boldsymbol s}=\sum_{j\in\{A,B,C,D,E\}}s_jQ_j,
 \quad \boldsymbol s=(s_A,\dots,s_E)\in\{\pm1\}^5 .
\end{equation}
There are $2^5=32$ such observables. Conjugation by $F$ and by
$-F$ gives the same channel $\sigma\mapsto F\sigma F^\dagger$, and two
unitaries induce the same channel precisely when they differ by a
global phase. Hence the 32 sign vectors fall into the 16 pairs
$\{\boldsymbol s,-\boldsymbol s\}$ of sign vectors that agree up to a
global sign, and no further identification occurs. Indeed, if
$F_{\boldsymbol s'}=e^{i\theta}F_{\boldsymbol s}$, then restricting both
sides to the range of $Q_j$, which is nonzero for every $j$ by the rank
assumption, gives $s'_j=e^{i\theta}s_j$ for every $j$. Both sides being
$\pm1$, this forces $e^{i\theta}=\pm1$ and
$\boldsymbol s'=\pm\boldsymbol s$. There are therefore exactly
$2^4=16$ distinct unitary conjugation channels. As
signed observables, $F_{\boldsymbol s}$ and $-F_{\boldsymbol s}$
exchange the outcome labels.

For every density operator $\sigma$ on $\mathcal K$, by \eqref{eq:boolean}, we obtain
\begin{equation}\label{eq:boolean-expectation}
 \Tr(F_{\boldsymbol s}\sigma)=\sum_js_jw_j(\sigma),
 \quad w_j(\sigma)=\Tr(Q_j\sigma).
\end{equation}
Hence a single sector measurement supplies data for all 32 observables
by classical postprocessing.

Clifford invariance alone gives a strictly larger family. On the
two-dimensional multiplicity space of the $E$ sector it permits an
arbitrary Hermitian unitary. The quantum Boolean functions are
\[
 \bigoplus_{j=A,B,C,D}s_jI_{r_j}\ \oplus\ (I_E\otimes B),
\]
where $B$ is a Hermitian unitary on $\mathbb C^2$. The finite count above
therefore uses permutation invariance as well.

\subsection{Sector measurement and access assumptions}

The following sequential measurement has effects
$\{Q_A,Q_B,Q_C,Q_D,Q_E\}$.
\begin{enumerate}
\item[\rm(i)] Measure the eigenvalue of $T=V_{(123)}$ by exact phase estimation
      with an ancilla of dimension three. Either nontrivial eigenvalue
      identifies the sector $E$, because the eigenvalues of $V_{(123)}$
      are $1,\omega,\omega^2$ and
      \begin{equation*}
       E_0=\frac{I+T+T^2}{3}
      \end{equation*}
      is the projection onto the subspace of eigenvalue one.
\item[\rm(ii)] On the eigenvalue one branch, measure the swap
$S=V_{(12)}$
and the mutually commuting CSS checks
$X_i^{\otimes3}$ and $Z_i^{\otimes3}$ for $1\leq i\leq n.$
Declare that the checks pass if every check outcome equals one;
otherwise, declare that they fail. These two events have projections
$P$ and $I-P,$
respectively. Assign the sector label according to
\[
\begin{array}{c|cc}
 & \text{all checks pass} & \text{at least one check fails}\\ \hline
S=+1 & A & B\\
S=-1 & C & D
\end{array}
\]
\end{enumerate}
\begin{lemma}\label{lem:sector-instrument}
The effects of the sequential procedure are exactly
$Q_A,\dots,Q_E$.
\end{lemma}

\begin{proof}
Since $STS^{-1}=T^{-1}$ and
$S^{-1}=S$, it follows that $ST^kS^{-1}=T^{-k}$ for every $k\in\mathbb{Z}$. Hence
\[
SE_0S^{-1}=\frac{I+T^{-1}+T^{-2}}{3}=\frac{I+T^2+T}{3}=E_0 .
\]
Thus $S$ commutes with $E_0$.
Moreover, the even permutations of $S_3$ are represented by
$I,T,T^2$ and the odd ones by $S,TS,T^2S$, the latter set being the same as
$S,ST,ST^2$ up to the convention chosen for the composition of permutations.
Since $T^3=I$, we have
\[
\sum_{\sigma\ \mathrm{even}}V_\sigma=I+T+T^2=3E_0,\quad
\sum_{\sigma\ \mathrm{odd}}V_\sigma=(I+T+T^2)S=3E_0S=3SE_0 .
\]
Combining this with \eqref{eq:mixed-central-projections}, we see that
\[
\Pi_s=\frac16\bigl(I+T+T^2+S+TS+T^2S\bigr)=E_0\frac{I+S}{2}
\]
and
\[
\Pi_a=\frac16\bigl(I+T+T^2-S-TS-T^2S\bigr)=E_0\frac{I-S}{2}.
\]
The CSS checks commute with each other, with $S$, and with all copy
permutations, and their joint $+1$ projection is the projector $P$ as in
\eqref{rep:css-weyl-average}. By this and \eqref{rep:five-projectors} with the definitions
$Q_A=P\Pi_s$, $Q_B=\Pi_s-Q_A$, $Q_C=P\Pi_a$, $Q_D=\Pi_a-Q_C$,
$Q_E=\Pi_m$, we obtain the asserted effects.
\end{proof}
The implementation assumes known controlled permutations, ancillary
systems, and bounded-arity gates, and it permits interactions between
copies. It uses $O(n)$ elementary operations per trial. We do not
assume that Clifford gates alone suffice. Recording the detailed check
outcomes refines the instrument but preserves the coarse sector
probabilities. To implement the unitary \eqref{eq:boolean} itself one
must compute the labels coherently, apply the corresponding signs, and
uncompute. Measuring the labels and applying classical phases does not
implement the same unitary channel.

All accuracy and gap parameters in the testing statements below are
positive, and $0<\delta<1$. Trials use independently prepared, identical
inputs. In the arbitrary input model one trial consumes one copy of
$\sigma$. In the tensor cube model it consumes three independent copies
of $\rho$, prepared as $\rho^{\otimes3}$. Three correlated systems with
marginal state $\rho$ do not generally provide this input. Likewise the
ensemble third moment
$\mathbb E_\psi[(|\psi\rangle\langle\psi|)^{\otimes3}]$
differs in general from
$(\mathbb E_\psi|\psi\rangle\langle\psi|)^{\otimes3}$. The preparation
model must specify which state is supplied.

\subsection{Estimation and testing moment equality}

Let $J_1,\dots,J_N$ be independent sector outcomes with
$\Pr(J_i=j)=w_j$, and let $\widehat w$ be their empirical frequency vector.
The same outcomes are reused for all estimators below. For independent
random variables with values in $[a,b]$, Hoeffding's
inequality~\cite{hoeffding63} gives
\begin{equation}\label{eq:hoeffding}
 \Pr\Bigl(\bigl|\overline X-\mathbb E\overline X\bigr|\geq\epsilon\Bigr)
 \leq2\exp\Bigl(-\frac{2N\epsilon^2}{(b-a)^2}\Bigr).
\end{equation}

\begin{enumerate}
\renewcommand{\theenumi}{\roman{enumi}}
\renewcommand{\labelenumi}{\textnormal{(\theenumi)}}
\item\label{item:weights}
  \textbf{Sector weights.} Let $X_i=\mathbf{1}[J_i=j]$, $a=0$, $b=1$ and
  $\mathbb E X_i=w_j$. If
  \begin{equation*}
   N\geq\Bigl\lceil\frac{25}{8\epsilon^2}\log\frac{10}{\delta}\Bigr\rceil,
  \end{equation*}
  then with probability at least $1-\delta$ one has
  $\frac12\|\widehat w-w\|_1\leq\epsilon$. Indeed, applying
\eqref{eq:hoeffding} to each indicator $\boldsymbol 1_{\{J_i=j\}}$ with
$a=0$, $b=1$ and accuracy $\eta=2\epsilon/5$, and taking a union bound over the
five sectors, gives
\[
 \Pr\Bigl(\max_{j}\bigl|\widehat w_j-w_j\bigr|>\eta\Bigr)
 \le 10\exp\Bigl(-2N\eta^2\Bigr).
\]
Thus $N\geq\bigl\lceil\frac{25}{8\epsilon^2}\log\frac{10}{\delta}\bigr\rceil$
suffices and
\[\frac12 \|\widehat w-w\|_{1}
 =\frac12\sum_{j}\bigl|\widehat w_j-w_j\bigr|
 \le\frac52\eta=\epsilon.
\]
The reconstructed state
  \begin{equation*}
   \widehat\sigma_{\mathrm{twirl}}
   =\sum_j\frac{\widehat w_j}{r_j}Q_j
  \end{equation*}
  then satisfies
  $\Dtr(\widehat\sigma_{\mathrm{twirl}},\mathcal E(\sigma))
  =\frac12\|\widehat w-w\|_1$. This reconstructs the twirl, and the
  input itself only under a joint invariance promise.

\item\label{item:boolean}
\textbf{Expectations of quantum Boolean functions.} Use the outcome function
$s_{J_i}$, whose range has length at most two. Thus $a=-1$, $b=1$ in
\eqref{eq:hoeffding}. For each sign vector $\boldsymbol s\in\{\pm1\}^5$ the
estimator $\overline X_{\boldsymbol s}=\frac1N\sum_i s_{J_i}$ satisfies
\[
 \Pr\Bigl(\bigl|\overline X_{\boldsymbol s}
 -\Tr(F_{\boldsymbol s}\sigma)\bigr|>\epsilon\Bigr)
 \le2\exp\Bigl(-\frac{N\epsilon^2}{2}\Bigr),
\]
and a union bound over the $32$ sign vectors gives
\[
 \Pr\Bigl(\exists\,\boldsymbol s:\
 \bigl|\overline X_{\boldsymbol s}-\Tr(F_{\boldsymbol s}\sigma)\bigr|>\epsilon\Bigr)
 \le64\exp\Bigl(-\frac{N\epsilon^2}{2}\Bigr).
\]
Hence, if
\begin{equation*}
 N\geq\Bigl\lceil\frac{2}{\epsilon^2}\log\frac{64}{\delta}\Bigr\rceil,
\end{equation*}
then every expectation in \eqref{eq:boolean-expectation} is estimated to
additive error $\epsilon$, simultaneously with probability at least
$1-\delta$. No invariance promise on the input is needed, since Hoeffding uses
only independence and boundedness of the sampled outcomes.

\item\label{item:cube}
  \textbf{Tensor cube inputs.} Since the mean of $\boldsymbol 1_{\{J_i=j\}}$ is
$w_j$, an outcome function with coefficients $c_j$ is unbiased for
$\sum_jc_jw_j$. In the tensor cube model the sector weights satisfy the
inverse relations, which express $t_2,t_3,u,v$ as four such combinations. The
following outcome functions are therefore unbiased for $t_2,t_3,u,v$.
\begin{equation*}
 \begin{array}{c|ccccc|c}
  & A & B & C & D & E & \text{range length}\\ \hline
  t_2 & 1 & 1 & -1 & -1 & 0 & 2\\
  t_3 & 1 & 1 & 1 & 1 & -1/2 & 3/2\\
  u   & 1 & 0 & 1 & 0 & 0 & 1\\
  v   & 1 & 0 & -1 & 0 & 0 & 2
 \end{array}
\end{equation*}
Writing $L_i$ for the range length of the $i$-th outcome function and
allocating failure probability $\delta/4$ to each of the four estimators,
Hoeffding's inequality gives the requirement
$N\geq\frac{L_i^2}{2\eta^2}\log\frac{8}{\delta}$ for the $i$-th of them. The
largest range length is $2$, attained by $t_2$ and by $v$, hence
\begin{equation}\label{eq:N-cube}
 N\geq\Bigl\lceil\frac{2}{\eta^2}\log\frac{8}{\delta}\Bigr\rceil
\end{equation}
suffices for all four quantities to be estimated to additive error $\eta$ with
probability at least $1-\delta$, using $3N$ copies of $\rho$. For $t_3$ the
range length is $3/2$, so Hoeffding's inequality requires only
$N\geq\frac{9}{8\eta^2}\log(8/\delta)$ for its failure probability to be at
most $\delta/4$. This is implied by \eqref{eq:N-cube}, and the other range
lengths are at most two.

\item\label{item:pure}
\textbf{Pure tensor cubes.} Only the sectors $A$ and $B$ occur, and
$\boldsymbol 1_{\{J_i=A\}}$ is Bernoulli with mean
$q=\Tr(P(|\psi\rangle\langle\psi|)^{\otimes3})$. Thus $a=0$, $b=1$ in
\eqref{eq:hoeffding}, and a single application of the inequality, with
failure probability budget $\delta$, gives
\[
 \Pr\bigl(|\widehat q-q|\geq\eta\bigr)\le2\exp\bigl(-2N\eta^2\bigr).
\]
Hence, if
\begin{equation}\label{eq:N-pure}
 N\geq\Bigl\lceil\frac{1}{2\eta^2}\log\frac2\delta\Bigr\rceil,
\end{equation}
then $q$ is estimated to additive error $\eta$ with probability at least
$1-\delta$, again using $3N$ copies. No union bound is needed here because a
single quantity is estimated.
\end{enumerate}

For comparison with a known vector $w^{(0)}$ consider the promise hypotheses
$H_0$ given by $w=w^{(0)}$ and $H_1$ given by
$\frac12\|w-w^{(0)}\|_1\geq\epsilon$. Estimate $w$ to total variation accuracy
$\epsilon/4$, which by the sector weight bound means per-coordinate accuracy
$\eta=2(\epsilon/4)/5=\epsilon/10$, and reject when
$\frac12\|\widehat w-w^{(0)}\|_1>\epsilon/2$. Thus
\[
 N\geq\Bigl\lceil\frac{1}{2\eta^2}\log\frac{10}{\delta}\Bigr\rceil
 =\Bigl\lceil\frac{50}{\epsilon^2}\log\frac{10}{\delta}\Bigr\rceil .
\]
On the event $\frac12\|\widehat w-w\|_1\leq\epsilon/4$, whose probability is at
least $1-\delta$, the triangle inequality gives
$\frac12\|\widehat w-w^{(0)}\|_1\leq\epsilon/4<\epsilon/2$ under $H_0$ and
$\frac12\|\widehat w-w^{(0)}\|_1\geq3\epsilon/4>\epsilon/2$ under $H_1$, so each
error probability is at most $\delta$. No guarantee is imposed between the two
hypotheses. This procedure tests equality of twirled states. It does not test
equality or Clifford equivalence of the original inputs.

\subsection{Testing invariance without random Clifford gates}

Define the squared Hilbert--Schmidt distance from the invariant subspace by
\begin{equation*}
 S_G(\sigma)
 =\|\sigma-\mathcal E(\sigma)\|_2^2
 =\Tr(\sigma^2)-\sum_j\frac{w_j(\sigma)^2}{r_j}.
\end{equation*}
The second identity holds because $\mathcal E$ is the orthogonal
projection onto the jointly invariant subspace for the Hilbert--Schmidt
inner product. Hence $S_G(\sigma)\geq0$ with equality exactly when
$\sigma$ is jointly invariant.

We recall the instrument formalism \cite{daviesLewis70}, using
completely positive maps as in \cite{kraus83}.
\begin{definition}\label{def:instrument}
Let $K$ be a finite outcome set and let $\mathcal V$ be a finite-dimensional
Hilbert space. A quantum instrument on $\mathcal V$ with outcome set $K$
is a family $\{\mathcal I_k\}_{k\in K}$ of completely positive linear maps on
$\mathrm{End}(\mathcal V)$ such that $\sum_k\mathcal I_k$ is trace preserving.
Its outcome probabilities and post-measurement states for an input
density operator $\rho$ on $\mathcal V$ are
\[
p_k=\operatorname{Tr}\bigl(\mathcal I_k(\rho)\bigr),\quad
\rho_k=\frac{\mathcal I_k(\rho)}{p_k}\quad(p_k>0),
\]
and its nonselective map is $\Phi=\sum_k\mathcal I_k$.
\end{definition}

\begin{definition}\label{def:luders}
Let $\{P_k\}_{k\in K}$ be a projection-valued measure on $\mathcal V$, that is,
\[
P_k=P_k^\dagger=P_k^2,\quad P_kP_l=\delta_{kl}P_k,\quad \sum_kP_k=I .
\]
The L\"uders instrument associated with $\{P_k\}$ is the instrument
$\mathcal I_k(\rho)=P_k\rho P_k$; see \cite{luders51}.
Thus the operation for outcome $k$ has the single Kraus operator $P_k$.
Its outcome probabilities, post-measurement states and nonselective map are
\[
p_k=\operatorname{Tr}(P_k\rho),\quad
\rho_k=\frac{P_k\rho P_k}{\operatorname{Tr}(P_k\rho)}\quad(p_k>0),\quad
\Phi(\rho)=\sum_kP_k\rho P_k.
\]
The map $\Phi$ is the pinching map associated with $\{P_k\}$.
\end{definition}
The following standard properties follow directly from these definitions.
\begin{lemma}\label{prop:luders}
For the L\"uders instrument in Definition~\ref{def:luders}, the following hold.
\begin{enumerate}
\item[\rm(i)] If the outcome $k$ is obtained,
then an immediate repetition of the same measurement yields $k$ with
probability one and leaves the post-measurement state unchanged.
\item[\rm(ii)] $\Phi$ is unital, completely positive, trace preserving and
idempotent. Its fixed point algebra is the commutant
$\{P_k\}'=\{A:[A,P_k]=0\ \text{for all }k\}$, and $\Phi$ satisfies the
bimodule property $\Phi(ARB)=A\,\Phi(R)\,B$ for all $A,B\in\{P_k\}'$ and
$R\in\mathrm{End}(\mathcal V)$.
Thus $\Phi$ is the trace preserving conditional expectation onto the
commutant of the projection-valued measure.
\item[\rm(iii)] The dual map acts by $\Phi^*(R)=\sum_kP_kRP_k$, and
$\Phi^*(R)=R$ if and only if $[R,P_k]=0$ for all $k$.
\end{enumerate}
\end{lemma}
\begin{proof}
For $p_k>0$, the post-measurement state satisfies
$P_l\rho_kP_l=\delta_{kl}\rho_k$, proving (i).
The Kraus representation gives complete positivity, while
$\sum_kP_k^2=I$ gives unitality and trace preservation.
Orthogonality of the projections gives $\Phi^2=\Phi$ and
$P_l\Phi(R)=P_lRP_l=\Phi(R)P_l$.
Conversely, if $R$ commutes with every $P_l$, then $\Phi(R)=R$.
Moving $A$ and $B$ past each $P_k$ proves the bimodule property and hence (ii).
Finally, cyclicity of the trace gives $\Phi^*(R)=\sum_kP_kRP_k=\Phi(R)$,
so the fixed point characterization proves (iii).
\end{proof}
By a swap measurement on two registers with the same Hilbert space, we mean the projective measurement of the swap operator $\mathrm{SWAP}$, defined by
\[
\mathrm{SWAP}(|a\rangle\otimes|b\rangle)=|b\rangle\otimes|a\rangle.
\]
Its outcomes are $\pm1$, with corresponding projections $(I\pm\mathrm{SWAP})/2$. On an input state $\tau$, the mean outcome is $\operatorname{Tr}(\mathrm{SWAP}\tau)$. In particular, on input $\tau_1\otimes\tau_2$, it equals $\operatorname{Tr}(\tau_1\tau_2)$.

\begin{proposition}\label{prop:invariance-test}
Let
\(
r_j=\operatorname{rank}Q_j\) and \(r_{\min}=\min_jr_j=\frac{3^n-1}{2}\). The promise hypotheses $S_G(\sigma)=0$ and $S_G(\sigma)\geq\alpha$ can be
distinguished with each error probability at most $\delta$ using $2N$ copies
of $\sigma$, where
\begin{equation*}
 N\geq\Bigl\lceil\frac{2(2+1/r_{\min})^2}{\alpha^2}\log\frac2\delta
 \Bigr\rceil .
\end{equation*}
The procedure uses a global SWAP measurement with the L\"uders instrument,
followed by sector measurements, and requires no random Clifford gates.
\end{proposition}

\begin{proof}
On each independent pair $\sigma\otimes\sigma$, measure SWAP by the
two-outcome L\"uders instrument with Kraus operators
\[
P_+=\frac{I+\mathrm{SWAP}}2,\quad P_-=\frac{I-\mathrm{SWAP}}2,
\]
the orthogonal projections onto $\mathrm{Sym}^2(\KK)$ and $\wedge^2\KK$,
and record $X=+1$ for $P_+$ and $X=-1$ for $P_-$.
The controlled SWAP ancilla circuit in \cite{bcww01} implements this instrument. By
$\operatorname{Tr}(\mathrm{SWAP}(A\otimes B))=\operatorname{Tr}(AB)$, we obtain
\begin{equation}\label{eq:EX}
 \mathbb E X=\operatorname{Tr}\bigl((P_+-P_-)(\sigma\otimes\sigma)\bigr)
 =\operatorname{Tr}\bigl(\mathrm{SWAP}(\sigma\otimes\sigma)\bigr)
 =\operatorname{Tr}(\sigma^2).
\end{equation}
 Next
measure the sector of each of the two systems, obtaining labels $J$ and $K$,
and put
\[
Z=\frac{\boldsymbol 1_{\{K=J\}}}{r_J},\quad H=X-Z .
\]
Let $\Phi(\tau)=P_+\tau P_++P_-\tau P_-$ be the nonselective map of the
first measurement, as in Definition~\ref{def:luders}, and set
\[
R=\sum_j\frac{Q_j\otimes Q_j}{r_j}.
\]
Since $\mathrm{SWAP}(\mathcal{A}\otimes \mathcal{A})\mathrm{SWAP}=\mathcal{A}\otimes \mathcal{A}$ for every operator
$\mathcal{A}$, one has $\mathrm{SWAP}\,R\,\mathrm{SWAP}=R$, hence $[R,\mathrm{SWAP}]=0$. By this and Lemma \ref{prop:luders}, we have
$\Phi^*(R)=R$. Averaging the conditional distribution of $(J,K)$ over
the recorded outcome of the first instrument therefore gives
\begin{equation}\label{eq:EZ}
 \mathbb E Z=\operatorname{Tr}\bigl(R\,\Phi(\sigma\otimes\sigma)\bigr)
 =\operatorname{Tr}\bigl(\Phi^*(R)(\sigma\otimes\sigma)\bigr)
 =\operatorname{Tr}\bigl(R(\sigma\otimes\sigma)\bigr)
 =\sum_j\frac{w_j(\sigma)^2}{r_j} .
\end{equation}
No independence of $X$ and $Z$ within a trial is used, only the independence
of the trials. Combining \eqref{eq:EX} and
\eqref{eq:EZ}, we know that
\begin{equation*}
 \mathbb E H=\mathbb E X-\mathbb E Z
 =\operatorname{Tr}(\sigma^2)-\sum_j\frac{w_j(\sigma)^2}{r_j}=S_G(\sigma).
\end{equation*}
Since $X\in\{-1,1\}$ and $0\leq Z\leq1/r_{\min}$, it follows that
\begin{equation*}
 H\in\Bigl[-1-\frac1{r_{\min}},\ 1\Bigr],
 \quad b-a=2+\frac1{r_{\min}} .
\end{equation*}
Let $\overline H$ be the
empirical mean of $H$ over the $N$ independent pairs and reject the invariant
hypothesis when $\overline H>\alpha/2$. Under $H_0$ one has
$\mathbb EH=S_G=0$, so rejection implies
$|\overline H-\mathbb EH|\geq\alpha/2$. Under $H_1$ one has
$\mathbb EH=S_G\geq\alpha$, so not rejecting implies
$\overline H\leq\alpha/2$, hence again
$|\overline H-\mathbb EH|\geq\alpha/2$. By \eqref{eq:hoeffding} with
$b-a=2+1/r_{\min}$ and accuracy $\alpha/2$, the two error probabilities are
therefore both at most
\[
 2\exp\Bigl(-\frac{2N(\alpha/2)^2}{(2+1/r_{\min})^2}\Bigr),
\]
which is at most $\delta$ precisely when
\[
 N\geq\frac{2(2+1/r_{\min})^2}{\alpha^2}\log\frac2\delta .
\]
Each trial consumes two copies of $\sigma$, whence $2N$ copies in total, and
no random Clifford gate is used.
\end{proof}
The measurement order is essential. Measuring the separate sectors
first can remove cross-sector coherences and alter a subsequent purity
estimate. An alternative uses independent pairs for the purity and for
the sector collision estimates. This quantity controls trace distance
only through the dimension-dependent bound
\begin{equation*}
 \Dtr\bigl(\sigma,\mathcal E(\sigma)\bigr)
 \leq\frac12\sqrt{3^{3n}S_G(\sigma)},
\end{equation*}
which follows from $\|\cdot\|_1\leq\sqrt{d}\|\cdot\|_2$ in dimension
$d=3^{3n}$. Thus this is not a dimension-independent trace distance
test.

For tensor cube inputs,
$S_G(\rho^{\otimes3})=t_2^3-\sum_jw_j^2/r_j$, and each paired trial
uses six copies of $\rho$. Moreover $\rho^{\otimes3}$ is Clifford
invariant if and only if $\rho=I/3^n$. Partial tracing the invariance
identity gives $\mathcal{C}\rho \mathcal{C}^\dagger=\rho$ for every Clifford operator, and
the Pauli operators have scalar commutant. Under the tensor cube
promise, sector data therefore determine invariance through
$t_2=1/3^n$, although finite sample testing requires a promise gap. For
arbitrary inputs, sector probabilities alone do not determine
invariance. The swap measurement supplies the additional purity
information used above.

\subsection{Stabilizer testing, soundness and dimension dependence}
\label{rep:stab-testing}

Throughout this subsection the input is promised to be pure, and
$F_{\mathrm{stab}}(\psi)$ denotes its maximum squared overlap with a
stabilizer state.

The equality $q=1/3^n$ from Theorem~\ref{thm:range} characterizes
stabilizer pure states, but measuring $P$ accepts a stabilizer cube
with probability $1/3^n$ rather than with probability one. The robust
estimate of Proposition~\ref{prop:robust}, written as
\begin{equation*}
 1-F_{\mathrm{stab}}(\psi)\leq K(1-3^nq),\quad K=2+\sqrt2,
\end{equation*}
converts the characterization into a promise test. It yields the
implication
\begin{equation*}
 F_{\mathrm{stab}}(\psi)\leq1-\zeta
 \quad\Longrightarrow\quad
 q\leq\frac1{3^n}-\frac{\zeta}{K3^n}.
\end{equation*}
Estimate $q$ by the indicator of the $A$ sector and accept the
stabilizer hypothesis when
$\widehat q>\frac1{3^n}-\frac{\zeta}{2K3^n}$. By \eqref{eq:N-pure}, the bound
\begin{equation*}
 N\geq\Bigl\lceil\frac{2K^23^{2n}}{\zeta^2}\log\frac2\delta\Bigr\rceil
\end{equation*}
on the number of independent three-copy batches bounds both the
soundness and completeness error probabilities by $\delta$. This bound
is not optimal for rare event estimation, but it makes its dimension
dependence explicit.

A POVM is a family of positive operators $\{E_a\}_a$ satisfying $\sum_a E_a=I$. Each operator $E_a$ is called the effect associated with outcome $a$ and determines its probability through
$
\Pr(a\mid\tau)=\operatorname{Tr}(E_a\tau).
$
In the sector measurement above, the effects are the mutually orthogonal projections $Q_A,\ldots,Q_E$ (see Lemma \ref{lem:sector-instrument}).

Now we consider tests that measure each of $N$ independent three-copy inputs
separately, with every retained POVM effect of every batch Clifford
invariant. On a pure cube every such POVM is a classical postprocessing of
the two-sector distribution $(q,1-q)$, because the symmetric subspace carries
only the two inequivalent Clifford sectors $A$ and $B$, each with multiplicity
one (see Proposition \ref{prop:pure-reduction}).
\begin{proposition}\label{prop:lower}
Fix any pure state in the promised alternative, and let $q_1$ be its
parameter. Then $0\leq q_1\leq3^{-n}$, and the error probabilities for this
state and a stabilizer state satisfy
\begin{equation}\label{eq:two-point}
 P_{\mathrm{err}}^{(0)}+P_{\mathrm{err}}^{(1)}
 \geq\bigl(1-3^{-n}\bigr)^N .
\end{equation}
Thus, bounding each error probability by $\delta<1/2$ requires
\begin{equation*}
 N\geq-\frac{\log\bigl(1/(2\delta)\bigr)}{\log\bigl(1-3^{-n}\bigr)}
\end{equation*}
and  $N=\Omega(3^n\log(1/\delta))$ for $0<\delta\leq1/4$.
\end{proposition}

\begin{proof}
By Theorem \ref{thm:range}, we know that $0\leq q_1\leq3^{-n}$.
Let $P_0$ and $P_1$ be the distributions of the sector record under the
stabilizer input and under the fixed alternative. For a stabilizer input the
record consisting entirely of $B$ outcomes has probability
$\bigl(1-3^{-n}\bigr)^N$, since the sector $A$ occurs with probability
$3^{-n}$ on a single cube, whereas under the alternative that same record has
probability $(1-q_1)^N$. Because $q_1\leq3^{-n}$ and $q\mapsto(1-q)^N$ is
decreasing,
\[
 (1-q_1)^N\geq\bigl(1-3^{-n}\bigr)^N .
\]
Let
\[
 \|P_0-P_1\|_{\mathrm{TV}}
 =\frac12\|P_0-P_1\|_1
 =\frac12\sum_r\bigl|P_0(r)-P_1(r)\bigr|.
\]
Then
$1-\|P_0-P_1\|_{\mathrm{TV}}=\sum_r\min\{P_0(r),P_1(r)\}$.
In fact,
for real $a,b$ one has $(a-b)^+=\frac{a-b+|a-b|}{2}$ and
$\min\{a,b\}=\frac{a+b-|a-b|}{2}$. Since
$\sum_r(P_0(r)-P_1(r))=0$, the positive and negative parts of
$P_0-P_1$ have equal total mass, hence
\[
 \sum_r\bigl(P_0(r)-P_1(r)\bigr)^+
 =\frac12\sum_r\bigl|P_0(r)-P_1(r)\bigr|
 =\|P_0-P_1\|_{\mathrm{TV}} .
\]
By $\sum_rP_0(r)=1$ and $a-(a-b)^+=\min\{a,b\}$, we have
\[
 \sum_r\min\{P_0(r),P_1(r)\}
 =\sum_rP_0(r)-\sum_r\bigl(P_0(r)-P_1(r)\bigr)^+
 =1-\|P_0-P_1\|_{\mathrm{TV}} .
\]
Then
\[
 \sum_r\min\{P_0(r),P_1(r)\}\ \geq\ \bigl(1-3^{-n}\bigr)^N
\]
and
\begin{equation}\label{eq:TV-bound}
 \|P_0-P_1\|_{\mathrm{TV}}\ \leq\ 1-\bigl(1-3^{-n}\bigr)^N .
\end{equation}
Let $\kappa$ be the classical postprocessing, that is the stochastic map that
the test applies to the sector record. For every event $B$ of the processed
record,
\[
 (\kappa P_0)(B)-(\kappa P_1)(B)
 =\sum_r\kappa(B\mid r)\bigl(P_0(r)-P_1(r)\bigr)
 \leq\sum_r\bigl(P_0(r)-P_1(r)\bigr)^+
 =\|P_0-P_1\|_{\mathrm{TV}},
\]
because $0\leq\kappa(B\mid r)\leq1$. Taking the supremum over $B$ gives the data
processing inequality
\begin{equation}\label{eq:DPI}
 \|\kappa P_0-\kappa P_1\|_{\mathrm{TV}}\leq\|P_0-P_1\|_{\mathrm{TV}} .
\end{equation}
Let $\varphi$ be a possibly randomized decision rule on the processed
record, and let
$P_{\mathrm{err}}^{(0)}=\Pr_{\kappa P_0}(\varphi=1)$ and
$P_{\mathrm{err}}^{(1)}=\Pr_{\kappa P_1}(\varphi=0)$ be its error probabilities.
Le Cam's two-point bound (see the proof of
\cite[Theorem~2.2(i)]{tsybakov2009}) gives
\begin{equation}\label{eq:lecam}
 P_{\mathrm{err}}^{(0)}+P_{\mathrm{err}}^{(1)}
 \geq1-\|\kappa P_0-\kappa P_1\|_{\mathrm{TV}} ,
\end{equation}
whose proof is as follows. It suffices to consider a deterministic test, since
a randomized test is a convex combination of deterministic ones and the sum of
the two error probabilities is linear in the test. For a deterministic test,
let $A$ be the region on which the null hypothesis is accepted, so that
$P_{\mathrm{err}}^{(0)}=1-(\kappa P_0)(A)$ and
$P_{\mathrm{err}}^{(1)}=(\kappa P_1)(A)$. Then
\[
 P_{\mathrm{err}}^{(0)}+P_{\mathrm{err}}^{(1)}
 =1-\bigl((\kappa P_0)(A)-(\kappa P_1)(A)\bigr)
 \geq 1-\|\kappa P_0-\kappa P_1\|_{\mathrm{TV}},
\]
which is \eqref{eq:lecam}. The bound is attained by the test that accepts the
null hypothesis on $A=\{y:(\kappa P_0)(y)>(\kappa P_1)(y)\}$,
where $y$ denotes the processed record.
From \eqref{eq:lecam}, \eqref{eq:DPI} and \eqref{eq:TV-bound}, it follows that
\[
 P_{\mathrm{err}}^{(0)}+P_{\mathrm{err}}^{(1)}
 \geq1-\|P_0-P_1\|_{\mathrm{TV}}
 \geq\bigl(1-3^{-n}\bigr)^N ,
\]
which is \eqref{eq:two-point}. If each error probability is at most $\delta$,
then $2\delta\geq(1-3^{-n})^N$, and taking logarithms gives
\[
 N\geq\frac{\log\bigl(1/(2\delta)\bigr)}{-\log\bigl(1-3^{-n}\bigr)} .
\]
Since $-\log(1-3^{-n})\leq\frac{ 3^{-n}}{1-3^{-n}}$ and
$-\log(1-3^{-n})\geq3^{-n}$ for $n\ge1$, and since
$\log\bigl(1/(2\delta)\bigr)\geq\frac12\log(1/\delta)$ whenever
$\delta\leq1/4$, the last display implies
$N=\Omega\bigl(3^n\log(1/\delta)\bigr)$ for $0<\delta\leq1/4$.
\end{proof}

This lower bound requires every retained outcome effect to be Clifford
invariant, as for the coarse sector measurement. Keeping detailed CSS
syndrome records generally falls outside this model, and the bound does
not cover arbitrary collective measurements on all $3N$ copies.
Moreover the normalized quantity $3^nq$ is not a measurement probability
on arbitrary three-copy inputs, since $3^nP$ is not a POVM effect.

\begin{proposition}
\label{prop:no-perfect}
Any test on three copies of a pure state that accepts every stabilizer
state with certainty also accepts every pure state with certainty. This
holds without an invariance assumption on the measurement.
\end{proposition}

\begin{proof}
Let
\begin{equation*}
 M_{\mathrm{stab}}^{(3)}
 =\mathbb E_{\phi\,\mathrm{stab}}
 \bigl[(|\phi\rangle\langle\phi|)^{\otimes3}\bigr].
\end{equation*}
By the sector decomposition with $q=1/3^n$, we have
\begin{equation*}
 M_{\mathrm{stab}}^{(3)}
 =\frac{1}{3^n r_A}Q_A+\frac{1-1/3^n}{r_B}Q_B .
\end{equation*}
Both coefficients are positive and $Q_A+Q_B=\Pi_s$, so
$M_{\mathrm{stab}}^{(3)}$ is strictly positive on the symmetric subspace.
If an effect $0\leq M\leq I$ accepts every stabilizer cube
with certainty, then
$\Tr[(I-M)M_{\mathrm{stab}}^{(3)}]=0$. Positivity forces
$(I-M)\Pi_s=0$, so $M$ acts as the identity on the symmetric subspace
and therefore accepts every pure cube. This is also a consequence of
\cite[Theorem~5.6 and Example~4.26]{gnw21}, which identify the span
of qutrit stabilizer tensor cubes with the symmetric subspace.
\end{proof}

For comparison, the four-copy qutrit test of
\cite[Theorem~3.11]{gnw21} has perfect completeness and, under a
fidelity promise stated in the same convention as the present paper,
rejection probability at least $\zeta/72$. Repeating this test on
independent batches and rejecting as soon as a batch rejects therefore
requires at most
\begin{equation}\label{eq:N-fourcopy}
 \Bigl\lceil\frac{72}{\zeta}\log\frac1\delta\Bigr\rceil
\end{equation}
four-copy batches to achieve soundness error at most $\delta$. This is
a different measurement, supported by a separate soundness theorem. It
does not follow from the classification of the 32 quantum Boolean
functions on three copies.

\begin{remark}\label{rem:constant-convention}
Theorem~3.11 of~\cite{gnw21} uses the hypothesis
$F_{\mathrm{stab}}\leq1-\varepsilon^2$ and gives rejection probability
at least $C_{d,s}\varepsilon^2$, where
$C_{d,s}=[1-(1-1/(4d^2))^{s-1}]/2$ and where the test operator of that
theorem acts on $2s$ copies. For $d=3$, $s=2$ and $\zeta=\varepsilon^2$
the test is the four-copy qutrit test, with $C_{3,2}=1/72$ and
rejection probability at least $\zeta/72$, which gives \eqref{eq:N-fourcopy}.
\end{remark}

\section{Testing unknown quantum Boolean functions}\label{sec:operators}

\subsection{Testing membership of an unknown quantum Boolean function}
\label{subsec:optest-membership}

As in the preceding section we assume $n\geq2$, so that all five sector
ranks $r_j$ of \eqref{rep:rank-calculation} are positive. State
testing and operator testing require different access models.
Here the unknown object is a unitary oracle $f$ on $\mathcal K$ that
also satisfies $f=f^\dagger$. Being both unitary and Hermitian, $f$ has
eigenvalues in $\{\pm1\}$ and therefore satisfies $f^2=I$. Define
\[
 \|\mathcal{A}\|_{2,n}^2=3^{-n}\operatorname{Tr}(\mathcal{A}^\dagger\mathcal{A}),\quad
 \mathbf{B}:=\left\{\sum_{j\in\{A,B,C,D,E\}}s_jQ_j:\,
                         s_j\in\{-1,1\}\right\}.
\]
Thus $ \mathbf{B}$ is the set of  all jointly invariant quantum Boolean
functions. Conversely, a Hermitian unitary in the commutant
$\bigoplus_j\mathbb C Q_j$ has coefficient $s_j=\pm1$ in every sector.
Every $Q_j$ is nonzero because $n\geq2$, so $F^2=\sum_jQ_j=I$
for every $F\in\mathbf{B}$ and $\mathbf{B}$ is closed under
$F\mapsto-F$.

Let $G=G_n\times S_3$, where $G_n$ is the projective Clifford group as in
\eqref{EG1}, and choose a unitary representative $\mathcal C_c$ of each
$c\in G_n$. For $g=(c,\pi)\in G$, write
$U(g):=\mathcal C_c^{\otimes3}V_\pi$. Only the conjugation action of
$U(g)$ is used, and that action does not depend on the phase of $\mathcal C_c$.
Because every joint sector carries a single copy of its irreducible
representation, the conjugation average is the orthogonal projection
onto the commutant of $G$ and equals
\[
 \mathcal E(\mathcal{A})=\mathbb E_{g\in G}U(g)\mathcal{A}U(g)^\dagger
       =\sum_j\frac{\operatorname{Tr}(Q_j\mathcal{A})}{r_j}Q_j .
\]
We assume independent uniform samples of $g$, preparation of
$|\Phi_{3^{n}}\rangle:=3^{-n/2}\sum_{k=1}^{3^{n}}|k\rangle|k\rangle$,
and controlled implementations of the known unitaries
\[
 K_g=U(g)\otimes\overline{U(g)}.
\]
Complex conjugation is taken in the computational basis used to
define $|\Phi_{3^n}\rangle$. This controlled operation is independent of the
phase of the chosen Clifford representative. Replacing $U(g)$ by
$e^{i\phi}U(g)$ leaves $K_g$ exactly unchanged. Controlled access to the
unknown $f$ is not required. One ordinary query prepares
\[
 |f\rangle=(f\otimes I)|\Phi_{3^n}\rangle.
\]

\begin{theorem}
\label{thm:optest-membership}
Let
\[
\mathcal{A}=\mathcal E(f)=\sum_ja_jQ_j,\quad
 a_j=\operatorname{Tr}(Q_jf)/r_j,\quad
 v=1-3^{-n}\operatorname{Tr}(\mathcal{A}^2).
\]
Then $a_j\in[-1,1]$ and
\begin{equation}\label{eq:optest-rounding}
 v\leq\min_{F\in\mathbf{B}}\|f-F\|_{2,n}^2
   =2-\frac2{3^{n}}\sum_jr_j|a_j|\leq2v .
\end{equation}
Moreover, $v=0$ if and only if $f\in\mathbf{B}$. For $\epsilon>0$ and
$0<\delta<1$, the promises $f\in\mathbf{B}$ and
$\min_{F\in\mathbf{B}}\|f-F\|_{2,n}\geq\epsilon$
can be distinguished with perfect completeness and false acceptance
probability at most $\delta$ using
\begin{equation}\label{eq:optest-membership-count}
 N=\left\lceil4\epsilon^{-2}\log\frac1\delta\right\rceil
\end{equation}
ordinary oracle calls to $f$. Each call uses a fresh preparation. For the
promises $v=0$ and $v\geq\gamma$, a Hoeffding version uses
$\lceil32\gamma^{-2}\log(2/\delta)\rceil$ calls.
\end{theorem}

\begin{proof}
The twirl $\mathcal E$ is the orthogonal projection onto the commutant of
$G$ with respect to the Hilbert--Schmidt inner product, and $\mathcal{A}$ is a
Hermitian contraction. Indeed
\[
 \|\mathcal{A}\|_\infty\leq\mathbb E_g\|U(g)fU(g)^\dagger\|_\infty
 =\|f\|_\infty\leq1,
\]
where
\(\|\cdot\|_\infty\) is the Schatten $\infty$-norm.
Since $Q_j$ is a projection of rank $r_j$, it follows that
\[
 |a_j|=\frac{|\operatorname{Tr}(Q_jfQ_j)|}{r_j}
 \leq\|Q_jfQ_j\|_\infty\leq1 ,
\]
and $\sum_jr_j=3^n$, while $\operatorname{Tr}(f^2)=\operatorname{Tr}(F^2)=3^n$
for every $F\in\mathbf{B}$.
For $F=\sum_js_jQ_j$ the projection property gives
$\operatorname{Tr}(fF)=\operatorname{Tr}(\mathcal{A}F)$, and hence
\[
 \|f-F\|_{2,n}^2
      =3^{-n}\bigl(\operatorname{Tr}(f^2)-2\operatorname{Tr}(fF)
                    +\operatorname{Tr}(F^2)\bigr)
      =2-\frac2{3^{n}}\sum_jr_ja_js_j.
\]
Choosing $s_j=\operatorname{sgn}(a_j)$, with
$\operatorname{sgn}(0)=1$, gives the minimum
$2-\frac2{3^n}\sum_jr_j|a_j|$. Applying
\[
 1-x^2\leq2(1-x)\leq2(1-x^2),\quad 0\leq x\leq1,
\]
to $x=|a_j|$ and averaging with the weights $r_j/3^n$, which sum to one,
gives
\[
 1-3^{-n}\sum_jr_ja_j^2
 \ \leq\ 2-\frac2{3^{n}}\sum_jr_j|a_j|
 \ \leq\ 2-\frac2{3^{n}}\sum_jr_ja_j^2 .
\]
The left side is $v$ and the right side is $2v$, which proves
\eqref{eq:optest-rounding}. Moreover $v=0$ says
$\|\mathcal{A}\|_{2,n}=\|f\|_{2,n}=1$, which for the orthogonal projection
$\mathcal E$ holds exactly when $\mathcal{A}=f$. Since a Hermitian unitary in the
commutant of $G$ is $\sum_js_jQ_j$ with $s_j=\pm1$, this is exactly
$f\in\mathbf{B}$.

For independently sampled $g$, initialize an ancilla in $|+\rangle$,
apply the controlled unitary $K_g$ to $|f\rangle$, and measure the
ancilla in the Pauli-$X$ basis. Conditionally on $g$, the resulting sign
$X\in\{-1,1\}$ has mean
$\operatorname{Re}\langle f|K_g|f\rangle$.
The vectorization identity
\( (U\otimes\overline U)|f\rangle=|UfU^\dagger\rangle\)
gives $\mathbb E_gK_g|f\rangle=|\mathcal E(f)\rangle$, so averaging over
$g$ as well,
\[
 \mathbb EX
 =3^{-n}\operatorname{Re}\operatorname{Tr}
                  (f^\dagger\mathcal E(f))
 =3^{-n}\operatorname{Tr}(f\mathcal{A})
 =3^{-n}\operatorname{Tr}(\mathcal{A}^2)=1-v,
\]
the third identity using that $\mathcal E$ is self-adjoint and
$\mathcal E(f)=\mathcal{A}$.
Hence $\Pr(X=-1)=v/2$. If $f\in\mathbf{B}$, then $f$ commutes with every
$U(g)$, so $K_g|f\rangle=|f\rangle$ for every $g$ and a negative
outcome is impossible. Under the distance promise,
\eqref{eq:optest-rounding} gives $v\geq\epsilon^2/2$, so each trial is
negative with probability at least $\epsilon^2/4$. Rejecting on the first
negative outcome consequently has perfect completeness and false
acceptance probability at most
$(1-\epsilon^2/4)^N\leq e^{-N\epsilon^2/4}$, proving
\eqref{eq:optest-membership-count}.
For the last assertion, estimate $v=1-\mathbb EX$ by $1-\overline X$ to
accuracy $\gamma/4$ and use the threshold $\gamma/2$. Hoeffding's
inequality for $X\in[-1,1]$, whose range has length two, bounds the error
probability of each decision by $2\exp(-N\gamma^2/32)$, so the stated
count $\lceil32\gamma^{-2}\log(2/\delta)\rceil$ suffices.
\end{proof}

The statistic is unchanged under $f\mapsto-f$, as are membership in and
distance to $\mathbf{B}$. Thus the test is consistent with the
indistinguishability of the ordinary unitary oracles $f$ and $-f$.

The condition $f^2=I$ is an input promise throughout. The protocol does
not test this condition for an unrestricted unknown unitary. The query
counts also exclude the costs of preparing the maximally entangled state,
of sampling $G$, and of implementing the known controlled operations.

\subsection{Testing equality with a fixed quantum Boolean function}
\label{subsec:optest-target}

A different task fixes $F=\sum_js_jQ_j\in\mathbf{B}$ and asks
whether an unknown quantum Boolean function $h$ equals $F$.
We now assume access to the labeled projective measurement
of $h$, with effects $(I\pm h)/2$, or to a controlled
implementation of $h$ with a specified phase reference. The latter
implements the former by a Hadamard ancilla circuit. Ordinary access to
the unitary channel of $h$ alone cannot distinguish $h=F$ from $h=-F$.

\begin{proposition}
\label{prop:optest-target}
Assume $h=h^\dagger$, $h^2=I$, and the access specified above.
The promises $h=F$ and $\|h-F\|_{2,n}\geq\epsilon$ can be
distinguished with perfect completeness and false acceptance
probability at most $\delta$ using
\( N=\lceil4\epsilon^{-2}\log(1/\delta)\rceil\)
independent trials, each of which consists of one known sector
measurement and one measurement of $h$. A two-sided Hoeffding
formulation uses $\lceil32\epsilon^{-4}\log(2/\delta)\rceil$ trials.
\end{proposition}

\begin{proof}
Prepare the maximally mixed state $I/3^n$, perform the known sector
measurement, and retain its coarse outcome $j$. This outcome has
probability $r_j/3^n$ and prepares $Q_j/r_j$ after forgetting the finer
records. For the specified projective implementation, if $E_a$ are the
orthogonal finer projectors in branch $j$, this follows from
\[
 \sum_{a\in j}E_a(I/3^n)E_a=Q_j/3^n.
\]
Measure $h$, record $Y\in\{-1,1\}$, set $X=s_jY$, and let
$c=\mathbb EX$. Then
\[
 c=3^{-n}\operatorname{Tr}(Fh),\quad
 \|h-F\|_{2,n}^2=2(1-c),\quad
 \Pr(X=-1)=\tfrac14\|h-F\|_{2,n}^2 .
\]
Here the first identity uses $\mathbb E[Y\mid j]=\operatorname{Tr}(hQ_j)/r_j$
together with the sector probabilities $r_j/3^n$, and the second uses
$\operatorname{Tr}(h^2)=\operatorname{Tr}(F^2)=3^n$. Under $h=F$ all
outcomes have $X=1$. Under the second promise each fresh trial is
negative with probability at least $\epsilon^2/4$. Rejecting on the
first negative outcome proves the first assertion. Alternatively the
two promises give $c=1$ and $c\leq1-\epsilon^2/2$, so Hoeffding's
inequality with estimation accuracy $\epsilon^2/4$ and threshold
$1-\epsilon^2/4$ bounds the error probability of each decision by
$2\exp(-N\epsilon^4/32)$ and hence gives the stated count.
\end{proof}

The preparation samples the sectors with probabilities $r_j/3^n$. It does
not require discarding trials until a rare sector appears. Nevertheless a
sign change confined to sector $j$ has normalized squared distance
$4r_j/3^n$, so resolving that change can require many trials. The promise
is formulated in the normalized Hilbert--Schmidt distance. Since that
distance can be as small as $2/\sqrt 3^n$ for operators whose distance in
operator norm is two, a promise in operator norm alone would give a
dimension-dependent cost, and no such bound is claimed here. The known
sector operations take $O(n)$ bounded-arity gates per trial under the
measurement model already stated, in addition to the cost of accessing
$h$.

\clearpage
\appendix
\section{Algebraic independence on a four-parameter family of density operators}
\label{app:independence}

The purpose of this appendix is to give a self-contained proof of the
last assertion of Theorem~\ref{thm:polynomials}.  Throughout, $p$ and
$P=p^{\otimes n}$ are the CSS projectors of
\eqref{rep:css-definition} and \eqref{rep:css-weyl-average}, and
$|\phi_x\rangle$, $x\in\F_3$, are the three-qutrit code vectors of
\eqref{rep:css-definition}.  We also use the single-qutrit operators
$Z^aX^b$ ($a,b\in\F_3$), which differ from $W_{(a,b)}$ only by
scalar cube roots of unity. The code projector satisfies
$p=\tfrac19\sum_{a,b\in\F_3}(Z^{\otimes3})^a(X^{\otimes3})^b$.
For a unit vector $|\psi\rangle\in\C^3$, let
$\kappa(\psi)=\Tr\bigl(p|\psi\rangle\langle\psi|^{\otimes3}\bigr)$.

\begin{proposition}\label{prop:independence}
For $n\geq2$ the functions $t_2,t_3,u,v$ on the affine space of trace-one Hermitian matrices are
algebraically independent.
\end{proposition}

\begin{proof}
For $k=n-2$, we let
$\rho'=\rho\otimes(I_3/3)^{\otimes k}$ for a two-qutrit state $\rho$.
Using $\Tr((I_3/3)^m)=3^{1-m}$ for $m=2,3$, we obtain
\[
 \Tr(\rho'^2)=3^{-k}\Tr(\rho^2),\quad
 \Tr(\rho'^3)=9^{-k}\Tr(\rho^3).
\]
The identity $P^{(n)}=P^{(2)}\otimes p^{\otimes k}$ and the traces
\[
 \Tr\bigl(p\,(I_3/3)^{\otimes3}\bigr)=\frac19,\quad
 \Tr\bigl(pV_{(12)}(I_3/3)^{\otimes3}\bigr)=\frac1{27}
\]
give \( u(\rho')=9^{-k}u(\rho)\) and \(v(\rho')=27^{-k}v(\rho).\)
Thus on this family the four invariants are obtained from their $n=2$
counterparts by the invertible diagonal substitution
$(t_2,t_3,u,v)\mapsto(3^{-k}t_2,9^{-k}t_3,9^{-k}u,27^{-k}v)$.
For each fixed $n\geq2$, a nonzero polynomial relation among the four
invariants would therefore induce a nonzero polynomial relation for
$n=2$. It suffices to prove the proposition in that case.

For $i=1,2$ let
\begin{equation*}
 0<\lambda_i<1,\quad
 |\psi_i\rangle=\cos\theta_i|0\rangle+\sin\theta_i|1\rangle,
 \quad 0<\theta_i<\pi/4,
\end{equation*}
and define the single-qutrit density operators
\[
 \rho_i=\frac{1-\lambda_i}{3}I_3+\lambda_i|\psi_i\rangle\langle\psi_i|,
 \quad
 \rho=\rho_1\otimes\rho_2 .
\]
Write
\[
 c_i=\cos\theta_i,\quad s_i=\sin\theta_i,\quad
 \kappa_i:=\Tr\bigl(p|\psi_i\rangle\langle\psi_i|^{\otimes3}\bigr)
 =\frac{(c_i^3+s_i^3)^2}{3},
\]
the last identity because the only code vector with nonzero overlap
with $|\psi_i\rangle^{\otimes3}$ is $|\phi_0\rangle$, giving
$\langle\phi_0|\psi_i^{\otimes3}\rangle=(c_i^3+s_i^3)/\sqrt3$.  On
$0<\theta<\pi/4$ the function
$\theta\mapsto(c^3+s^3)^2/3$ is continuous and strictly decreasing from
$1/3$ to $1/6$ (its derivative is
$2(c^3+s^3)\cdot3cs(s-c)/3<0$), so $\kappa_i$ ranges over the open
interval $(1/6,1/3)$. In particular, the value $1/4$ is attained, at the
unique $\theta_*$ with $(c_*^3+s_*^3)^2=3/4$.
Hence $(\lambda_1,\lambda_2,\kappa_1,\kappa_2)$ ranges over an open
subset of $(0,1)^2\times(1/6,1/3)^2$ that contains
$(1/2,1/3,1/4,1/4)$.
The eigenvalues of
$\rho_i$ are $(1+2\lambda_i)/3$ (once) and $(1-\lambda_i)/3$ (twice). Hence
\[
 x_i:=\Tr(\rho_i^2)=\frac{(1+2\lambda_i)^2+2(1-\lambda_i)^2}{9}
 =\frac{1+2\lambda_i^2}{3}
\]
and
\[
 y_i:=\Tr(\rho_i^3)=\frac{(1+2\lambda_i)^3+2(1-\lambda_i)^3}{27}
 =\frac{1+6\lambda_i^2+2\lambda_i^3}{9}.
\]
Since $\rho=\rho_1\otimes\rho_2$, it follows that
\begin{equation}\label{eq:app-spectral}
 t_2=\Tr(\rho^2)=x_1x_2,\quad t_3=\Tr(\rho^3)=y_1y_2 .
\end{equation}
By
$p=\tfrac19\sum_{a,b}(Z^{\otimes3})^a(X^{\otimes3})^b$ and $\Tr(Z^aX^b\rho_i)=\lambda_i\langle\psi_i|Z^aX^b|\psi_i\rangle$ for
$(a,b)\neq(0,0)$ while it equals $1$ for $(a,b)=(0,0)$, we obtain
\begin{equation}\label{eq:app-A}
 A(\lambda_i,\kappa_i):=\Tr\bigl(p\rho_i^{\otimes3}\bigr)
 =\frac19\left(1+\lambda_i^3\sum_{(a,b)\neq(0,0)}
   \langle\psi_i|Z^aX^b|\psi_i\rangle^3\right)
 =\frac{1-\lambda_i^3}{9}+\lambda_i^3\kappa_i ,
\end{equation}
where the last step uses
$\sum_{(a,b)}\langle\psi_i|Z^aX^b|\psi_i\rangle^3=9\kappa_i$ and the
fact that the $(a,b)=(0,0)$ term contributes $1$.

For the swap we use the identity
$\Tr(B_1\otimes B_2\otimes B_3\,V_{(12)})
=\Tr(B_1B_2)\Tr(B_3)$ with $B_j=Z^aX^b\rho_i$, together with
$(Z^aX^b)^2=\omega^{-ab}Z^{2a}X^{2b}$.  Since $\Tr((Z^aX^b)^2)=0$ for
$(a,b)\neq(0,0)$, the eight nontrivial labels contribute
\[
 \begin{aligned}
 &\Tr\bigl((Z^aX^b\rho_i)^2\bigr)\Tr(Z^aX^b\rho_i)\\
 &\quad=\Bigl(\frac{2\lambda_i(1-\lambda_i)}{3}
   \langle\psi_i|(Z^aX^b)^2|\psi_i\rangle
   +\lambda_i^2\langle\psi_i|Z^aX^b|\psi_i\rangle^2\Bigr)
    \lambda_i\langle\psi_i|Z^aX^b|\psi_i\rangle .
 \end{aligned}
\]
We also use the identity
\begin{equation}\label{eq:app-Sigma}
 \sum_{(a,b)\neq(0,0)}
 \langle\psi|Z^aX^b|\psi\rangle\,\langle\psi|(Z^aX^b)^2|\psi\rangle=2
 \quad\text{for every unit vector }\ |\psi\rangle=c|0\rangle+s|1\rangle,
\end{equation}
which follows from Weyl orthogonality. Indeed, by $(Z^aX^b)^3=I$, we see that
$(Z^aX^b)^2=(Z^aX^b)^\dagger$, which implies that
\[
 \sum_{(a,b)\neq(0,0)}|\langle\psi|Z^aX^b|\psi\rangle|^2
 =3\Tr\bigl((|\psi\rangle\langle\psi|)^2\bigr)-1=2.
\]
Therefore,
\begin{equation*}
 9\,B(\lambda_i,\kappa_i):=9\,\Tr\bigl(p\,V_{(12)}\rho_i^{\otimes3}\bigr)
 =\frac{1+2\lambda_i^2}{3}+\frac{4\lambda_i^2(1-\lambda_i)}{3}
 +\lambda_i^3(9\kappa_i-1),
\end{equation*}
where the three terms come from the trivial label
($\Tr(\rho_i^2)=(1+2\lambda_i^2)/3$), from the identity
\eqref{eq:app-Sigma} and
$\sum_{(a,b)\neq(0,0)}\langle\psi_i|Z^aX^b|\psi_i\rangle^3=9\kappa_i-1$.
Comparing with \eqref{eq:app-A}, we obtain
\begin{equation*}
 B(\lambda_i,\kappa_i)=A(\lambda_i,\kappa_i)-2d(\lambda_i),
 \quad
 d(\lambda)=\frac{(1-\lambda)^2(1+2\lambda)}{27},
\end{equation*}
as one checks by expanding
$(1-\lambda^3)/9-2(1-3\lambda^2+2\lambda^3)/27
=(1+6\lambda^2-7\lambda^3)/27$.
Finally, $P=p^{\otimes2}$ and
$V_{(12)}=V_{(12)}^{(1)}\otimes V_{(12)}^{(2)}$ (the copy swap acts on
the two sites separately), so
\begin{equation}\label{eq:app-uv}
 u=A(\lambda_1,\kappa_1)A(\lambda_2,\kappa_2),\quad
 v=B(\lambda_1,\kappa_1)B(\lambda_2,\kappa_2).
\end{equation}
 By \eqref{eq:app-spectral} and
\eqref{eq:app-uv}, the map
$\Phi(\lambda_1,\lambda_2,\kappa_1,\kappa_2)
=(t_2,t_3,u,v)$ has the Jacobian
\[
 J=\begin{pmatrix}
  \dfrac{\partial(t_2,t_3)}{\partial(\lambda_1,\lambda_2)} & 0\\[6pt]
  \dfrac{\partial(u,v)}{\partial(\lambda_1,\lambda_2)} &
  \dfrac{\partial(u,v)}{\partial(\kappa_1,\kappa_2)}
 \end{pmatrix},
\]
because $t_2$ and $t_3$ do not depend on $\kappa_1,\kappa_2$; the two
diagonal blocks are
\[
 J_{11}=\begin{pmatrix}
  x'(\lambda_1)x(\lambda_2) & x(\lambda_1)x'(\lambda_2)\\
  y'(\lambda_1)y(\lambda_2) & y(\lambda_1)y'(\lambda_2)
 \end{pmatrix},
 \quad
 J_{22}=\begin{pmatrix}
  \lambda_1^3A_2 & A_1\lambda_2^3\\
  \lambda_1^3B_2 & B_1\lambda_2^3
 \end{pmatrix},
\]
with $A_i=A(\lambda_i,\kappa_i)$, $B_i=B(\lambda_i,\kappa_i)$,
$x'(\lambda)=4\lambda/3$, $y'(\lambda)=(4\lambda+2\lambda^2)/3$;
here $\partial_\kappa A=\partial_\kappa B=\lambda^3$ was used.  Hence
\[
 \det J=(\det J_{11})(\det J_{22}),\quad
 \det J_{22}=\lambda_1^3\lambda_2^3\,(A_2B_1-A_1B_2).
\]
At the parameter point
$(\lambda_1,\lambda_2,\kappa_1,\kappa_2)=(1/2,1/3,1/4,1/4)$ one has
\begin{align*}
x(1/2)&=\tfrac12,\quad x(1/3)=\tfrac{11}{27},\quad
 x'(1/2)=\tfrac23,\quad x'(1/3)=\tfrac49,\\
y(1/2)&=\tfrac{11}{36},\quad y(1/3)=\tfrac{47}{243},\quad
 y'(1/2)=\tfrac56,\quad y'(1/3)=\tfrac{14}{27},\\
A_1&=\tfrac{37}{288},\quad A_2=\tfrac{113}{972},\quad
 B_1=\tfrac{79}{864},\quad B_2=\tfrac{179}{2916},
\end{align*}
and therefore
\[
 \det J_{11}
 =\frac{847}{19683}-\frac{235}{6561}=\frac{142}{19683},
\]
\[ \det J_{22}
 =\frac{1}{8}\cdot\frac{1}{27}\cdot\frac{2}{729}
 =\frac{1}{78732}
\]
and
\[
 \det J=\frac{142}{19683}\cdot\frac{1}{78732}
 =\frac{71}{774840978}\neq0 .
\]
 Since $\det J\neq0$ at
$(1/2,1/3,1/4,1/4)$, the inverse function theorem shows that $\Phi$ maps
some neighbourhood of that point diffeomorphically onto an open subset
$U\subseteq\mathbb{R}^4$.  If a nonzero polynomial $Q$ in four variables
satisfied $Q(t_2,t_3,u,v)=0$ for all states of the family, then
$Q\circ\Phi\equiv0$ near the parameter point, hence $Q\equiv0$ on the
open set $U$ and therefore $Q\equiv0$.  This proves the proposition for
$n=2$, and the initial reduction extends the conclusion to all $n\geq2$.
\end{proof}

\section{The complete character table for two qutrits}
\label{app:character-table-n2}
We use the projective Clifford group
$G_2=\mathbb F_3^4\rtimes\operatorname{Sp}(4,3)$, of order
$4\,199\,040$. Its ordinary character table is already available in the
GAP Character Table Library~\cite{ctbllib25} under the identifier
\texttt{3\^{}4.2U4(2)}. The $58\times58$ table below was computed from the
affine permutation action on $81$ points in GAP~4.15.1 \cite{GAP25}.
All the GAP code is available in the accompanying GitHub repository \cite{github_luan26}.
The accompanying file \texttt{qutrit\_n2.g} fixes the row and column labels, which need not
coincide with those in other versions of the table. Each row is an
ordinary irreducible character. Exact row orthogonality and the sum of
squared degrees were checked in GAP and independently by integer
arithmetic in $\mathbb Z[\zeta]$, where $\zeta=e^{2\pi i/3}$.

An entry $\langle a,b\rangle$ denotes $a+b\zeta$; an integer entry has
its usual meaning. All entries are exact, and the first column gives the
character degrees. The table is printed in six consecutive column panels
with the same row labels throughout.

\subsection*{Class metadata}
The identity class is $K_1$. For each class $K_j$, the table below gives
the element order $o_j$, class size $s_j$, and centralizer order $c_j$.
The machine readable GAP file contains group generators, class
representatives, and the power maps for the primes dividing $|G_2|$.
\begingroup\small
\begin{longtable}{r|rrr@{  \qquad}r|rrr}
$j$&$o_j$&$s_j$&$c_j$&$j$&$o_j$&$s_j$&$c_j$\\\hline\endhead
1 & 1 & 1 & 4199040 & 30 & 12 & 174960 & 24\\
2 & 3 & 80 & 52488 & 31 & 3 & 17280 & 243\\
3 & 2 & 81 & 51840 & 32 & 3 & 8640 & 486\\
4 & 6 & 6480 & 648 & 33 & 3 & 8640 & 486\\
5 & 2 & 810 & 5184 & 34 & 3 & 4320 & 972\\
6 & 6 & 3240 & 1296 & 35 & 6 & 38880 & 108\\
7 & 3 & 2160 & 1944 & 36 & 18 & 233280 & 18\\
8 & 3 & 960 & 4374 & 37 & 9 & 77760 & 54\\
9 & 3 & 120 & 34992 & 38 & 9 & 155520 & 27\\
10 & 6 & 3240 & 1296 & 39 & 9 & 155520 & 27\\
11 & 3 & 960 & 4374 & 40 & 9 & 77760 & 54\\
12 & 3 & 2160 & 1944 & 41 & 18 & 233280 & 18\\
13 & 3 & 120 & 34992 & 42 & 4 & 43740 & 96\\
14 & 12 & 38880 & 108 & 43 & 3 & 2160 & 1944\\
15 & 4 & 4860 & 864 & 44 & 3 & 8640 & 486\\
16 & 4 & 43740 & 96 & 45 & 3 & 8640 & 486\\
17 & 6 & 25920 & 162 & 46 & 6 & 19440 & 216\\
18 & 6 & 3240 & 1296 & 47 & 12 & 349920 & 12\\
19 & 6 & 19440 & 216 & 48 & 6 & 77760 & 54\\
20 & 6 & 9720 & 432 & 49 & 6 & 38880 & 108\\
21 & 6 & 25920 & 162 & 50 & 6 & 77760 & 54\\
22 & 6 & 3240 & 1296 & 51 & 6 & 38880 & 108\\
23 & 6 & 19440 & 216 & 52 & 6 & 77760 & 54\\
24 & 6 & 9720 & 432 & 53 & 6 & 38880 & 108\\
25 & 12 & 174960 & 24 & 54 & 6 & 77760 & 54\\
26 & 12 & 116640 & 36 & 55 & 6 & 38880 & 108\\
27 & 12 & 58320 & 72 & 56 & 8 & 524880 & 8\\
28 & 12 & 116640 & 36 & 57 & 10 & 419904 & 10\\
29 & 12 & 58320 & 72 & 58 & 5 & 419904 & 10\\
\end{longtable}
\endgroup
\clearpage
\charactertableheading{Character values: columns \texorpdfstring{$K_{1},\ldots,K_{10}$}{K1--K10}}
\begingroup
\fontsize{8}{9.4}\selectfont\setlength{\tabcolsep}{2.1pt}
\renewcommand{\arraystretch}{1.02}
\begin{longtable}{r|rrrrrrrrrr}
 & $K_{1}$ & $K_{2}$ & $K_{3}$ & $K_{4}$ & $K_{5}$ & $K_{6}$ & $K_{7}$ & $K_{8}$ & $K_{9}$ & $K_{10}$\\\hline\endhead
$\chi_{1}$ & $1$ & $1$ & $1$ & $1$ & $1$ & $1$ & $1$ & $1$ & $1$ & $1$\\
$\chi_{2}$ & $4$ & $4$ & $-4$ & $0$ & $0$ & $\langle 2,3\rangle$ & $\langle -2,-3\rangle$ & $\langle -2,-3\rangle$ & $\langle -2,-3\rangle$ & $\langle -1,-3\rangle$\\
$\chi_{3}$ & $4$ & $4$ & $-4$ & $0$ & $0$ & $\langle -1,-3\rangle$ & $\langle 1,3\rangle$ & $\langle 1,3\rangle$ & $\langle 1,3\rangle$ & $\langle 2,3\rangle$\\
$\chi_{4}$ & $5$ & $5$ & $5$ & $-3$ & $-3$ & $\langle -1,-3\rangle$ & $\langle -1,-3\rangle$ & $\langle -1,-3\rangle$ & $\langle -1,-3\rangle$ & $\langle 2,3\rangle$\\
$\chi_{5}$ & $5$ & $5$ & $5$ & $-3$ & $-3$ & $\langle 2,3\rangle$ & $\langle 2,3\rangle$ & $\langle 2,3\rangle$ & $\langle 2,3\rangle$ & $\langle -1,-3\rangle$\\
$\chi_{6}$ & $6$ & $6$ & $6$ & $-2$ & $-2$ & $-3$ & $-3$ & $-3$ & $-3$ & $-3$\\
$\chi_{7}$ & $10$ & $10$ & $10$ & $2$ & $2$ & $\langle -2,3\rangle$ & $\langle -2,3\rangle$ & $\langle -2,3\rangle$ & $\langle -2,3\rangle$ & $\langle -5,-3\rangle$\\
$\chi_{8}$ & $10$ & $10$ & $10$ & $2$ & $2$ & $\langle -5,-3\rangle$ & $\langle -5,-3\rangle$ & $\langle -5,-3\rangle$ & $\langle -5,-3\rangle$ & $\langle -2,3\rangle$\\
$\chi_{9}$ & $15$ & $15$ & $15$ & $7$ & $7$ & $-3$ & $-3$ & $-3$ & $-3$ & $-3$\\
$\chi_{10}$ & $15$ & $15$ & $15$ & $-1$ & $-1$ & $6$ & $6$ & $6$ & $6$ & $6$\\
$\chi_{11}$ & $20$ & $20$ & $-20$ & $0$ & $0$ & $7$ & $-7$ & $-7$ & $-7$ & $7$\\
$\chi_{12}$ & $20$ & $20$ & $20$ & $4$ & $4$ & $2$ & $2$ & $2$ & $2$ & $2$\\
$\chi_{13}$ & $20$ & $20$ & $-20$ & $0$ & $0$ & $\langle -8,-3\rangle$ & $\langle 8,3\rangle$ & $\langle 8,3\rangle$ & $\langle 8,3\rangle$ & $\langle -5,3\rangle$\\
$\chi_{14}$ & $20$ & $20$ & $-20$ & $0$ & $0$ & $\langle -5,3\rangle$ & $\langle 5,-3\rangle$ & $\langle 5,-3\rangle$ & $\langle 5,-3\rangle$ & $\langle -8,-3\rangle$\\
$\chi_{15}$ & $20$ & $20$ & $-20$ & $0$ & $0$ & $\langle -5,-6\rangle$ & $\langle 5,6\rangle$ & $\langle 5,6\rangle$ & $\langle 5,6\rangle$ & $\langle 1,6\rangle$\\
$\chi_{16}$ & $20$ & $20$ & $-20$ & $0$ & $0$ & $\langle 1,6\rangle$ & $\langle -1,-6\rangle$ & $\langle -1,-6\rangle$ & $\langle -1,-6\rangle$ & $\langle -5,-6\rangle$\\
$\chi_{17}$ & $24$ & $24$ & $24$ & $8$ & $8$ & $6$ & $6$ & $6$ & $6$ & $6$\\
$\chi_{18}$ & $30$ & $30$ & $30$ & $-10$ & $-10$ & $3$ & $3$ & $3$ & $3$ & $3$\\
$\chi_{19}$ & $30$ & $30$ & $30$ & $6$ & $6$ & $\langle 3,9\rangle$ & $\langle 3,9\rangle$ & $\langle 3,9\rangle$ & $\langle 3,9\rangle$ & $\langle -6,-9\rangle$\\
$\chi_{20}$ & $30$ & $30$ & $30$ & $6$ & $6$ & $\langle -6,-9\rangle$ & $\langle -6,-9\rangle$ & $\langle -6,-9\rangle$ & $\langle -6,-9\rangle$ & $\langle 3,9\rangle$\\
$\chi_{21}$ & $36$ & $36$ & $-36$ & $0$ & $0$ & $\langle 9,9\rangle$ & $\langle -9,-9\rangle$ & $\langle -9,-9\rangle$ & $\langle -9,-9\rangle$ & $\langle 0,-9\rangle$\\
$\chi_{22}$ & $36$ & $36$ & $-36$ & $0$ & $0$ & $\langle 0,-9\rangle$ & $\langle 0,9\rangle$ & $\langle 0,9\rangle$ & $\langle 0,9\rangle$ & $\langle 9,9\rangle$\\
$\chi_{23}$ & $40$ & $40$ & $40$ & $-8$ & $-8$ & $\langle -8,-6\rangle$ & $\langle -8,-6\rangle$ & $\langle -8,-6\rangle$ & $\langle -8,-6\rangle$ & $\langle -2,6\rangle$\\
$\chi_{24}$ & $40$ & $40$ & $40$ & $-8$ & $-8$ & $\langle -2,6\rangle$ & $\langle -2,6\rangle$ & $\langle -2,6\rangle$ & $\langle -2,6\rangle$ & $\langle -8,-6\rangle$\\
$\chi_{25}$ & $45$ & $45$ & $45$ & $-3$ & $-3$ & $\langle 9,9\rangle$ & $\langle 9,9\rangle$ & $\langle 9,9\rangle$ & $\langle 9,9\rangle$ & $\langle 0,-9\rangle$\\
$\chi_{26}$ & $45$ & $45$ & $45$ & $-3$ & $-3$ & $\langle 0,-9\rangle$ & $\langle 0,-9\rangle$ & $\langle 0,-9\rangle$ & $\langle 0,-9\rangle$ & $\langle 9,9\rangle$\\
$\chi_{27}$ & $60$ & $60$ & $-60$ & $0$ & $0$ & $3$ & $-3$ & $-3$ & $-3$ & $3$\\
$\chi_{28}$ & $60$ & $60$ & $60$ & $-4$ & $-4$ & $6$ & $6$ & $6$ & $6$ & $6$\\
$\chi_{29}$ & $60$ & $60$ & $-60$ & $0$ & $0$ & $\langle 3,9\rangle$ & $\langle -3,-9\rangle$ & $\langle -3,-9\rangle$ & $\langle -3,-9\rangle$ & $\langle -6,-9\rangle$\\
$\chi_{30}$ & $60$ & $60$ & $-60$ & $0$ & $0$ & $\langle -6,-9\rangle$ & $\langle 6,9\rangle$ & $\langle 6,9\rangle$ & $\langle 6,9\rangle$ & $\langle 3,9\rangle$\\
$\chi_{31}$ & $64$ & $64$ & $-64$ & $0$ & $0$ & $8$ & $-8$ & $-8$ & $-8$ & $8$\\
$\chi_{32}$ & $64$ & $64$ & $64$ & $0$ & $0$ & $-8$ & $-8$ & $-8$ & $-8$ & $-8$\\
$\chi_{33}$ & $80$ & $-1$ & $0$ & $-1$ & $8$ & $0$ & $-1$ & $-1$ & $26$ & $0$\\
$\chi_{34}$ & $80$ & $80$ & $-80$ & $0$ & $0$ & $-8$ & $8$ & $8$ & $8$ & $-8$\\
$\chi_{35}$ & $80$ & $-1$ & $0$ & $-1$ & $8$ & $0$ & $-1$ & $\langle 5,3\rangle$ & $\langle -22,-24\rangle$ & $0$\\
$\chi_{36}$ & $80$ & $-1$ & $0$ & $-1$ & $8$ & $0$ & $-1$ & $\langle 2,-3\rangle$ & $\langle 2,24\rangle$ & $0$\\
$\chi_{37}$ & $81$ & $81$ & $81$ & $9$ & $9$ & $0$ & $0$ & $0$ & $0$ & $0$\\
$\chi_{38}$ & $160$ & $-2$ & $0$ & $2$ & $-16$ & $0$ & $-2$ & $7$ & $-20$ & $0$\\
$\chi_{39}$ & $160$ & $-2$ & $0$ & $2$ & $-16$ & $0$ & $-2$ & $\langle 4,3\rangle$ & $\langle 4,-24\rangle$ & $0$\\
$\chi_{40}$ & $160$ & $-2$ & $0$ & $2$ & $-16$ & $0$ & $-2$ & $\langle 1,-3\rangle$ & $\langle 28,24\rangle$ & $0$\\
$\chi_{41}$ & $240$ & $-3$ & $0$ & $-3$ & $24$ & $0$ & $-3$ & $6$ & $6$ & $0$\\
$\chi_{42}$ & $240$ & $-3$ & $0$ & $1$ & $-8$ & $0$ & $\langle 3,3\rangle$ & $\langle -3,-9\rangle$ & $\langle -30,18\rangle$ & $0$\\
$\chi_{43}$ & $240$ & $-3$ & $0$ & $1$ & $-8$ & $0$ & $\langle 0,-3\rangle$ & $\langle 6,9\rangle$ & $\langle -48,-18\rangle$ & $0$\\
$\chi_{44}$ & $240$ & $-3$ & $0$ & $1$ & $-8$ & $0$ & $\langle 3,3\rangle$ & $\langle -12,-9\rangle$ & $\langle 42,18\rangle$ & $0$\\
$\chi_{45}$ & $240$ & $-3$ & $0$ & $1$ & $-8$ & $0$ & $\langle 0,-3\rangle$ & $\langle -3,9\rangle$ & $\langle 24,-18\rangle$ & $0$\\
$\chi_{46}$ & $240$ & $-3$ & $0$ & $1$ & $-8$ & $0$ & $\langle 3,3\rangle$ & $-3$ & $\langle -30,-54\rangle$ & $0$\\
$\chi_{47}$ & $240$ & $-3$ & $0$ & $1$ & $-8$ & $0$ & $\langle 0,-3\rangle$ & $-3$ & $\langle 24,54\rangle$ & $0$\\
$\chi_{48}$ & $480$ & $-6$ & $0$ & $-2$ & $16$ & $0$ & $\langle 6,6\rangle$ & $\langle -15,-9\rangle$ & $\langle 12,-36\rangle$ & $0$\\
$\chi_{49}$ & $480$ & $-6$ & $0$ & $-2$ & $16$ & $0$ & $\langle 0,-6\rangle$ & $\langle -6,9\rangle$ & $\langle 48,36\rangle$ & $0$\\
$\chi_{50}$ & $480$ & $-6$ & $0$ & $-2$ & $16$ & $0$ & $\langle 0,-6\rangle$ & $\langle 3,9\rangle$ & $\langle -24,36\rangle$ & $0$\\
$\chi_{51}$ & $480$ & $-6$ & $0$ & $-2$ & $16$ & $0$ & $\langle 6,6\rangle$ & $\langle -6,-9\rangle$ & $\langle -60,-36\rangle$ & $0$\\
$\chi_{52}$ & $480$ & $-6$ & $0$ & $-2$ & $16$ & $0$ & $\langle 0,-6\rangle$ & $\langle 3,18\rangle$ & $\langle -24,-36\rangle$ & $0$\\
$\chi_{53}$ & $480$ & $-6$ & $0$ & $-2$ & $16$ & $0$ & $\langle 6,6\rangle$ & $\langle -15,-18\rangle$ & $\langle 12,36\rangle$ & $0$\\
$\chi_{54}$ & $640$ & $-8$ & $0$ & $0$ & $0$ & $0$ & $-8$ & $10$ & $64$ & $0$\\
$\chi_{55}$ & $640$ & $-8$ & $0$ & $0$ & $0$ & $0$ & $-8$ & $\langle 16,-6\rangle$ & $\langle 16,48\rangle$ & $0$\\
$\chi_{56}$ & $640$ & $-8$ & $0$ & $0$ & $0$ & $0$ & $-8$ & $\langle 22,6\rangle$ & $\langle -32,-48\rangle$ & $0$\\
$\chi_{57}$ & $720$ & $-9$ & $0$ & $3$ & $-24$ & $0$ & $\langle 9,9\rangle$ & $\langle -18,-18\rangle$ & $\langle -18,-18\rangle$ & $0$\\
$\chi_{58}$ & $720$ & $-9$ & $0$ & $3$ & $-24$ & $0$ & $\langle 0,-9\rangle$ & $\langle 0,18\rangle$ & $\langle 0,18\rangle$ & $0$\\
\end{longtable}
\endgroup
\clearpage
\charactertableheading{Character values: columns \texorpdfstring{$K_{11},\ldots,K_{20}$}{K11--K20}}
\begingroup
\fontsize{8}{9.4}\selectfont\setlength{\tabcolsep}{2.1pt}
\renewcommand{\arraystretch}{1.02}
\begin{longtable}{r|rrrrrrrrrr}
 & $K_{11}$ & $K_{12}$ & $K_{13}$ & $K_{14}$ & $K_{15}$ & $K_{16}$ & $K_{17}$ & $K_{18}$ & $K_{19}$ & $K_{20}$\\\hline\endhead
$\chi_{1}$ & $1$ & $1$ & $1$ & $1$ & $1$ & $1$ & $1$ & $1$ & $1$ & $1$\\
$\chi_{2}$ & $\langle 1,3\rangle$ & $\langle 1,3\rangle$ & $\langle 1,3\rangle$ & $2$ & $2$ & $-2$ & $\langle -1,1\rangle$ & $\langle -1,1\rangle$ & $\langle 1,-1\rangle$ & $\langle 1,-1\rangle$\\
$\chi_{3}$ & $\langle -2,-3\rangle$ & $\langle -2,-3\rangle$ & $\langle -2,-3\rangle$ & $2$ & $2$ & $-2$ & $\langle -2,-1\rangle$ & $\langle -2,-1\rangle$ & $\langle 2,1\rangle$ & $\langle 2,1\rangle$\\
$\chi_{4}$ & $\langle 2,3\rangle$ & $\langle 2,3\rangle$ & $\langle 2,3\rangle$ & $1$ & $1$ & $1$ & $\langle -2,-1\rangle$ & $\langle -2,-1\rangle$ & $\langle -2,-1\rangle$ & $\langle -2,-1\rangle$\\
$\chi_{5}$ & $\langle -1,-3\rangle$ & $\langle -1,-3\rangle$ & $\langle -1,-3\rangle$ & $1$ & $1$ & $1$ & $\langle -1,1\rangle$ & $\langle -1,1\rangle$ & $\langle -1,1\rangle$ & $\langle -1,1\rangle$\\
$\chi_{6}$ & $-3$ & $-3$ & $-3$ & $2$ & $2$ & $2$ & $1$ & $1$ & $1$ & $1$\\
$\chi_{7}$ & $\langle -5,-3\rangle$ & $\langle -5,-3\rangle$ & $\langle -5,-3\rangle$ & $2$ & $2$ & $2$ & $\langle -1,-3\rangle$ & $\langle -1,-3\rangle$ & $\langle -1,-3\rangle$ & $\langle -1,-3\rangle$\\
$\chi_{8}$ & $\langle -2,3\rangle$ & $\langle -2,3\rangle$ & $\langle -2,3\rangle$ & $2$ & $2$ & $2$ & $\langle 2,3\rangle$ & $\langle 2,3\rangle$ & $\langle 2,3\rangle$ & $\langle 2,3\rangle$\\
$\chi_{9}$ & $-3$ & $-3$ & $-3$ & $-1$ & $-1$ & $-1$ & $1$ & $1$ & $1$ & $1$\\
$\chi_{10}$ & $6$ & $6$ & $6$ & $3$ & $3$ & $3$ & $2$ & $2$ & $2$ & $2$\\
$\chi_{11}$ & $-7$ & $-7$ & $-7$ & $2$ & $2$ & $-2$ & $3$ & $3$ & $-3$ & $-3$\\
$\chi_{12}$ & $2$ & $2$ & $2$ & $0$ & $0$ & $0$ & $-2$ & $-2$ & $-2$ & $-2$\\
$\chi_{13}$ & $\langle 5,-3\rangle$ & $\langle 5,-3\rangle$ & $\langle 5,-3\rangle$ & $2$ & $2$ & $-2$ & $\langle 3,3\rangle$ & $\langle 3,3\rangle$ & $\langle -3,-3\rangle$ & $\langle -3,-3\rangle$\\
$\chi_{14}$ & $\langle 8,3\rangle$ & $\langle 8,3\rangle$ & $\langle 8,3\rangle$ & $2$ & $2$ & $-2$ & $\langle 0,-3\rangle$ & $\langle 0,-3\rangle$ & $\langle 0,3\rangle$ & $\langle 0,3\rangle$\\
$\chi_{15}$ & $\langle -1,-6\rangle$ & $\langle -1,-6\rangle$ & $\langle -1,-6\rangle$ & $2$ & $2$ & $-2$ & $\langle 1,2\rangle$ & $\langle 1,2\rangle$ & $\langle -1,-2\rangle$ & $\langle -1,-2\rangle$\\
$\chi_{16}$ & $\langle 5,6\rangle$ & $\langle 5,6\rangle$ & $\langle 5,6\rangle$ & $2$ & $2$ & $-2$ & $\langle -1,-2\rangle$ & $\langle -1,-2\rangle$ & $\langle 1,2\rangle$ & $\langle 1,2\rangle$\\
$\chi_{17}$ & $6$ & $6$ & $6$ & $0$ & $0$ & $0$ & $2$ & $2$ & $2$ & $2$\\
$\chi_{18}$ & $3$ & $3$ & $3$ & $-2$ & $-2$ & $-2$ & $-1$ & $-1$ & $-1$ & $-1$\\
$\chi_{19}$ & $\langle -6,-9\rangle$ & $\langle -6,-9\rangle$ & $\langle -6,-9\rangle$ & $2$ & $2$ & $2$ & $\langle -2,-1\rangle$ & $\langle -2,-1\rangle$ & $\langle -2,-1\rangle$ & $\langle -2,-1\rangle$\\
$\chi_{20}$ & $\langle 3,9\rangle$ & $\langle 3,9\rangle$ & $\langle 3,9\rangle$ & $2$ & $2$ & $2$ & $\langle -1,1\rangle$ & $\langle -1,1\rangle$ & $\langle -1,1\rangle$ & $\langle -1,1\rangle$\\
$\chi_{21}$ & $\langle 0,9\rangle$ & $\langle 0,9\rangle$ & $\langle 0,9\rangle$ & $2$ & $2$ & $-2$ & $\langle 0,3\rangle$ & $\langle 0,3\rangle$ & $\langle 0,-3\rangle$ & $\langle 0,-3\rangle$\\
$\chi_{22}$ & $\langle -9,-9\rangle$ & $\langle -9,-9\rangle$ & $\langle -9,-9\rangle$ & $2$ & $2$ & $-2$ & $\langle -3,-3\rangle$ & $\langle -3,-3\rangle$ & $\langle 3,3\rangle$ & $\langle 3,3\rangle$\\
$\chi_{23}$ & $\langle -2,6\rangle$ & $\langle -2,6\rangle$ & $\langle -2,6\rangle$ & $0$ & $0$ & $0$ & $\langle 2,2\rangle$ & $\langle 2,2\rangle$ & $\langle 2,2\rangle$ & $\langle 2,2\rangle$\\
$\chi_{24}$ & $\langle -8,-6\rangle$ & $\langle -8,-6\rangle$ & $\langle -8,-6\rangle$ & $0$ & $0$ & $0$ & $\langle 0,-2\rangle$ & $\langle 0,-2\rangle$ & $\langle 0,-2\rangle$ & $\langle 0,-2\rangle$\\
$\chi_{25}$ & $\langle 0,-9\rangle$ & $\langle 0,-9\rangle$ & $\langle 0,-9\rangle$ & $1$ & $1$ & $1$ & $\langle 0,3\rangle$ & $\langle 0,3\rangle$ & $\langle 0,3\rangle$ & $\langle 0,3\rangle$\\
$\chi_{26}$ & $\langle 9,9\rangle$ & $\langle 9,9\rangle$ & $\langle 9,9\rangle$ & $1$ & $1$ & $1$ & $\langle -3,-3\rangle$ & $\langle -3,-3\rangle$ & $\langle -3,-3\rangle$ & $\langle -3,-3\rangle$\\
$\chi_{27}$ & $-3$ & $-3$ & $-3$ & $-2$ & $-2$ & $2$ & $3$ & $3$ & $-3$ & $-3$\\
$\chi_{28}$ & $6$ & $6$ & $6$ & $0$ & $0$ & $0$ & $2$ & $2$ & $2$ & $2$\\
$\chi_{29}$ & $\langle 6,9\rangle$ & $\langle 6,9\rangle$ & $\langle 6,9\rangle$ & $-2$ & $-2$ & $2$ & $\langle -2,-1\rangle$ & $\langle -2,-1\rangle$ & $\langle 2,1\rangle$ & $\langle 2,1\rangle$\\
$\chi_{30}$ & $\langle -3,-9\rangle$ & $\langle -3,-9\rangle$ & $\langle -3,-9\rangle$ & $-2$ & $-2$ & $2$ & $\langle -1,1\rangle$ & $\langle -1,1\rangle$ & $\langle 1,-1\rangle$ & $\langle 1,-1\rangle$\\
$\chi_{31}$ & $-8$ & $-8$ & $-8$ & $0$ & $0$ & $0$ & $0$ & $0$ & $0$ & $0$\\
$\chi_{32}$ & $-8$ & $-8$ & $-8$ & $0$ & $0$ & $0$ & $0$ & $0$ & $0$ & $0$\\
$\chi_{33}$ & $-1$ & $-1$ & $26$ & $-1$ & $8$ & $0$ & $-1$ & $8$ & $-1$ & $2$\\
$\chi_{34}$ & $8$ & $8$ & $8$ & $0$ & $0$ & $0$ & $0$ & $0$ & $0$ & $0$\\
$\chi_{35}$ & $\langle 2,-3\rangle$ & $-1$ & $\langle 2,24\rangle$ & $-1$ & $8$ & $0$ & $\langle 0,-1\rangle$ & $\langle 0,8\rangle$ & $-1$ & $2$\\
$\chi_{36}$ & $\langle 5,3\rangle$ & $-1$ & $\langle -22,-24\rangle$ & $-1$ & $8$ & $0$ & $\langle 1,1\rangle$ & $\langle -8,-8\rangle$ & $-1$ & $2$\\
$\chi_{37}$ & $0$ & $0$ & $0$ & $-3$ & $-3$ & $-3$ & $0$ & $0$ & $0$ & $0$\\
$\chi_{38}$ & $7$ & $-2$ & $-20$ & $0$ & $0$ & $0$ & $-1$ & $8$ & $2$ & $-4$\\
$\chi_{39}$ & $\langle 1,-3\rangle$ & $-2$ & $\langle 28,24\rangle$ & $0$ & $0$ & $0$ & $\langle 1,1\rangle$ & $\langle -8,-8\rangle$ & $2$ & $-4$\\
$\chi_{40}$ & $\langle 4,3\rangle$ & $-2$ & $\langle 4,-24\rangle$ & $0$ & $0$ & $0$ & $\langle 0,-1\rangle$ & $\langle 0,8\rangle$ & $2$ & $-4$\\
$\chi_{41}$ & $6$ & $-3$ & $6$ & $1$ & $-8$ & $0$ & $0$ & $0$ & $-3$ & $6$\\
$\chi_{42}$ & $\langle 6,9\rangle$ & $\langle 0,-3\rangle$ & $\langle -48,-18\rangle$ & $-1$ & $8$ & $0$ & $\langle 0,1\rangle$ & $\langle 0,-8\rangle$ & $\langle 0,1\rangle$ & $\langle 0,-2\rangle$\\
$\chi_{43}$ & $\langle -3,-9\rangle$ & $\langle 3,3\rangle$ & $\langle -30,18\rangle$ & $-1$ & $8$ & $0$ & $\langle -1,-1\rangle$ & $\langle 8,8\rangle$ & $\langle -1,-1\rangle$ & $\langle 2,2\rangle$\\
$\chi_{44}$ & $\langle -3,9\rangle$ & $\langle 0,-3\rangle$ & $\langle 24,-18\rangle$ & $-1$ & $8$ & $0$ & $\langle -1,-1\rangle$ & $\langle 8,8\rangle$ & $\langle 0,1\rangle$ & $\langle 0,-2\rangle$\\
$\chi_{45}$ & $\langle -12,-9\rangle$ & $\langle 3,3\rangle$ & $\langle 42,18\rangle$ & $-1$ & $8$ & $0$ & $\langle 0,1\rangle$ & $\langle 0,-8\rangle$ & $\langle -1,-1\rangle$ & $\langle 2,2\rangle$\\
$\chi_{46}$ & $-3$ & $\langle 0,-3\rangle$ & $\langle 24,54\rangle$ & $-1$ & $8$ & $0$ & $1$ & $-8$ & $\langle 0,1\rangle$ & $\langle 0,-2\rangle$\\
$\chi_{47}$ & $-3$ & $\langle 3,3\rangle$ & $\langle -30,-54\rangle$ & $-1$ & $8$ & $0$ & $1$ & $-8$ & $\langle -1,-1\rangle$ & $\langle 2,2\rangle$\\
$\chi_{48}$ & $\langle -6,9\rangle$ & $\langle 0,-6\rangle$ & $\langle 48,36\rangle$ & $0$ & $0$ & $0$ & $\langle 0,1\rangle$ & $\langle 0,-8\rangle$ & $\langle 0,-2\rangle$ & $\langle 0,4\rangle$\\
$\chi_{49}$ & $\langle -15,-9\rangle$ & $\langle 6,6\rangle$ & $\langle 12,-36\rangle$ & $0$ & $0$ & $0$ & $\langle -1,-1\rangle$ & $\langle 8,8\rangle$ & $\langle 2,2\rangle$ & $\langle -4,-4\rangle$\\
$\chi_{50}$ & $\langle -6,-9\rangle$ & $\langle 6,6\rangle$ & $\langle -60,-36\rangle$ & $0$ & $0$ & $0$ & $\langle 0,1\rangle$ & $\langle 0,-8\rangle$ & $\langle 2,2\rangle$ & $\langle -4,-4\rangle$\\
$\chi_{51}$ & $\langle 3,9\rangle$ & $\langle 0,-6\rangle$ & $\langle -24,36\rangle$ & $0$ & $0$ & $0$ & $\langle -1,-1\rangle$ & $\langle 8,8\rangle$ & $\langle 0,-2\rangle$ & $\langle 0,4\rangle$\\
$\chi_{52}$ & $\langle -15,-18\rangle$ & $\langle 6,6\rangle$ & $\langle 12,36\rangle$ & $0$ & $0$ & $0$ & $1$ & $-8$ & $\langle 2,2\rangle$ & $\langle -4,-4\rangle$\\
$\chi_{53}$ & $\langle 3,18\rangle$ & $\langle 0,-6\rangle$ & $\langle -24,-36\rangle$ & $0$ & $0$ & $0$ & $1$ & $-8$ & $\langle 0,-2\rangle$ & $\langle 0,4\rangle$\\
$\chi_{54}$ & $10$ & $-8$ & $64$ & $0$ & $0$ & $0$ & $0$ & $0$ & $0$ & $0$\\
$\chi_{55}$ & $\langle 22,6\rangle$ & $-8$ & $\langle -32,-48\rangle$ & $0$ & $0$ & $0$ & $0$ & $0$ & $0$ & $0$\\
$\chi_{56}$ & $\langle 16,-6\rangle$ & $-8$ & $\langle 16,48\rangle$ & $0$ & $0$ & $0$ & $0$ & $0$ & $0$ & $0$\\
$\chi_{57}$ & $\langle 0,18\rangle$ & $\langle 0,-9\rangle$ & $\langle 0,18\rangle$ & $1$ & $-8$ & $0$ & $0$ & $0$ & $\langle 0,3\rangle$ & $\langle 0,-6\rangle$\\
$\chi_{58}$ & $\langle -18,-18\rangle$ & $\langle 9,9\rangle$ & $\langle -18,-18\rangle$ & $1$ & $-8$ & $0$ & $0$ & $0$ & $\langle -3,-3\rangle$ & $\langle 6,6\rangle$\\
\end{longtable}
\endgroup
\clearpage
\charactertableheading{Character values: columns \texorpdfstring{$K_{21},\ldots,K_{30}$}{K21--K30}}
\begingroup
\fontsize{8}{9.4}\selectfont\setlength{\tabcolsep}{2.1pt}
\renewcommand{\arraystretch}{1.02}
\begin{longtable}{r|rrrrrrrrrr}
 & $K_{21}$ & $K_{22}$ & $K_{23}$ & $K_{24}$ & $K_{25}$ & $K_{26}$ & $K_{27}$ & $K_{28}$ & $K_{29}$ & $K_{30}$\\\hline\endhead
$\chi_{1}$ & $1$ & $1$ & $1$ & $1$ & $1$ & $1$ & $1$ & $1$ & $1$ & $1$\\
$\chi_{2}$ & $\langle -2,-1\rangle$ & $\langle -2,-1\rangle$ & $\langle 2,1\rangle$ & $\langle 2,1\rangle$ & $\langle 0,1\rangle$ & $\langle 0,-1\rangle$ & $\langle 0,-1\rangle$ & $\langle 1,1\rangle$ & $\langle 1,1\rangle$ & $\langle -1,-1\rangle$\\
$\chi_{3}$ & $\langle -1,1\rangle$ & $\langle -1,1\rangle$ & $\langle 1,-1\rangle$ & $\langle 1,-1\rangle$ & $\langle -1,-1\rangle$ & $\langle 1,1\rangle$ & $\langle 1,1\rangle$ & $\langle 0,-1\rangle$ & $\langle 0,-1\rangle$ & $\langle 0,1\rangle$\\
$\chi_{4}$ & $\langle -1,1\rangle$ & $\langle -1,1\rangle$ & $\langle -1,1\rangle$ & $\langle -1,1\rangle$ & $\langle -1,-1\rangle$ & $\langle -1,-1\rangle$ & $\langle -1,-1\rangle$ & $\langle 0,1\rangle$ & $\langle 0,1\rangle$ & $\langle 0,1\rangle$\\
$\chi_{5}$ & $\langle -2,-1\rangle$ & $\langle -2,-1\rangle$ & $\langle -2,-1\rangle$ & $\langle -2,-1\rangle$ & $\langle 0,1\rangle$ & $\langle 0,1\rangle$ & $\langle 0,1\rangle$ & $\langle -1,-1\rangle$ & $\langle -1,-1\rangle$ & $\langle -1,-1\rangle$\\
$\chi_{6}$ & $1$ & $1$ & $1$ & $1$ & $-1$ & $-1$ & $-1$ & $-1$ & $-1$ & $-1$\\
$\chi_{7}$ & $\langle 2,3\rangle$ & $\langle 2,3\rangle$ & $\langle 2,3\rangle$ & $\langle 2,3\rangle$ & $\langle 0,-1\rangle$ & $\langle 0,-1\rangle$ & $\langle 0,-1\rangle$ & $\langle 1,1\rangle$ & $\langle 1,1\rangle$ & $\langle 1,1\rangle$\\
$\chi_{8}$ & $\langle -1,-3\rangle$ & $\langle -1,-3\rangle$ & $\langle -1,-3\rangle$ & $\langle -1,-3\rangle$ & $\langle 1,1\rangle$ & $\langle 1,1\rangle$ & $\langle 1,1\rangle$ & $\langle 0,-1\rangle$ & $\langle 0,-1\rangle$ & $\langle 0,-1\rangle$\\
$\chi_{9}$ & $1$ & $1$ & $1$ & $1$ & $-1$ & $-1$ & $-1$ & $-1$ & $-1$ & $-1$\\
$\chi_{10}$ & $2$ & $2$ & $2$ & $2$ & $0$ & $0$ & $0$ & $0$ & $0$ & $0$\\
$\chi_{11}$ & $3$ & $3$ & $-3$ & $-3$ & $1$ & $-1$ & $-1$ & $-1$ & $-1$ & $1$\\
$\chi_{12}$ & $-2$ & $-2$ & $-2$ & $-2$ & $0$ & $0$ & $0$ & $0$ & $0$ & $0$\\
$\chi_{13}$ & $\langle 0,-3\rangle$ & $\langle 0,-3\rangle$ & $\langle 0,3\rangle$ & $\langle 0,3\rangle$ & $\langle 0,1\rangle$ & $\langle 0,-1\rangle$ & $\langle 0,-1\rangle$ & $\langle 1,1\rangle$ & $\langle 1,1\rangle$ & $\langle -1,-1\rangle$\\
$\chi_{14}$ & $\langle 3,3\rangle$ & $\langle 3,3\rangle$ & $\langle -3,-3\rangle$ & $\langle -3,-3\rangle$ & $\langle -1,-1\rangle$ & $\langle 1,1\rangle$ & $\langle 1,1\rangle$ & $\langle 0,-1\rangle$ & $\langle 0,-1\rangle$ & $\langle 0,1\rangle$\\
$\chi_{15}$ & $\langle -1,-2\rangle$ & $\langle -1,-2\rangle$ & $\langle 1,2\rangle$ & $\langle 1,2\rangle$ & $1$ & $-1$ & $-1$ & $-1$ & $-1$ & $1$\\
$\chi_{16}$ & $\langle 1,2\rangle$ & $\langle 1,2\rangle$ & $\langle -1,-2\rangle$ & $\langle -1,-2\rangle$ & $1$ & $-1$ & $-1$ & $-1$ & $-1$ & $1$\\
$\chi_{17}$ & $2$ & $2$ & $2$ & $2$ & $0$ & $0$ & $0$ & $0$ & $0$ & $0$\\
$\chi_{18}$ & $-1$ & $-1$ & $-1$ & $-1$ & $1$ & $1$ & $1$ & $1$ & $1$ & $1$\\
$\chi_{19}$ & $\langle -1,1\rangle$ & $\langle -1,1\rangle$ & $\langle -1,1\rangle$ & $\langle -1,1\rangle$ & $\langle 1,1\rangle$ & $\langle 1,1\rangle$ & $\langle 1,1\rangle$ & $\langle 0,-1\rangle$ & $\langle 0,-1\rangle$ & $\langle 0,-1\rangle$\\
$\chi_{20}$ & $\langle -2,-1\rangle$ & $\langle -2,-1\rangle$ & $\langle -2,-1\rangle$ & $\langle -2,-1\rangle$ & $\langle 0,-1\rangle$ & $\langle 0,-1\rangle$ & $\langle 0,-1\rangle$ & $\langle 1,1\rangle$ & $\langle 1,1\rangle$ & $\langle 1,1\rangle$\\
$\chi_{21}$ & $\langle -3,-3\rangle$ & $\langle -3,-3\rangle$ & $\langle 3,3\rangle$ & $\langle 3,3\rangle$ & $\langle -1,-1\rangle$ & $\langle 1,1\rangle$ & $\langle 1,1\rangle$ & $\langle 0,-1\rangle$ & $\langle 0,-1\rangle$ & $\langle 0,1\rangle$\\
$\chi_{22}$ & $\langle 0,3\rangle$ & $\langle 0,3\rangle$ & $\langle 0,-3\rangle$ & $\langle 0,-3\rangle$ & $\langle 0,1\rangle$ & $\langle 0,-1\rangle$ & $\langle 0,-1\rangle$ & $\langle 1,1\rangle$ & $\langle 1,1\rangle$ & $\langle -1,-1\rangle$\\
$\chi_{23}$ & $\langle 0,-2\rangle$ & $\langle 0,-2\rangle$ & $\langle 0,-2\rangle$ & $\langle 0,-2\rangle$ & $0$ & $0$ & $0$ & $0$ & $0$ & $0$\\
$\chi_{24}$ & $\langle 2,2\rangle$ & $\langle 2,2\rangle$ & $\langle 2,2\rangle$ & $\langle 2,2\rangle$ & $0$ & $0$ & $0$ & $0$ & $0$ & $0$\\
$\chi_{25}$ & $\langle -3,-3\rangle$ & $\langle -3,-3\rangle$ & $\langle -3,-3\rangle$ & $\langle -3,-3\rangle$ & $\langle -1,-1\rangle$ & $\langle -1,-1\rangle$ & $\langle -1,-1\rangle$ & $\langle 0,1\rangle$ & $\langle 0,1\rangle$ & $\langle 0,1\rangle$\\
$\chi_{26}$ & $\langle 0,3\rangle$ & $\langle 0,3\rangle$ & $\langle 0,3\rangle$ & $\langle 0,3\rangle$ & $\langle 0,1\rangle$ & $\langle 0,1\rangle$ & $\langle 0,1\rangle$ & $\langle -1,-1\rangle$ & $\langle -1,-1\rangle$ & $\langle -1,-1\rangle$\\
$\chi_{27}$ & $3$ & $3$ & $-3$ & $-3$ & $-1$ & $1$ & $1$ & $1$ & $1$ & $-1$\\
$\chi_{28}$ & $2$ & $2$ & $2$ & $2$ & $0$ & $0$ & $0$ & $0$ & $0$ & $0$\\
$\chi_{29}$ & $\langle -1,1\rangle$ & $\langle -1,1\rangle$ & $\langle 1,-1\rangle$ & $\langle 1,-1\rangle$ & $\langle 1,1\rangle$ & $\langle -1,-1\rangle$ & $\langle -1,-1\rangle$ & $\langle 0,1\rangle$ & $\langle 0,1\rangle$ & $\langle 0,-1\rangle$\\
$\chi_{30}$ & $\langle -2,-1\rangle$ & $\langle -2,-1\rangle$ & $\langle 2,1\rangle$ & $\langle 2,1\rangle$ & $\langle 0,-1\rangle$ & $\langle 0,1\rangle$ & $\langle 0,1\rangle$ & $\langle -1,-1\rangle$ & $\langle -1,-1\rangle$ & $\langle 1,1\rangle$\\
$\chi_{31}$ & $0$ & $0$ & $0$ & $0$ & $0$ & $0$ & $0$ & $0$ & $0$ & $0$\\
$\chi_{32}$ & $0$ & $0$ & $0$ & $0$ & $0$ & $0$ & $0$ & $0$ & $0$ & $0$\\
$\chi_{33}$ & $-1$ & $8$ & $-1$ & $2$ & $0$ & $-1$ & $2$ & $-1$ & $2$ & $0$\\
$\chi_{34}$ & $0$ & $0$ & $0$ & $0$ & $0$ & $0$ & $0$ & $0$ & $0$ & $0$\\
$\chi_{35}$ & $\langle 1,1\rangle$ & $\langle -8,-8\rangle$ & $-1$ & $2$ & $0$ & $-1$ & $2$ & $-1$ & $2$ & $0$\\
$\chi_{36}$ & $\langle 0,-1\rangle$ & $\langle 0,8\rangle$ & $-1$ & $2$ & $0$ & $-1$ & $2$ & $-1$ & $2$ & $0$\\
$\chi_{37}$ & $0$ & $0$ & $0$ & $0$ & $0$ & $0$ & $0$ & $0$ & $0$ & $0$\\
$\chi_{38}$ & $-1$ & $8$ & $2$ & $-4$ & $0$ & $0$ & $0$ & $0$ & $0$ & $0$\\
$\chi_{39}$ & $\langle 0,-1\rangle$ & $\langle 0,8\rangle$ & $2$ & $-4$ & $0$ & $0$ & $0$ & $0$ & $0$ & $0$\\
$\chi_{40}$ & $\langle 1,1\rangle$ & $\langle -8,-8\rangle$ & $2$ & $-4$ & $0$ & $0$ & $0$ & $0$ & $0$ & $0$\\
$\chi_{41}$ & $0$ & $0$ & $-3$ & $6$ & $0$ & $1$ & $-2$ & $1$ & $-2$ & $0$\\
$\chi_{42}$ & $\langle -1,-1\rangle$ & $\langle 8,8\rangle$ & $\langle -1,-1\rangle$ & $\langle 2,2\rangle$ & $0$ & $\langle 1,1\rangle$ & $\langle -2,-2\rangle$ & $\langle 0,-1\rangle$ & $\langle 0,2\rangle$ & $0$\\
$\chi_{43}$ & $\langle 0,1\rangle$ & $\langle 0,-8\rangle$ & $\langle 0,1\rangle$ & $\langle 0,-2\rangle$ & $0$ & $\langle 0,-1\rangle$ & $\langle 0,2\rangle$ & $\langle 1,1\rangle$ & $\langle -2,-2\rangle$ & $0$\\
$\chi_{44}$ & $\langle 0,1\rangle$ & $\langle 0,-8\rangle$ & $\langle -1,-1\rangle$ & $\langle 2,2\rangle$ & $0$ & $\langle 1,1\rangle$ & $\langle -2,-2\rangle$ & $\langle 0,-1\rangle$ & $\langle 0,2\rangle$ & $0$\\
$\chi_{45}$ & $\langle -1,-1\rangle$ & $\langle 8,8\rangle$ & $\langle 0,1\rangle$ & $\langle 0,-2\rangle$ & $0$ & $\langle 0,-1\rangle$ & $\langle 0,2\rangle$ & $\langle 1,1\rangle$ & $\langle -2,-2\rangle$ & $0$\\
$\chi_{46}$ & $1$ & $-8$ & $\langle -1,-1\rangle$ & $\langle 2,2\rangle$ & $0$ & $\langle 1,1\rangle$ & $\langle -2,-2\rangle$ & $\langle 0,-1\rangle$ & $\langle 0,2\rangle$ & $0$\\
$\chi_{47}$ & $1$ & $-8$ & $\langle 0,1\rangle$ & $\langle 0,-2\rangle$ & $0$ & $\langle 0,-1\rangle$ & $\langle 0,2\rangle$ & $\langle 1,1\rangle$ & $\langle -2,-2\rangle$ & $0$\\
$\chi_{48}$ & $\langle -1,-1\rangle$ & $\langle 8,8\rangle$ & $\langle 2,2\rangle$ & $\langle -4,-4\rangle$ & $0$ & $0$ & $0$ & $0$ & $0$ & $0$\\
$\chi_{49}$ & $\langle 0,1\rangle$ & $\langle 0,-8\rangle$ & $\langle 0,-2\rangle$ & $\langle 0,4\rangle$ & $0$ & $0$ & $0$ & $0$ & $0$ & $0$\\
$\chi_{50}$ & $\langle -1,-1\rangle$ & $\langle 8,8\rangle$ & $\langle 0,-2\rangle$ & $\langle 0,4\rangle$ & $0$ & $0$ & $0$ & $0$ & $0$ & $0$\\
$\chi_{51}$ & $\langle 0,1\rangle$ & $\langle 0,-8\rangle$ & $\langle 2,2\rangle$ & $\langle -4,-4\rangle$ & $0$ & $0$ & $0$ & $0$ & $0$ & $0$\\
$\chi_{52}$ & $1$ & $-8$ & $\langle 0,-2\rangle$ & $\langle 0,4\rangle$ & $0$ & $0$ & $0$ & $0$ & $0$ & $0$\\
$\chi_{53}$ & $1$ & $-8$ & $\langle 2,2\rangle$ & $\langle -4,-4\rangle$ & $0$ & $0$ & $0$ & $0$ & $0$ & $0$\\
$\chi_{54}$ & $0$ & $0$ & $0$ & $0$ & $0$ & $0$ & $0$ & $0$ & $0$ & $0$\\
$\chi_{55}$ & $0$ & $0$ & $0$ & $0$ & $0$ & $0$ & $0$ & $0$ & $0$ & $0$\\
$\chi_{56}$ & $0$ & $0$ & $0$ & $0$ & $0$ & $0$ & $0$ & $0$ & $0$ & $0$\\
$\chi_{57}$ & $0$ & $0$ & $\langle -3,-3\rangle$ & $\langle 6,6\rangle$ & $0$ & $\langle -1,-1\rangle$ & $\langle 2,2\rangle$ & $\langle 0,1\rangle$ & $\langle 0,-2\rangle$ & $0$\\
$\chi_{58}$ & $0$ & $0$ & $\langle 0,3\rangle$ & $\langle 0,-6\rangle$ & $0$ & $\langle 0,1\rangle$ & $\langle 0,-2\rangle$ & $\langle -1,-1\rangle$ & $\langle 2,2\rangle$ & $0$\\
\end{longtable}
\endgroup
\clearpage
\charactertableheading{Character values: columns \texorpdfstring{$K_{31},\ldots,K_{40}$}{K31--K40}}
\begingroup
\fontsize{8}{9.4}\selectfont\setlength{\tabcolsep}{2.1pt}
\renewcommand{\arraystretch}{1.02}
\begin{longtable}{r|rrrrrrrrrr}
 & $K_{31}$ & $K_{32}$ & $K_{33}$ & $K_{34}$ & $K_{35}$ & $K_{36}$ & $K_{37}$ & $K_{38}$ & $K_{39}$ & $K_{40}$\\\hline\endhead
$\chi_{1}$ & $1$ & $1$ & $1$ & $1$ & $1$ & $1$ & $1$ & $1$ & $1$ & $1$\\
$\chi_{2}$ & $1$ & $1$ & $1$ & $1$ & $-1$ & $\langle 0,-1\rangle$ & $\langle 0,1\rangle$ & $\langle 0,1\rangle$ & $\langle -1,-1\rangle$ & $\langle -1,-1\rangle$\\
$\chi_{3}$ & $1$ & $1$ & $1$ & $1$ & $-1$ & $\langle 1,1\rangle$ & $\langle -1,-1\rangle$ & $\langle -1,-1\rangle$ & $\langle 0,1\rangle$ & $\langle 0,1\rangle$\\
$\chi_{4}$ & $2$ & $2$ & $2$ & $2$ & $2$ & $\langle 1,1\rangle$ & $\langle 1,1\rangle$ & $\langle 1,1\rangle$ & $\langle 0,-1\rangle$ & $\langle 0,-1\rangle$\\
$\chi_{5}$ & $2$ & $2$ & $2$ & $2$ & $2$ & $\langle 0,-1\rangle$ & $\langle 0,-1\rangle$ & $\langle 0,-1\rangle$ & $\langle 1,1\rangle$ & $\langle 1,1\rangle$\\
$\chi_{6}$ & $0$ & $0$ & $0$ & $0$ & $0$ & $0$ & $0$ & $0$ & $0$ & $0$\\
$\chi_{7}$ & $1$ & $1$ & $1$ & $1$ & $1$ & $\langle -1,-1\rangle$ & $\langle -1,-1\rangle$ & $\langle -1,-1\rangle$ & $\langle 0,1\rangle$ & $\langle 0,1\rangle$\\
$\chi_{8}$ & $1$ & $1$ & $1$ & $1$ & $1$ & $\langle 0,1\rangle$ & $\langle 0,1\rangle$ & $\langle 0,1\rangle$ & $\langle -1,-1\rangle$ & $\langle -1,-1\rangle$\\
$\chi_{9}$ & $3$ & $3$ & $3$ & $3$ & $3$ & $0$ & $0$ & $0$ & $0$ & $0$\\
$\chi_{10}$ & $0$ & $0$ & $0$ & $0$ & $0$ & $0$ & $0$ & $0$ & $0$ & $0$\\
$\chi_{11}$ & $2$ & $2$ & $2$ & $2$ & $-2$ & $1$ & $-1$ & $-1$ & $-1$ & $-1$\\
$\chi_{12}$ & $-1$ & $-1$ & $-1$ & $-1$ & $-1$ & $-1$ & $-1$ & $-1$ & $-1$ & $-1$\\
$\chi_{13}$ & $2$ & $2$ & $2$ & $2$ & $-2$ & $\langle 0,1\rangle$ & $\langle 0,-1\rangle$ & $\langle 0,-1\rangle$ & $\langle 1,1\rangle$ & $\langle 1,1\rangle$\\
$\chi_{14}$ & $2$ & $2$ & $2$ & $2$ & $-2$ & $\langle -1,-1\rangle$ & $\langle 1,1\rangle$ & $\langle 1,1\rangle$ & $\langle 0,-1\rangle$ & $\langle 0,-1\rangle$\\
$\chi_{15}$ & $-1$ & $-1$ & $-1$ & $-1$ & $1$ & $\langle -1,-1\rangle$ & $\langle 1,1\rangle$ & $\langle 1,1\rangle$ & $\langle 0,-1\rangle$ & $\langle 0,-1\rangle$\\
$\chi_{16}$ & $-1$ & $-1$ & $-1$ & $-1$ & $1$ & $\langle 0,1\rangle$ & $\langle 0,-1\rangle$ & $\langle 0,-1\rangle$ & $\langle 1,1\rangle$ & $\langle 1,1\rangle$\\
$\chi_{17}$ & $3$ & $3$ & $3$ & $3$ & $3$ & $0$ & $0$ & $0$ & $0$ & $0$\\
$\chi_{18}$ & $3$ & $3$ & $3$ & $3$ & $3$ & $0$ & $0$ & $0$ & $0$ & $0$\\
$\chi_{19}$ & $0$ & $0$ & $0$ & $0$ & $0$ & $0$ & $0$ & $0$ & $0$ & $0$\\
$\chi_{20}$ & $0$ & $0$ & $0$ & $0$ & $0$ & $0$ & $0$ & $0$ & $0$ & $0$\\
$\chi_{21}$ & $0$ & $0$ & $0$ & $0$ & $0$ & $0$ & $0$ & $0$ & $0$ & $0$\\
$\chi_{22}$ & $0$ & $0$ & $0$ & $0$ & $0$ & $0$ & $0$ & $0$ & $0$ & $0$\\
$\chi_{23}$ & $1$ & $1$ & $1$ & $1$ & $1$ & $\langle -1,-1\rangle$ & $\langle -1,-1\rangle$ & $\langle -1,-1\rangle$ & $\langle 0,1\rangle$ & $\langle 0,1\rangle$\\
$\chi_{24}$ & $1$ & $1$ & $1$ & $1$ & $1$ & $\langle 0,1\rangle$ & $\langle 0,1\rangle$ & $\langle 0,1\rangle$ & $\langle -1,-1\rangle$ & $\langle -1,-1\rangle$\\
$\chi_{25}$ & $0$ & $0$ & $0$ & $0$ & $0$ & $0$ & $0$ & $0$ & $0$ & $0$\\
$\chi_{26}$ & $0$ & $0$ & $0$ & $0$ & $0$ & $0$ & $0$ & $0$ & $0$ & $0$\\
$\chi_{27}$ & $0$ & $0$ & $0$ & $0$ & $0$ & $0$ & $0$ & $0$ & $0$ & $0$\\
$\chi_{28}$ & $-3$ & $-3$ & $-3$ & $-3$ & $-3$ & $0$ & $0$ & $0$ & $0$ & $0$\\
$\chi_{29}$ & $3$ & $3$ & $3$ & $3$ & $-3$ & $0$ & $0$ & $0$ & $0$ & $0$\\
$\chi_{30}$ & $3$ & $3$ & $3$ & $3$ & $-3$ & $0$ & $0$ & $0$ & $0$ & $0$\\
$\chi_{31}$ & $-2$ & $-2$ & $-2$ & $-2$ & $2$ & $-1$ & $1$ & $1$ & $1$ & $1$\\
$\chi_{32}$ & $-2$ & $-2$ & $-2$ & $-2$ & $-2$ & $1$ & $1$ & $1$ & $1$ & $1$\\
$\chi_{33}$ & $-1$ & $-1$ & $-1$ & $8$ & $0$ & $0$ & $2$ & $-1$ & $-1$ & $2$\\
$\chi_{34}$ & $-4$ & $-4$ & $-4$ & $-4$ & $4$ & $1$ & $-1$ & $-1$ & $-1$ & $-1$\\
$\chi_{35}$ & $2$ & $\langle -4,-3\rangle$ & $\langle -1,3\rangle$ & $2$ & $0$ & $0$ & $\langle 0,2\rangle$ & $\langle 0,-1\rangle$ & $\langle 1,1\rangle$ & $\langle -2,-2\rangle$\\
$\chi_{36}$ & $2$ & $\langle -1,3\rangle$ & $\langle -4,-3\rangle$ & $2$ & $0$ & $0$ & $\langle -2,-2\rangle$ & $\langle 1,1\rangle$ & $\langle 0,-1\rangle$ & $\langle 0,2\rangle$\\
$\chi_{37}$ & $0$ & $0$ & $0$ & $0$ & $0$ & $0$ & $0$ & $0$ & $0$ & $0$\\
$\chi_{38}$ & $4$ & $-5$ & $-5$ & $4$ & $0$ & $0$ & $-2$ & $1$ & $1$ & $-2$\\
$\chi_{39}$ & $1$ & $\langle -5,-3\rangle$ & $\langle -2,3\rangle$ & $10$ & $0$ & $0$ & $\langle 2,2\rangle$ & $\langle -1,-1\rangle$ & $\langle 0,1\rangle$ & $\langle 0,-2\rangle$\\
$\chi_{40}$ & $1$ & $\langle -2,3\rangle$ & $\langle -5,-3\rangle$ & $10$ & $0$ & $0$ & $\langle 0,-2\rangle$ & $\langle 0,1\rangle$ & $\langle -1,-1\rangle$ & $\langle 2,2\rangle$\\
$\chi_{41}$ & $3$ & $-6$ & $-6$ & $12$ & $0$ & $0$ & $0$ & $0$ & $0$ & $0$\\
$\chi_{42}$ & $0$ & $\langle -3,-6\rangle$ & $\langle 3,6\rangle$ & $0$ & $0$ & $0$ & $0$ & $0$ & $0$ & $0$\\
$\chi_{43}$ & $0$ & $\langle 3,6\rangle$ & $\langle -3,-6\rangle$ & $0$ & $0$ & $0$ & $0$ & $0$ & $0$ & $0$\\
$\chi_{44}$ & $3$ & $\langle 0,3\rangle$ & $\langle -3,-3\rangle$ & $-6$ & $0$ & $0$ & $0$ & $0$ & $0$ & $0$\\
$\chi_{45}$ & $3$ & $\langle -3,-3\rangle$ & $\langle 0,3\rangle$ & $-6$ & $0$ & $0$ & $0$ & $0$ & $0$ & $0$\\
$\chi_{46}$ & $-3$ & $\langle 3,3\rangle$ & $\langle 0,-3\rangle$ & $6$ & $0$ & $0$ & $0$ & $0$ & $0$ & $0$\\
$\chi_{47}$ & $-3$ & $\langle 0,-3\rangle$ & $\langle 3,3\rangle$ & $6$ & $0$ & $0$ & $0$ & $0$ & $0$ & $0$\\
$\chi_{48}$ & $0$ & $\langle 3,6\rangle$ & $\langle -3,-6\rangle$ & $0$ & $0$ & $0$ & $0$ & $0$ & $0$ & $0$\\
$\chi_{49}$ & $0$ & $\langle -3,-6\rangle$ & $\langle 3,6\rangle$ & $0$ & $0$ & $0$ & $0$ & $0$ & $0$ & $0$\\
$\chi_{50}$ & $-3$ & $\langle 3,3\rangle$ & $\langle 0,-3\rangle$ & $6$ & $0$ & $0$ & $0$ & $0$ & $0$ & $0$\\
$\chi_{51}$ & $-3$ & $\langle 0,-3\rangle$ & $\langle 3,3\rangle$ & $6$ & $0$ & $0$ & $0$ & $0$ & $0$ & $0$\\
$\chi_{52}$ & $3$ & $\langle 0,3\rangle$ & $\langle -3,-3\rangle$ & $-6$ & $0$ & $0$ & $0$ & $0$ & $0$ & $0$\\
$\chi_{53}$ & $3$ & $\langle -3,-3\rangle$ & $\langle 0,3\rangle$ & $-6$ & $0$ & $0$ & $0$ & $0$ & $0$ & $0$\\
$\chi_{54}$ & $-5$ & $4$ & $4$ & $4$ & $0$ & $0$ & $-2$ & $1$ & $1$ & $-2$\\
$\chi_{55}$ & $1$ & $\langle 4,6\rangle$ & $\langle -2,-6\rangle$ & $-8$ & $0$ & $0$ & $\langle 2,2\rangle$ & $\langle -1,-1\rangle$ & $\langle 0,1\rangle$ & $\langle 0,-2\rangle$\\
$\chi_{56}$ & $1$ & $\langle -2,-6\rangle$ & $\langle 4,6\rangle$ & $-8$ & $0$ & $0$ & $\langle 0,-2\rangle$ & $\langle 0,1\rangle$ & $\langle -1,-1\rangle$ & $\langle 2,2\rangle$\\
$\chi_{57}$ & $0$ & $0$ & $0$ & $0$ & $0$ & $0$ & $0$ & $0$ & $0$ & $0$\\
$\chi_{58}$ & $0$ & $0$ & $0$ & $0$ & $0$ & $0$ & $0$ & $0$ & $0$ & $0$\\
\end{longtable}
\endgroup
\clearpage
\charactertableheading{Character values: columns \texorpdfstring{$K_{41},\ldots,K_{50}$}{K41--K50}}
\begingroup
\fontsize{8}{9.4}\selectfont\setlength{\tabcolsep}{2.1pt}
\renewcommand{\arraystretch}{1.02}
\begin{longtable}{r|rrrrrrrrrr}
 & $K_{41}$ & $K_{42}$ & $K_{43}$ & $K_{44}$ & $K_{45}$ & $K_{46}$ & $K_{47}$ & $K_{48}$ & $K_{49}$ & $K_{50}$\\\hline\endhead
$\chi_{1}$ & $1$ & $1$ & $1$ & $1$ & $1$ & $1$ & $1$ & $1$ & $1$ & $1$\\
$\chi_{2}$ & $\langle 1,1\rangle$ & $0$ & $-2$ & $-2$ & $-2$ & $2$ & $0$ & $\langle -1,-2\rangle$ & $\langle -1,-2\rangle$ & $\langle 1,2\rangle$\\
$\chi_{3}$ & $\langle 0,-1\rangle$ & $0$ & $-2$ & $-2$ & $-2$ & $2$ & $0$ & $\langle 1,2\rangle$ & $\langle 1,2\rangle$ & $\langle -1,-2\rangle$\\
$\chi_{4}$ & $\langle 0,-1\rangle$ & $1$ & $-1$ & $-1$ & $-1$ & $-1$ & $1$ & $0$ & $0$ & $0$\\
$\chi_{5}$ & $\langle 1,1\rangle$ & $1$ & $-1$ & $-1$ & $-1$ & $-1$ & $1$ & $0$ & $0$ & $0$\\
$\chi_{6}$ & $0$ & $2$ & $3$ & $3$ & $3$ & $3$ & $-1$ & $-2$ & $-2$ & $-2$\\
$\chi_{7}$ & $\langle 0,1\rangle$ & $-2$ & $1$ & $1$ & $1$ & $1$ & $1$ & $-1$ & $-1$ & $-1$\\
$\chi_{8}$ & $\langle -1,-1\rangle$ & $-2$ & $1$ & $1$ & $1$ & $1$ & $1$ & $-1$ & $-1$ & $-1$\\
$\chi_{9}$ & $0$ & $3$ & $0$ & $0$ & $0$ & $0$ & $0$ & $1$ & $1$ & $1$\\
$\chi_{10}$ & $0$ & $-1$ & $3$ & $3$ & $3$ & $3$ & $-1$ & $2$ & $2$ & $2$\\
$\chi_{11}$ & $1$ & $0$ & $2$ & $2$ & $2$ & $-2$ & $0$ & $0$ & $0$ & $0$\\
$\chi_{12}$ & $-1$ & $4$ & $5$ & $5$ & $5$ & $5$ & $1$ & $1$ & $1$ & $1$\\
$\chi_{13}$ & $\langle -1,-1\rangle$ & $0$ & $2$ & $2$ & $2$ & $-2$ & $0$ & $0$ & $0$ & $0$\\
$\chi_{14}$ & $\langle 0,1\rangle$ & $0$ & $2$ & $2$ & $2$ & $-2$ & $0$ & $0$ & $0$ & $0$\\
$\chi_{15}$ & $\langle 0,1\rangle$ & $0$ & $-4$ & $-4$ & $-4$ & $4$ & $0$ & $\langle 1,2\rangle$ & $\langle 1,2\rangle$ & $\langle -1,-2\rangle$\\
$\chi_{16}$ & $\langle -1,-1\rangle$ & $0$ & $-4$ & $-4$ & $-4$ & $4$ & $0$ & $\langle -1,-2\rangle$ & $\langle -1,-2\rangle$ & $\langle 1,2\rangle$\\
$\chi_{17}$ & $0$ & $0$ & $0$ & $0$ & $0$ & $0$ & $0$ & $-1$ & $-1$ & $-1$\\
$\chi_{18}$ & $0$ & $2$ & $3$ & $3$ & $3$ & $3$ & $-1$ & $-1$ & $-1$ & $-1$\\
$\chi_{19}$ & $0$ & $2$ & $-3$ & $-3$ & $-3$ & $-3$ & $-1$ & $0$ & $0$ & $0$\\
$\chi_{20}$ & $0$ & $2$ & $-3$ & $-3$ & $-3$ & $-3$ & $-1$ & $0$ & $0$ & $0$\\
$\chi_{21}$ & $0$ & $0$ & $0$ & $0$ & $0$ & $0$ & $0$ & $0$ & $0$ & $0$\\
$\chi_{22}$ & $0$ & $0$ & $0$ & $0$ & $0$ & $0$ & $0$ & $0$ & $0$ & $0$\\
$\chi_{23}$ & $\langle 0,1\rangle$ & $0$ & $-2$ & $-2$ & $-2$ & $-2$ & $0$ & $1$ & $1$ & $1$\\
$\chi_{24}$ & $\langle -1,-1\rangle$ & $0$ & $-2$ & $-2$ & $-2$ & $-2$ & $0$ & $1$ & $1$ & $1$\\
$\chi_{25}$ & $0$ & $-3$ & $0$ & $0$ & $0$ & $0$ & $0$ & $0$ & $0$ & $0$\\
$\chi_{26}$ & $0$ & $-3$ & $0$ & $0$ & $0$ & $0$ & $0$ & $0$ & $0$ & $0$\\
$\chi_{27}$ & $0$ & $0$ & $-6$ & $-6$ & $-6$ & $6$ & $0$ & $0$ & $0$ & $0$\\
$\chi_{28}$ & $0$ & $4$ & $-3$ & $-3$ & $-3$ & $-3$ & $1$ & $-1$ & $-1$ & $-1$\\
$\chi_{29}$ & $0$ & $0$ & $0$ & $0$ & $0$ & $0$ & $0$ & $\langle 1,2\rangle$ & $\langle 1,2\rangle$ & $\langle -1,-2\rangle$\\
$\chi_{30}$ & $0$ & $0$ & $0$ & $0$ & $0$ & $0$ & $0$ & $\langle -1,-2\rangle$ & $\langle -1,-2\rangle$ & $\langle 1,2\rangle$\\
$\chi_{31}$ & $-1$ & $0$ & $4$ & $4$ & $4$ & $-4$ & $0$ & $0$ & $0$ & $0$\\
$\chi_{32}$ & $1$ & $0$ & $4$ & $4$ & $4$ & $4$ & $0$ & $0$ & $0$ & $0$\\
$\chi_{33}$ & $0$ & $0$ & $8$ & $-1$ & $-1$ & $0$ & $0$ & $-1$ & $2$ & $-1$\\
$\chi_{34}$ & $1$ & $0$ & $2$ & $2$ & $2$ & $-2$ & $0$ & $0$ & $0$ & $0$\\
$\chi_{35}$ & $0$ & $0$ & $-4$ & $\langle -1,-3\rangle$ & $\langle 2,3\rangle$ & $0$ & $0$ & $\langle 1,1\rangle$ & $\langle -2,-2\rangle$ & $\langle 0,-1\rangle$\\
$\chi_{36}$ & $0$ & $0$ & $-4$ & $\langle 2,3\rangle$ & $\langle -1,-3\rangle$ & $0$ & $0$ & $\langle 0,-1\rangle$ & $\langle 0,2\rangle$ & $\langle 1,1\rangle$\\
$\chi_{37}$ & $0$ & $-3$ & $0$ & $0$ & $0$ & $0$ & $0$ & $0$ & $0$ & $0$\\
$\chi_{38}$ & $0$ & $0$ & $-8$ & $1$ & $1$ & $0$ & $0$ & $-1$ & $2$ & $-1$\\
$\chi_{39}$ & $0$ & $0$ & $4$ & $\langle -2,-3\rangle$ & $\langle 1,3\rangle$ & $0$ & $0$ & $\langle 0,-1\rangle$ & $\langle 0,2\rangle$ & $\langle 1,1\rangle$\\
$\chi_{40}$ & $0$ & $0$ & $4$ & $\langle 1,3\rangle$ & $\langle -2,-3\rangle$ & $0$ & $0$ & $\langle 1,1\rangle$ & $\langle -2,-2\rangle$ & $\langle 0,-1\rangle$\\
$\chi_{41}$ & $0$ & $0$ & $0$ & $0$ & $0$ & $0$ & $0$ & $0$ & $0$ & $0$\\
$\chi_{42}$ & $0$ & $0$ & $12$ & $\langle 0,3\rangle$ & $\langle -3,-3\rangle$ & $0$ & $0$ & $1$ & $-2$ & $1$\\
$\chi_{43}$ & $0$ & $0$ & $12$ & $\langle -3,-3\rangle$ & $\langle 0,3\rangle$ & $0$ & $0$ & $1$ & $-2$ & $1$\\
$\chi_{44}$ & $0$ & $0$ & $0$ & $\langle -3,-6\rangle$ & $\langle 3,6\rangle$ & $0$ & $0$ & $\langle -1,-1\rangle$ & $\langle 2,2\rangle$ & $\langle 0,1\rangle$\\
$\chi_{45}$ & $0$ & $0$ & $0$ & $\langle 3,6\rangle$ & $\langle -3,-6\rangle$ & $0$ & $0$ & $\langle 0,1\rangle$ & $\langle 0,-2\rangle$ & $\langle -1,-1\rangle$\\
$\chi_{46}$ & $0$ & $0$ & $-12$ & $\langle 3,3\rangle$ & $\langle 0,-3\rangle$ & $0$ & $0$ & $\langle 0,1\rangle$ & $\langle 0,-2\rangle$ & $\langle -1,-1\rangle$\\
$\chi_{47}$ & $0$ & $0$ & $-12$ & $\langle 0,-3\rangle$ & $\langle 3,3\rangle$ & $0$ & $0$ & $\langle -1,-1\rangle$ & $\langle 2,2\rangle$ & $\langle 0,1\rangle$\\
$\chi_{48}$ & $0$ & $0$ & $-12$ & $\langle 0,-3\rangle$ & $\langle 3,3\rangle$ & $0$ & $0$ & $1$ & $-2$ & $1$\\
$\chi_{49}$ & $0$ & $0$ & $-12$ & $\langle 3,3\rangle$ & $\langle 0,-3\rangle$ & $0$ & $0$ & $1$ & $-2$ & $1$\\
$\chi_{50}$ & $0$ & $0$ & $0$ & $\langle -3,-6\rangle$ & $\langle 3,6\rangle$ & $0$ & $0$ & $\langle 0,1\rangle$ & $\langle 0,-2\rangle$ & $\langle -1,-1\rangle$\\
$\chi_{51}$ & $0$ & $0$ & $0$ & $\langle 3,6\rangle$ & $\langle -3,-6\rangle$ & $0$ & $0$ & $\langle -1,-1\rangle$ & $\langle 2,2\rangle$ & $\langle 0,1\rangle$\\
$\chi_{52}$ & $0$ & $0$ & $12$ & $\langle 0,3\rangle$ & $\langle -3,-3\rangle$ & $0$ & $0$ & $\langle -1,-1\rangle$ & $\langle 2,2\rangle$ & $\langle 0,1\rangle$\\
$\chi_{53}$ & $0$ & $0$ & $12$ & $\langle -3,-3\rangle$ & $\langle 0,3\rangle$ & $0$ & $0$ & $\langle 0,1\rangle$ & $\langle 0,-2\rangle$ & $\langle -1,-1\rangle$\\
$\chi_{54}$ & $0$ & $0$ & $16$ & $-2$ & $-2$ & $0$ & $0$ & $0$ & $0$ & $0$\\
$\chi_{55}$ & $0$ & $0$ & $-8$ & $\langle 4,6\rangle$ & $\langle -2,-6\rangle$ & $0$ & $0$ & $0$ & $0$ & $0$\\
$\chi_{56}$ & $0$ & $0$ & $-8$ & $\langle -2,-6\rangle$ & $\langle 4,6\rangle$ & $0$ & $0$ & $0$ & $0$ & $0$\\
$\chi_{57}$ & $0$ & $0$ & $0$ & $0$ & $0$ & $0$ & $0$ & $0$ & $0$ & $0$\\
$\chi_{58}$ & $0$ & $0$ & $0$ & $0$ & $0$ & $0$ & $0$ & $0$ & $0$ & $0$\\
\end{longtable}
\endgroup
\clearpage
\charactertableheading{Character values: columns \texorpdfstring{$K_{51},\ldots,K_{58}$}{K51--K58}}
\begingroup
\fontsize{8}{9.4}\selectfont\setlength{\tabcolsep}{2.1pt}
\renewcommand{\arraystretch}{1.02}
\begin{longtable}{r|rrrrrrrr}
 & $K_{51}$ & $K_{52}$ & $K_{53}$ & $K_{54}$ & $K_{55}$ & $K_{56}$ & $K_{57}$ & $K_{58}$\\\hline\endhead
$\chi_{1}$ & $1$ & $1$ & $1$ & $1$ & $1$ & $1$ & $1$ & $1$\\
$\chi_{2}$ & $\langle 1,2\rangle$ & $0$ & $0$ & $0$ & $0$ & $0$ & $1$ & $-1$\\
$\chi_{3}$ & $\langle -1,-2\rangle$ & $0$ & $0$ & $0$ & $0$ & $0$ & $1$ & $-1$\\
$\chi_{4}$ & $0$ & $\langle 1,2\rangle$ & $\langle 1,2\rangle$ & $\langle -1,-2\rangle$ & $\langle -1,-2\rangle$ & $-1$ & $0$ & $0$\\
$\chi_{5}$ & $0$ & $\langle -1,-2\rangle$ & $\langle -1,-2\rangle$ & $\langle 1,2\rangle$ & $\langle 1,2\rangle$ & $-1$ & $0$ & $0$\\
$\chi_{6}$ & $-2$ & $1$ & $1$ & $1$ & $1$ & $0$ & $1$ & $1$\\
$\chi_{7}$ & $-1$ & $-1$ & $-1$ & $-1$ & $-1$ & $0$ & $0$ & $0$\\
$\chi_{8}$ & $-1$ & $-1$ & $-1$ & $-1$ & $-1$ & $0$ & $0$ & $0$\\
$\chi_{9}$ & $1$ & $-2$ & $-2$ & $-2$ & $-2$ & $1$ & $0$ & $0$\\
$\chi_{10}$ & $2$ & $-1$ & $-1$ & $-1$ & $-1$ & $-1$ & $0$ & $0$\\
$\chi_{11}$ & $0$ & $0$ & $0$ & $0$ & $0$ & $0$ & $0$ & $0$\\
$\chi_{12}$ & $1$ & $1$ & $1$ & $1$ & $1$ & $0$ & $0$ & $0$\\
$\chi_{13}$ & $0$ & $0$ & $0$ & $0$ & $0$ & $0$ & $0$ & $0$\\
$\chi_{14}$ & $0$ & $0$ & $0$ & $0$ & $0$ & $0$ & $0$ & $0$\\
$\chi_{15}$ & $\langle -1,-2\rangle$ & $0$ & $0$ & $0$ & $0$ & $0$ & $0$ & $0$\\
$\chi_{16}$ & $\langle 1,2\rangle$ & $0$ & $0$ & $0$ & $0$ & $0$ & $0$ & $0$\\
$\chi_{17}$ & $-1$ & $2$ & $2$ & $2$ & $2$ & $0$ & $-1$ & $-1$\\
$\chi_{18}$ & $-1$ & $-1$ & $-1$ & $-1$ & $-1$ & $0$ & $0$ & $0$\\
$\chi_{19}$ & $0$ & $\langle 1,2\rangle$ & $\langle 1,2\rangle$ & $\langle -1,-2\rangle$ & $\langle -1,-2\rangle$ & $0$ & $0$ & $0$\\
$\chi_{20}$ & $0$ & $\langle -1,-2\rangle$ & $\langle -1,-2\rangle$ & $\langle 1,2\rangle$ & $\langle 1,2\rangle$ & $0$ & $0$ & $0$\\
$\chi_{21}$ & $0$ & $0$ & $0$ & $0$ & $0$ & $0$ & $-1$ & $1$\\
$\chi_{22}$ & $0$ & $0$ & $0$ & $0$ & $0$ & $0$ & $-1$ & $1$\\
$\chi_{23}$ & $1$ & $\langle 2,2\rangle$ & $\langle 2,2\rangle$ & $\langle 0,-2\rangle$ & $\langle 0,-2\rangle$ & $0$ & $0$ & $0$\\
$\chi_{24}$ & $1$ & $\langle 0,-2\rangle$ & $\langle 0,-2\rangle$ & $\langle 2,2\rangle$ & $\langle 2,2\rangle$ & $0$ & $0$ & $0$\\
$\chi_{25}$ & $0$ & $0$ & $0$ & $0$ & $0$ & $1$ & $0$ & $0$\\
$\chi_{26}$ & $0$ & $0$ & $0$ & $0$ & $0$ & $1$ & $0$ & $0$\\
$\chi_{27}$ & $0$ & $0$ & $0$ & $0$ & $0$ & $0$ & $0$ & $0$\\
$\chi_{28}$ & $-1$ & $-1$ & $-1$ & $-1$ & $-1$ & $0$ & $0$ & $0$\\
$\chi_{29}$ & $\langle -1,-2\rangle$ & $0$ & $0$ & $0$ & $0$ & $0$ & $0$ & $0$\\
$\chi_{30}$ & $\langle 1,2\rangle$ & $0$ & $0$ & $0$ & $0$ & $0$ & $0$ & $0$\\
$\chi_{31}$ & $0$ & $0$ & $0$ & $0$ & $0$ & $0$ & $1$ & $-1$\\
$\chi_{32}$ & $0$ & $0$ & $0$ & $0$ & $0$ & $0$ & $-1$ & $-1$\\
$\chi_{33}$ & $2$ & $-1$ & $2$ & $-1$ & $2$ & $0$ & $0$ & $0$\\
$\chi_{34}$ & $0$ & $0$ & $0$ & $0$ & $0$ & $0$ & $0$ & $0$\\
$\chi_{35}$ & $\langle 0,2\rangle$ & $\langle 1,1\rangle$ & $\langle -2,-2\rangle$ & $\langle 0,-1\rangle$ & $\langle 0,2\rangle$ & $0$ & $0$ & $0$\\
$\chi_{36}$ & $\langle -2,-2\rangle$ & $\langle 0,-1\rangle$ & $\langle 0,2\rangle$ & $\langle 1,1\rangle$ & $\langle -2,-2\rangle$ & $0$ & $0$ & $0$\\
$\chi_{37}$ & $0$ & $0$ & $0$ & $0$ & $0$ & $-1$ & $1$ & $1$\\
$\chi_{38}$ & $2$ & $-1$ & $2$ & $-1$ & $2$ & $0$ & $0$ & $0$\\
$\chi_{39}$ & $\langle -2,-2\rangle$ & $\langle 0,-1\rangle$ & $\langle 0,2\rangle$ & $\langle 1,1\rangle$ & $\langle -2,-2\rangle$ & $0$ & $0$ & $0$\\
$\chi_{40}$ & $\langle 0,2\rangle$ & $\langle 1,1\rangle$ & $\langle -2,-2\rangle$ & $\langle 0,-1\rangle$ & $\langle 0,2\rangle$ & $0$ & $0$ & $0$\\
$\chi_{41}$ & $0$ & $0$ & $0$ & $0$ & $0$ & $0$ & $0$ & $0$\\
$\chi_{42}$ & $-2$ & $\langle 0,1\rangle$ & $\langle 0,-2\rangle$ & $\langle -1,-1\rangle$ & $\langle 2,2\rangle$ & $0$ & $0$ & $0$\\
$\chi_{43}$ & $-2$ & $\langle -1,-1\rangle$ & $\langle 2,2\rangle$ & $\langle 0,1\rangle$ & $\langle 0,-2\rangle$ & $0$ & $0$ & $0$\\
$\chi_{44}$ & $\langle 0,-2\rangle$ & $1$ & $-2$ & $1$ & $-2$ & $0$ & $0$ & $0$\\
$\chi_{45}$ & $\langle 2,2\rangle$ & $1$ & $-2$ & $1$ & $-2$ & $0$ & $0$ & $0$\\
$\chi_{46}$ & $\langle 2,2\rangle$ & $\langle -1,-1\rangle$ & $\langle 2,2\rangle$ & $\langle 0,1\rangle$ & $\langle 0,-2\rangle$ & $0$ & $0$ & $0$\\
$\chi_{47}$ & $\langle 0,-2\rangle$ & $\langle 0,1\rangle$ & $\langle 0,-2\rangle$ & $\langle -1,-1\rangle$ & $\langle 2,2\rangle$ & $0$ & $0$ & $0$\\
$\chi_{48}$ & $-2$ & $\langle 0,1\rangle$ & $\langle 0,-2\rangle$ & $\langle -1,-1\rangle$ & $\langle 2,2\rangle$ & $0$ & $0$ & $0$\\
$\chi_{49}$ & $-2$ & $\langle -1,-1\rangle$ & $\langle 2,2\rangle$ & $\langle 0,1\rangle$ & $\langle 0,-2\rangle$ & $0$ & $0$ & $0$\\
$\chi_{50}$ & $\langle 2,2\rangle$ & $1$ & $-2$ & $1$ & $-2$ & $0$ & $0$ & $0$\\
$\chi_{51}$ & $\langle 0,-2\rangle$ & $1$ & $-2$ & $1$ & $-2$ & $0$ & $0$ & $0$\\
$\chi_{52}$ & $\langle 0,-2\rangle$ & $\langle 0,1\rangle$ & $\langle 0,-2\rangle$ & $\langle -1,-1\rangle$ & $\langle 2,2\rangle$ & $0$ & $0$ & $0$\\
$\chi_{53}$ & $\langle 2,2\rangle$ & $\langle -1,-1\rangle$ & $\langle 2,2\rangle$ & $\langle 0,1\rangle$ & $\langle 0,-2\rangle$ & $0$ & $0$ & $0$\\
$\chi_{54}$ & $0$ & $0$ & $0$ & $0$ & $0$ & $0$ & $0$ & $0$\\
$\chi_{55}$ & $0$ & $0$ & $0$ & $0$ & $0$ & $0$ & $0$ & $0$\\
$\chi_{56}$ & $0$ & $0$ & $0$ & $0$ & $0$ & $0$ & $0$ & $0$\\
$\chi_{57}$ & $0$ & $0$ & $0$ & $0$ & $0$ & $0$ & $0$ & $0$\\
$\chi_{58}$ & $0$ & $0$ & $0$ & $0$ & $0$ & $0$ & $0$ & $0$\\
\end{longtable}
\endgroup

\section*{AI disclosure}
During the preparation of this work, the authors used ChatGPT 6 Astra for language editing and for improving the exposition of the arguments.
Mathematical ideas and proof strategies were developed by the authors.

\section*{Acknowledgements}
X. Fan is supported by the National Natural Science Foundation of China (Grant No. 12561018). Y. Luan is supported by the Shandong Provincial Natural Science Foundation (Grant No. ZR2025MS36).

\bibliographystyle{amsplain}
\bibliography{FHL-2026}

\end{document}